\documentclass[prl, aps, reprint, twocolumn, superscriptaddress]{revtex4-2}
\usepackage{amsmath,amsthm,amssymb,epsfig,graphicx,mathrsfs,amsfonts,dsfont,bbm}
\usepackage[utf8]{inputenc}
\usepackage{dcolumn}
\usepackage{floatrow}
\usepackage[caption=false, justification=centerlast, label font=bf]{subfig}
\usepackage{bm}
\usepackage{tikz}
\usepackage{graphicx}
\usepackage{xcolor}
\usepackage[colorlinks=true,citecolor=magenta,urlcolor=cyan,linkcolor=black,filecolor=green]{hyperref}
\usepackage[capitalise]{cleveref}
\usepackage{crossreftools}
\usepackage{enumitem}
\setlist[itemize]{leftmargin=20pt}
\setlist[enumerate]{leftmargin=20pt}
\usepackage{amsmath,amsfonts,amssymb}
\usepackage{braket}
\usepackage{physics}
\usepackage[linesnumbered,ruled,vlined]{algorithm2e}
\usepackage{algpseudocode}
\usepackage{adjustbox}
\usepackage{lineno}
\usepackage{appendix}
\theoremstyle{plain}
\newtheorem{theorem}{Theorem}[section]
\newtheorem{corollary}{Corollary}[section]
\newtheorem{lemma}{Lemma}[section]
\newtheorem{proposition}{Proposition}[section]
\newtheorem{theoremprl}{Theorem}
\newtheorem{corollaryprl}{Corollary}

\newtheorem{propositionprl}{Proposition}
\newtheorem{remark}{Remark}

\theoremstyle{definition}
\newtheorem{definition}{Definition}

\newtheorem*{problem*}{Problem} 
\newtheorem{problem}{Problem}

\newcommand{\TS}{\textup{TS}}

\newcommand{\dpf}{\textup{\texttt{d-PF}}}
\newcommand{\dts}{\textup{\texttt{d-TS}}}
\newcommand{\dlcu}{\textup{\texttt{d-LCU}}}
\newcommand{\dqsp}{\textup{\texttt{d-QSP}}}
\newcommand{\dbe}{\textup{\texttt{d-BE}}}
\newcommand{\dro}{\textup{\texttt{d-RO}}}
\newcommand{\dqpe}{\textup{\texttt{d-QPE}}}
\newcommand{\dgrover}{\textup{\texttt{d-Grover}}}

\newcommand{\select}{\texttt{select}}
\newcommand{\supp}{\textup{supp}}

\newcommand{\parity}{\textsc{Parity}}
\newcommand{\InnerProduct}{\textsc{InnerProduct}}

\def\qdhs{\textup{\textsc{Qs}}}
\def\dqdhs{\textup{d-}\qdhs}

\newcommand{\IP}{\textup{\textsc{Ip}}}

\newcommand{\nff}{\textup{No fast-forwarding}}
\newcommand{\HC}{\mathrm{H.C.}}

\newcommand{\CCNOT}{\mathrm{Toffoli}}

\newcommand{\QC}{\textsf{\textup{Qcc}}}

\newcommand{\CC}{\textsf{\textup{Ccc}}}
\newcommand{\cc}{\mathrm{\bf c}}
\newcommand{\oo}{\mathrm{\bf o}}
\newcommand{\bigO}{\mathcal{O}}

\newcommand{\ii}{\textup{i}}
\newcommand{\talpha}{\tilde{\alpha}}

\newcommand{\eps}{\epsilon}

\newcommand{\poly}{\textup{poly}}
\newcommand{\even}{\textup{even}}
\newcommand{\odd}{\textup{odd}}
\newcommand{\intt}{\textup{int}}

\newcommand{\anc}{\mathrm{anc}}
\newcommand{\sys}{\mathrm{sys}}

\newcommand{\pfU}{\tilde{U}}
\newcommand{\calH}{\mathcal{H}}
\newcommand{\calV}{\mathcal{V}}

\newcommand{\calE}{\mathcal{E}}
\newcommand{\cmm}{\textup{comm}}
\newcommand{\unary}{\mathrm{unary}}

\newcommand{\predicate}{\textup{F}}
\newcommand{\vbx}{\vb{x}}
\newcommand{\vby}{\vb{y}}
\newcommand{\vbz}{\vb{z}}

\newcommand{\vertiii}[1]{{\left\vert\kern-0.25ex\left\vert\kern-0.25ex\left\vert #1
		\right\vert\kern-0.25ex\right\vert\kern-0.25ex\right\vert}}
\newcommand{\Vertiii}[1]{{\vert\kern-0.25ex\vert\kern-0.25ex\vert #1
		\vert\kern-0.25ex\vert\kern-0.25ex\vert}}
\newcommand{\floor}[1]{\lfloor{#1}\rfloor}

\newcommand{\teal}[1]{\textcolor{teal}{#1}}
\newcommand{\violet}[1]{\textcolor{violet}{#1}}
\newcommand{\olive}[1]{\textcolor{olive}{#1}}

\newcommand{\HKU}{QICI Quantum Information and Computation Initiative, Department of Computer Science, School of Computing and Data Science, The University of Hong Kong, Pokfulam Road, Hong Kong}
\newcommand{\NJU}{State Key Laboratory for Novel Software Technology, New Cornerstone Science Laboratory, Nanjing University, Nanjing 210023, China}
\newcommand{\HNL}{Hefei National Laboratory, Hefei 230088, China}
\newcommand{\FZU}{College of Computer and Data Science, Fuzhou University, Fuzhou 350108, China}

\usepackage[normalem]{ulem}
\makeatletter
\def\l@subsubsection#1#2{}
\makeatother

\def\onlysm{0}

\begin{document}
\let\oldacl\addcontentsline
\renewcommand{\addcontentsline}[3]{}

\title{Distributed Quantum Simulation}

\author{Tianfeng Feng}
\email{These authors contributed equally to this work.}
\affiliation{\HKU}

\author{Jue Xu}
\email{These authors contributed equally to this work.}
\affiliation{\HKU}

\author{Wenjun Yu}
\affiliation{\HKU}

\author{Zekun Ye}
\affiliation{\NJU}
\affiliation{\FZU}

\author{Penghui Yao}
\email{phyao1985@gmail.com}
\affiliation{\NJU}
\affiliation{\HNL}
  
\author{Qi Zhao}
\email{zhaoqi@cs.hku.hk}
\affiliation{\HKU}

\date{\today}
\begin{abstract}
        Quantum simulation is a promising pathway toward practical quantum advantage by simulating large-scale quantum systems.
    In this work, we propose communication-efficient distributed quantum simulation protocols by exploring three quantum simulation algorithms, including the product formula, the truncated Taylor series, and quantum signal processing over a quantum network. Crucially, our protocols significantly reduce the local hardware requirements of individual quantum processing units while ensuring the overall gate complexity does not exceed that of standard monolithic architectures.
    Our protocols are further shown to be optimal by
    deriving a lower bound on the quantum communication complexity for distributed quantum simulations with respect to evolution time and the number of distributed quantum processing units.
    Additionally, our distributed techniques go beyond quantum simulation and are applied to distributed versions of Grover's algorithm and quantum phase estimation. 
    Our work not only paves the way for achieving a practical quantum advantage by scalable quantum simulation but also enlightens the design of more general distributed architectures across various physical systems for quantum computation.
\end{abstract}

\maketitle

The pivotal goal of quantum computing is to significantly outperform classical computers in solving certain problems, namely quantum advantage \cite{preskillQuantumComputingEntanglement2012,aruteQuantumSupremacyUsing2019,zhongQuantumComputationalAdvantage2020}.
For this objective, a quantum computer with a substantial number of qubits is required.
However, a single quantum processing unit (QPU) faces numerous challenges when scaling up the number of qubits. 
To fully leverage the potential of quantum power, distributed quantum computing \cite{buhrmanDistributedQuantumComputing2003, bealsEfficientDistributedQuantum2013, ainleyMultipartiteEntanglementMultinode2024} that utilizes a network of multiple QPUs,
emerges as an essential solution \cite{kimbleQuantumInternet2008,elkinCanQuantumCommunication2014,wehnerQuantumInternetVision2018, legallQuantumAdvantageLOCAL2019,caleffiDistributedQuantumComputing2024,cacciapuotiMultipartiteEntanglementDistribution2024}.
So far, many applications of distributed quantum computing have been proposed, including distributed Shor's algorithm \cite{yimsiriwattanaDistributedQuantumComputing2004, jiangDistributedShorAlgorithm2023},
distributed machine learning \cite{tangCommunicationefficientQuantumAlgorithm2023, liBlindQuantumMachine2024}, etc \cite{liuDistributedQuantumPhase2021, tanDistributedQuantumAlgorithm2022, guoDistributedQuantumSensing2020,anshuDistributedQuantumInner2022, montanaroQuantumCommunicationComplexity2024, xuDistributedQuantumError2022}.
Besides, quantum simulation, Feynman's original motivation for proposing quantum computers \cite{feynmanSimulatingPhysicsComputers1982, feynmanQuantumMechanicalComputers1985, georgescuQuantumSimulation2014},   
is widely considered one of the most promising applications for achieving practical quantum advantage 
\cite{ciracGoalsOpportunitiesQuantum2012, bernienProbingManybodyDynamics2017, monroeProgrammableQuantumSimulations2021, altmanQuantumSimulatorsArchitectures2021, kimEvidenceUtilityQuantum2023, daleyPracticalQuantumAdvantage2022}.
Yet, a general distributed protocol for quantum simulation has not been proposed to date. 
Our work aims to fill this gap.


Quantum simulation of many-body systems generally involves implementing a dynamical evolution $e^{-\ii Ht}$ for a given Hamiltonian $H$.
There are two types of quantum algorithms for this task. 
One is the product formula (PF), also known as the Trotter-Suzuki method, 
which decomposes global dynamics into a product of easily implementable evolutions \cite{suzukiGeneralTheoryFractal1991, lloydUniversalQuantumSimulators1996, berryEfficientQuantumAlgorithms2007}. 
The other type is referred to as the post-Trotter method, which includes the truncated Taylor series (TS) method with the linear combination of unitaries
\cite{ berrySimulatingHamiltonianDynamics2015,longGeneralQuantumInterference2006,longDualityQuantumComputing2011,childsHamiltonianSimulationUsing2012} 
and quantum signal processing (QSP) method \cite{lowHamiltonianSimulationQubitization2019,lowOptimalHamiltonianSimulation2017,gilyenQuantumSingularValue2019, haahQuantumAlgorithmSimulating2018}.
In distributed scenarios, it is natural to ask how to utilize these methods to perform distributed quantum simulation, which involves executing quantum dynamics evolution with multiple QPUs.

In this work, we answer this question by proposing and systematically analyzing three concrete
distributed quantum simulation protocols: distributed PF ($\dpf$), distributed TS ($\dts$), and distributed QSP ($\dqsp$). 
These protocols are explored in the context of clustered Hamiltonians \cite{pengSimulatingLargeQuantum2020,harrowOptimalQuantumCircuit2025}, which depend on the partitioning of the network.
For $\dpf$, we formalize the relationship between quantum communication cost and the number of Hamiltonian terms involving interactions across the network.
For the distributed implementation of post-Trotter protocols ($\dts$ and $\dqsp$), 
the primary challenge lies in the distributed execution of multi-partite entangling quantum operations of interest.
We solve this problem by introducing two core subroutines: distributed block encoding and distributed reflection operation with the help of an ancillary node.


Meanwhile, we show the efficiency and the optimality of our distributed quantum simulation protocols.
While gate complexity is a major measure of quantum computation complexity, quantum communication complexity typically arises as a more stringent concern in distributed scenarios \cite{dewolfQuantumComputingCommunication2001, dewolfQuantumCommunicationComplexity2002, buhrmanNonlocalityCommunicationComplexity2010}.
This is because establishing quantum connectivity in distributed quantum processors is usually much more costly than elementary gates due to
the constraints of practical physical systems \cite{cacciapuotiQuantumInternetNetworking2020,jiangDistributedQuantumComputation2007,monroeLargescaleModularQuantumcomputer2014}. 
Regarding this, 
we prove that our protocols are communication-efficient for simulating generic many-body systems with general quantum networks by establishing a tight lower bound on the quantum communication complexity that depends on the number of parties involved and the evolution time of simulated systems. 
The $\dpf$ achieves optimal scaling in the number of parties, 
while the distributed post-Trotter protocols, i.e. $\dts$ and $\dqsp$, attain optimality concerning both the number of parties and the evolution time provided that the evolution time is shorter than the system size of a single party. Crucially, this hardware distribution generally incurs no overhead; the total gate complexity remains bounded by—and frequently improves upon—the monolithic one.

Additionally, we examine the quantum communication complexity associated with simulating $k$-local and nearest-neighbor Hamiltonian models. We also demonstrate that our distributed technique can go beyond quantum simulation and be applied to Grover's algorithm and quantum phase estimation. 

\begin{figure}[t]
    \center
    \includegraphics[width=0.96\linewidth]{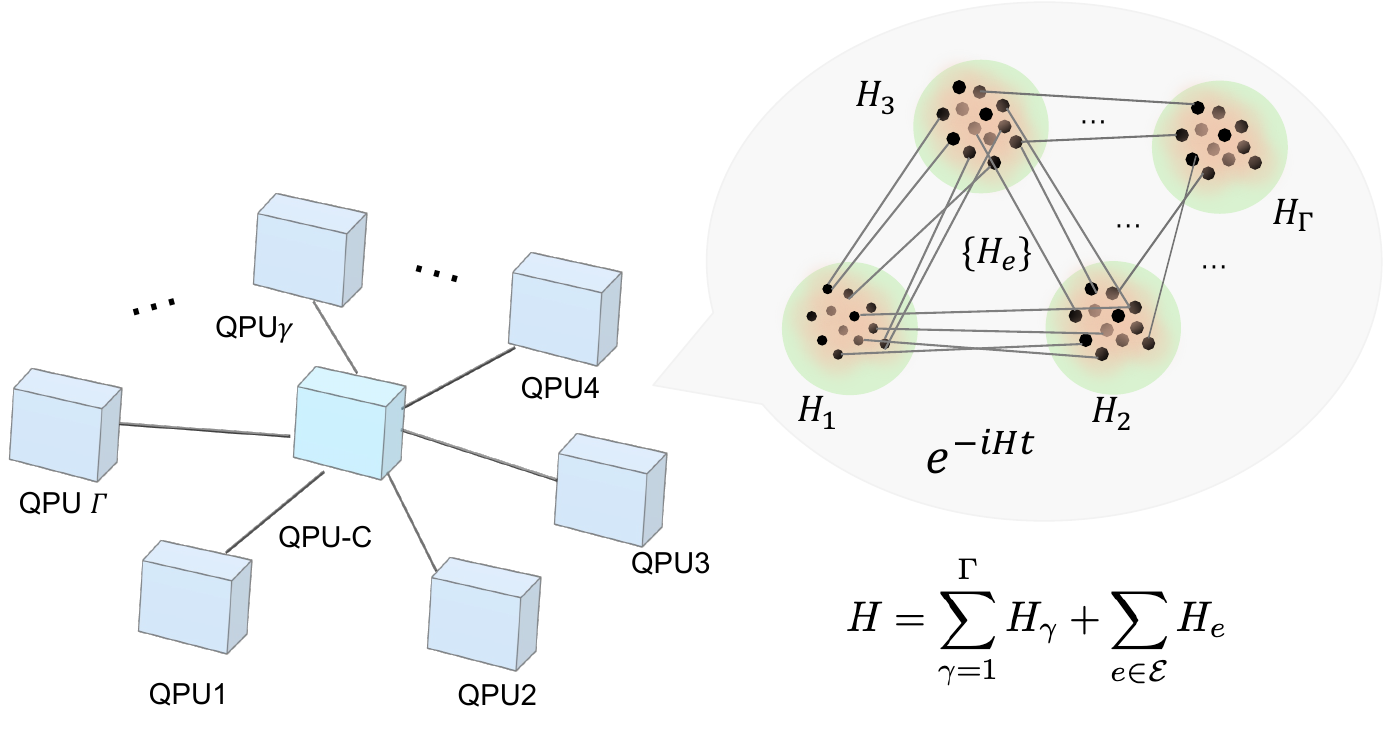}
    \caption{
    The star quantum network consists of one control node connecting $\Gamma$ distributed QPUs.  
    The distributed quantum simulation is to implement $e^{-\ii H t}$ with a clustered Hamiltonian \cref{eq:cluster_hamiltonian} induced by the network.
    $H_{\gamma}$ is the local term on $\gamma$th node (QPU), while $H_e$ across nodes are called interaction terms.
    }
    \label{fig:clustered}
\end{figure}

\emph{Setup.}--
In this work, distributed quantum simulation is defined as: 
given a $\Gamma n$-qubit Hamiltonian $H=\sum_{l=1}^L H_l$  with $L=\poly(\Gamma n)$ 
and evolution time $t$, the goal is to implement the real-time evolution operator $e^{-\ii Ht}$ with $\Gamma$ distributed QPUs.
Each QPU owns $n$ qubits and allows $o(n)$ ancillary qubits for quantum communication or information processing, etc.
Without loss of generality, we utilize the star network \cite{branciardCharacterizingNonlocalCorrelations2010, tavakoliNonlocalCorrelationsStarnetwork2014, mengConcurrencePercolationQuantum2021} as \cref{fig:clustered} 
to analyze the distributed protocols for quantum simulation \footnote{At the end of the paper, we will show that the quantum communication complexity of general distributed protocols is independent of the specific topology of the quantum network}. 
 
To implement distributed quantum simulation, the partition of qubits for a quantum network has to be considered.
Analogous to Refs.~\cite{pengSimulatingLargeQuantum2020, loweFastQuantumCircuit2023, childsTheoryTrotterError2021},  in a quantum network,
   the original Hamiltonian $H=\sum_{l=1}^L H_l$ can be massaged into two parts, i.e., 
    \begin{equation}\label{eq:cluster_hamiltonian}
        H=\sum_{\gamma=1}^\Gamma H_\gamma+\sum_{e\in\calE} H_e,
    \end{equation}
which is called the \emph{induced clustered Hamiltonian}. 
The first summand contains the \emph{induced local} terms $H_\gamma:=\sum_{l:\supp(H_l)\subseteq v_\gamma} H_l$, where $v_\gamma$ is the set of qubits on the $\gamma$th node. 
The second summand contains all remaining terms called \emph{induced interaction} terms that act on multiple nodes \footnote{In the following, we demonstrate how clustered Hamiltonians can aid in various distributed quantum simulation protocols. }.
For simplicity, we will omit the word ``induced'' in the rest of the paper.


\begin{figure*}[tb]
    \center
    \includegraphics[width=0.96\linewidth]{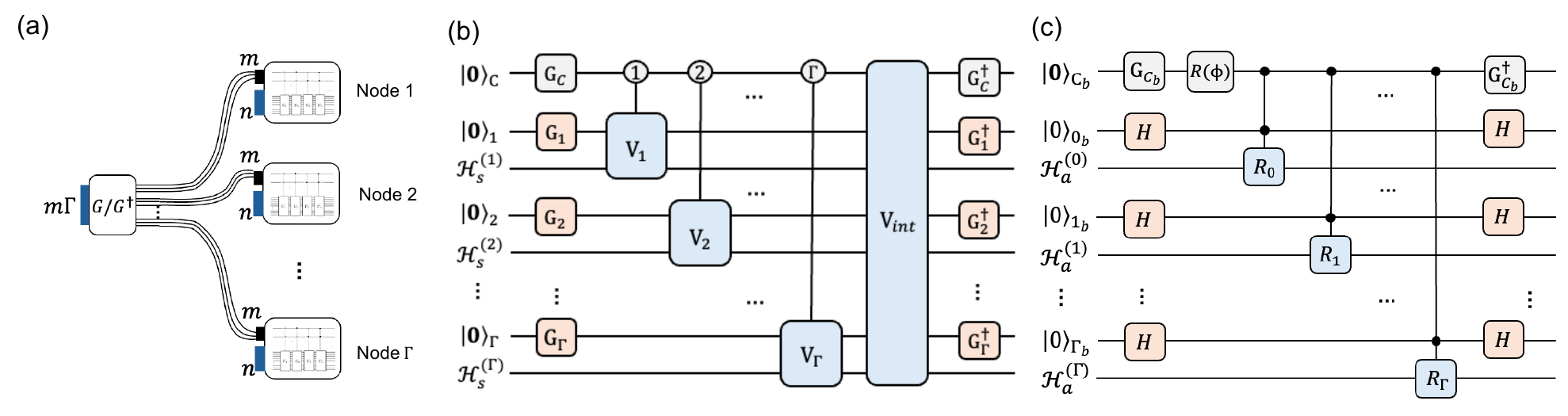}
    \caption{
    (a) General distributed architecture for $\dbe$, $\dro$ and various algorithms. 
    This distributed architecture utilizes the star network of the control (ancillary) node and $\Gamma$ local nodes.  Each node has $n$ qubits assisted with $m=\bigO(\log (\Gamma n))$ ancillary qubits, while the control node has $m\Gamma$ qubits for quantum information processing. 
    The general distributed protocols employ the control node to prepare and distribute desired entangled states to each node, which in turn performs global quantum operations and information processing on $\Gamma$ nodes. 
    Similarly, after conducting the appropriate local operations, each node may return the corresponding auxiliary qubits to the control node for quantum information processing.
    (b) Quantum circuit 
    for distributed block encoding ($\dbe$); (c) Quantum circuit for distributed reflection operation ($\dro$). 
    In this context, $\mathcal{H}^{\gamma}_{a}$ refers to the Hilbert space corresponding to the auxiliary system in $\dbe$.  
    For simplicity, we designate the control node as node 0, i.e., $\gamma=0$. 
    }
    \label{fig:alldown}
\end{figure*}

\emph{Distributed product formula (d-PF).}--\label{sec:d_pf}
The product formula (PF) method (also known as the Trotter-Suzuki formula) is deemed the most straightforward way for quantum simulation.
Given a local Hamiltonian $H=\sum_l^L H_l$ and time $t$, 
the first-order PF approximates  $U=e^{-\ii Ht}$ by the product of the exponentials of the individual operators, i.e.
$\pfU(t):=\qty(\prod\nolimits_{l=1}^{L}e^{-\ii H_lt/r})^r$
where each exponential can be implemented by elementary quantum gates efficiently \cite{lloydUniversalQuantumSimulators1996}.
By the $p$th-order PF,
the segments $r$ to achieve Trotter $\epsilon$ can be further improved to 
$\bigO(\alpha_{\cmm,p}^{1/p} t^{1+1/p}/\epsilon^{1/p})$
where $\alpha_{\cmm,p}$ is the nested commutator norm \cite{childsTheoryTrotterError2021}.
In distributed scenarios,
with the clustered Hamiltonian as described in \cref{eq:cluster_hamiltonian}, 
the procedure for our $\dpf$ protocol is as follows. 
Similar to the partition in \cite{childsTheoryTrotterError2021,buessenSimulatingTimeEvolution2023}, 
for a single segment, each node ($\gamma$) locally performs the evolution $e^{-\ii H_\gamma t/r}$. 
Then, the evolution of the interaction term $e^{-\ii H_e t/r}$ is implemented through quantum communication, assisted by the control node as illustrated in \cref{fig:clustered}, since the support of $H_e$ spans at least two nodes. 
Consequently, the quantum communication complexity of the $\dpf$ protocol is directly provided by the following theorem.
\begin{theoremprl}[Quantum communication complexity of $\dpf$]\label{thm:d_pf}
    Given a $\Gamma n$-qubit clustered Hamiltonian $H=H_0 +\sum_{e=1}^{\abs{\calE}} H_e $ with $H_0:=\sum_{\gamma=1}^\Gamma H_\gamma$ and evolution time $t\in\mathbb{R}^+$,
    quantum communication complexity of the $p$th-order distributed product formula (\dpf) to simulate $e^{-\ii H t}$ within accuracy $\eps$ is
    \begin{equation}
        \bigO(\abs{\calE} \Gamma \talpha_{\mathrm{qc},p}^{1/p} \ t^{1+1/p}/\eps^{1/p}), 
    \end{equation}
    where $\tilde{\alpha}_{\mathrm{qc},p}:=\sum\limits_{e_1,\dots,e_{p+1}=0}^{\abs{\calE}} \norm{[H_{e_1},\dots,[H_{e_p}, H_{e_{p+1}}]]}$ is the norm of the nested commutator of terms in $H$ with $H_{e=0}=H_0$, $\Gamma$ is the number of nodes, and $\abs{\calE}$ is the number of interaction terms across multiple nodes. 
\end{theoremprl}

See the proof of \cref{thm:d_pf} in \cite{smDQS}. We also extend $\dpf $ for observable simulation (see \cref{prop:observable_error} in the End Matter), which suggests that local observables may reduce communication overhead for short-time evolution.
The $\dpf$ protocol can be extended to the analog quantum simulation \cite{ farhiAnalogAnalogueDigital1998, cubittUniversalQuantumHamiltonians2018, liuEfficientlyVerifiableQuantum2025} with a similar procedure. 

\emph{Distributed post-Trotter protocols.}--
In contrast to the Trotter algorithms approximating $ e^{-\ii Ht}$ by products, the evolution operator can also be expressed as a summation, namely $ e^{-\ii Ht}=\sum_k g_k f_k(H)$,
where $g_k$ is the coefficient and $f_k(H)$ is the $k$-th function of $H$.
Two typical examples are the truncated Taylor series (TS) method by the Taylor expansion and the quantum signal processing (QSP) method by the Jacobi-Anger expansion. 
Generally, these post-Trotter algorithms need block encoding (BE) of Hamiltonians in a larger Hilbert space. The block encoding of an operator $H$ can be regarded as a unitary $U$ with ancillary system $a$ such that $(\bra{0}_a \otimes I_s) U (\ket{0}_a \otimes I_s) = H/\alpha$ where $\alpha$ is a positive number \cite{lowHamiltonianSimulationQubitization2019}.
Typically, the linear combination of unitaries (LCU) is the standard way to implement BE of $H$ \cite{longGeneralQuantumInterference2006,longDualityQuantumComputing2011,childsHamiltonianSimulationUsing2012,lowHamiltonianSimulationQubitization2019}.  
Nonetheless, LCU incurs a probability issue that entails the oblivious amplitude amplification \cite{berryHamiltonianSimulationNearly2015}, necessitating 
the reflection operation (RO). Both BE and RO serve as core subroutines in the construction of various algorithms \cite{gilyenQuantumSingularValue2019}.
In a distributed quantum simulation scenario, the distributed implementation of BE and RO should be thoroughly considered. We analyze these implementations based on the distributed architecture depicted in \cref{fig:alldown} (a). 

\emph{Distributed Block encoding (d-BE).}--
We present a distributed quantum circuit to implement $\dbe$ with low quantum communication cost based on a clustered Hamiltonian $H=\sum_{\gamma=1}^\Gamma H_\gamma+\sum_{e\in\calE} H_e$ with $\calE \ne \emptyset$.
Recall that $H_\gamma$ is a local term supported on the $ \ gamma$-th node, which in principle can be block encoded locally without communication.
Contrarily, the interaction terms $H_e$ that cross different nodes 
require multiple nodes to implement BE together, and their coefficients $\{\norm{H_e}\}$ should be shared among the involved nodes.
To coherently control all the terms of $H$ for BE (including all $H_\gamma$ and $H_e$) in a quantum network, 
the auxiliary system should carry the information of the clustered Hamiltonian.
To maintain the nice property that no quantum communication is required for BE of local terms $H_\gamma$, we are supposed to design a novel $\dbe$ protocol.

Our $\dbe$ circuit entails $\Gamma+1$ nodes, 
including one control node and $\Gamma$ local nodes as in \cref{fig:alldown} (b). 
Specifically, the control node is employed to prepare 
\begin{equation}
   G_C\ket{\mathbf{0}}_C=
   \frac{1}{\sqrt{\alpha}}
   \qty(\sum_{\gamma=1}^{\Gamma} \sqrt{\norm{H_\gamma}_1} \ket{\gamma}^{\otimes \Gamma}_C+\sum_{e=1}^{\abs{\calE}}\sqrt{\norm{H_e}}\ket{e}^{\otimes \Gamma}_C), 
\end{equation}
where $\alpha=\norm{H}_1=\sum_{\gamma=1}^{\Gamma} \norm{H_\gamma}_1 +\sum_{e=1}^{\abs{\calE}}\norm{H_e}$.
Here, $\norm{\cdot}$ denotes the operator norm, 
and $\norm{\cdot}_1$ denotes the 1-norm 
(i.e., the sum of absolute values of Pauli coefficients).
Note that $\ket{i}^{\otimes \Gamma}$ is used instead of $\ket{i}$ because, in our distributed protocol, we subsequently have to distribute the prepared auxiliary quantum states to each node. The advantage of using the basis $\ket{i}^{\otimes \Gamma}$ is that it allows the simultaneous transmission of individual qubits to each node, all of which share the same amplitude.

The other $\Gamma$ ancillary registers are distributed across $\Gamma$ nodes, each of which is used to perform the BE of the local $H_{\gamma}$. 
Specifically, $\gamma$-th ancillary register prepares
$G_\gamma\ket{\mathbf{0}}_\gamma=\frac{1}{\sqrt{\alpha_{\gamma}}}\sum_{l\in v_\gamma}\sqrt{\norm{H_l} }\ket{l}_\gamma$,
where $\alpha_\gamma=\norm{H_\gamma}_1$ and $v_\gamma$ denotes the set of qubits on the $\gamma$th processor.  
Similarly, we define the effective select operator as 
$\select(U)=\sum_{\gamma=1}^{\Gamma} |\gamma\rangle^{\otimes \Gamma} \langle \gamma|^{\otimes \Gamma}_C \otimes V_\gamma +V_{\intt}$ 
where $V_\gamma=\sum_{l\in v_\gamma}\op{l}_\gamma\otimes \frac{H_l}{\norm{H_l}}$ 
and $V_{\intt}=\sum_{e=1}^{\abs{\calE}}|e\rangle^{\otimes \Gamma} \langle e|^{\otimes \Gamma}_C \otimes \frac{H_e}{\norm{H_e}}$ 
\footnote{It should be noted that we have projected the ancillary system onto its encoding basis $\{\ket{0}^{\otimes \Gamma},\ket{1}^{\otimes \Gamma},...,\ket{i}^{\otimes \Gamma},...\}$.}
such that the $\dbe$ of $H$ reads 
\begin{equation}
        \left( \langle \mathbf{0}| G_{C}^{\dagger}\bigotimes_{\gamma =1}^{\Gamma}G_{\gamma}^{\dagger} \right) 
    \select(U) 
    \left( G_C \bigotimes_{\gamma =1}^{\Gamma}{G_{\gamma}|\mathbf{0}\rangle}\right) =\frac{H}{\alpha },
\end{equation}
where $\ket{\mathbf{0}}=\ket{\mathbf{0}}_C\bigotimes_{\gamma =1}^{\Gamma}{|\mathbf{0}\rangle}_{\gamma}$.
It can be seen that the $\dbe$ has the same denominator $\alpha$ as the standard BE \cite{lowHamiltonianSimulationQubitization2019}.

The quantum communication cost of $\dbe$ stems from the process of distributing the quantum state $ G_C\ket{\mathbf{0}}$, which necessitates the transmission of $\Gamma \log(\abs{\calE}+\Gamma)$ qubits to $\Gamma$ nodes.
After executing $\select(U)$, the control node needs to perform the $G^{\dagger}_C$ operation,  
requiring each node to resend the corresponding $ \log(\abs{\calE}+\Gamma)$ qubits back to the control node. 
Thus, the overall quantum communication complexity of $\dbe$ is $\bigO(\Gamma \log(\abs{\calE}+\Gamma))$. 
It should be emphasized that if we employ the original decomposition of the Hamiltonian $H=\sum^L_l H_l$ (where $L$ may be significantly larger than $\abs{\calE}+\Gamma$), 
the quantum communication complexity of the $\dbe$ based on it would be $\bigO(\Gamma \log L)$ which would be unfavorable. 

Remarkably, we detail a simple modified d-BE construction for bipartite networks in the Supplementary Materials \cite{smDQS}. In the End Matter, we present the upper bound and establish a fundamental information-theoretic communication lower bound, rigorously proving that our bipartite construction achieves optimal communication complexity.


\emph{Distributed reflection operation (d-RO).} 
The reflection operation (RO) is another core element of post-Trotter protocols.
The RO is defined as $R:=I-(1-e^{-\ii\phi})\op{\mathbf{0}}$, where $\op{\mathbf{0}}:=\op{\mathbf{0}}_{C} \otimes\bigotimes_{\gamma=1}^{\Gamma}  \op{\mathbf{0}}_\gamma$ in our setting (for the all ancillary systems in $\dbe$). 
For notational simplicity, we designate the control node as node 0.
In this way, the RO can be rewritten as $R=I-(1-e^{-\ii\phi})\bigotimes_{\gamma=0}^{\Gamma} \op{\mathbf{0}}_\gamma$.
Next, we examine how to perform the distributed reflection operation (denoted $\dro$) in a quantum network. 
First, we introduce the local reflection operation of the $\gamma$-th node
as $R_\gamma= 2 \op{\mathbf{0}}_\gamma-I_{\gamma}$.
Then, one may rewrite the operator as \cite{smDQS} 
\begin{equation}\label{eq:reflection}
    R =I-(1-e^{-\ii\phi})  \bigotimes_{\gamma=0}^{\Gamma} \qty(\frac{R_\gamma +I_\gamma}{2}).
\end{equation}
Since $R$ is unitary and the above decomposition suggests that $R$ can be rewritten as a linear combination of unitary operations, we may introduce new auxiliary systems to implement $\dro$ via the nested LCU technique. 
In this setting, the auxiliary system $\{C_b,0_b,1_b,...,\Gamma_b\}$ of the control node and each node for $\dro$ is only 2-dimensional. We present the quantum circuit of $\dro$ in \cref{fig:alldown} (c).
We prove that the quantum communication complexity of implementing the $\dro$ \cref{eq:reflection} is  $\bigO(\Gamma)$ 
\footnote{Remarkably, our scheme can be generalized to the distributed phase operation of any component of interest, i.e. $I-(1-e^{-\ii\phi})\op{j}$.
Since RO also serves as the key operation in Grover's algorithm, we also present a distributed version of it in the End Matter.} 



\emph{Communication complexity of post-Trotter protocols.}--
With the above communication-efficient $\dbe$ and $\dro$ designs, the quantum communication complexities of $\dts$ and $\dqsp$ are obtained, respectively, based on the quantum query complexity of these algorithms \cite{smDQS}.

\begin{theoremprl}[Quantum communication complexity of $\dts$]\label{thm:d_ts}
    Given a $\Gamma n$-qubit Hamiltonian 
    $H=\sum_{\gamma=1}^\Gamma H_\gamma+\sum_{e\in\calE} H_e$ with $\calE \ne \emptyset$
    and evolution time $t$,
    the quantum communication complexity of the distributed TS protocol (\dts) for quantum simulation with accuracy $\eps$
    is 
   \begin{equation}
       \bigO\qty( \alpha t \Gamma \log (\abs{\calE}+\Gamma) K),
    \end{equation}
    where $\Gamma$ is the number of nodes, $\abs{\calE}$ is the number of the interaction terms, $K= \frac{\log(\alpha t/\epsilon)}{\log\log(\alpha t/\epsilon)} $ is the truncated order of Taylor series
     and $\alpha= \norm{H}_1$.
\end{theoremprl}

\begin{theoremprl}[Quantum communication complexity of \dqsp]\label{thm:d_qsp}
    Given a $\Gamma n$-qubit Hamiltonian $H=\sum_{\gamma=1}^\Gamma H_\gamma+\sum_{e\in\calE} H_e$ with $\calE \ne \emptyset$ and evolution time $t$,
    the quantum communication complexity of the distributed QSP protocol (\dqsp) for quantum simulation with accuracy $\eps$ is
    \begin{equation}
        \bigO (\Gamma \log( \abs{\calE} + \Gamma)[\alpha t+ \log(1/\epsilon)]),
    \end{equation}
    where  $\Gamma$ is the number of nodes,  $\abs{\calE}$ is the number of the interaction edge terms, 
    and $\alpha=\norm{H}_1$. 
\end{theoremprl}
It can be seen
that the parameters of the interaction graph (i.e., $\Gamma$ and $\abs{\calE}$) crucially determine the quantum communication complexity of these distributed simulation protocols \footnote{Note that if $\abs{\calE}=0$, this suggests that there are no interaction terms crossing over a quantum network, and the communication cost should be $0$ since one can implement the quantum simulation locally. In this case,  distributed quantum simulation protocols including $\dts$ and $\dqsp$ are no longer necessary.}.  Table \ref{tab:compact_complexity} in End Matter summarizes the resource trade-offs between monolithic and distributed architectures.

\emph{Lower bound on quantum communication.}--
In the context of standard (non-distributed) quantum simulation, Berry et al. \cite{berryEfficientQuantumAlgorithms2007} and subsequent works \cite{berryHamiltonianSimulationNearly2015, atiaFastforwardingHamiltoniansExponentially2017, haahQuantumAlgorithmSimulating2018} established a lower bound on the quantum complexity, known as the No fast-forwarding theorem. 
It states that the gate complexity of any generic quantum (dynamics) simulation algorithm scales at least linearly in scaled evolution time $\norm{H}t$, 
which establishes the optimality of the state-of-the-art quantum simulation algorithms. 
In a similar manner, we give a lower bound on the communication complexity of the distributed quantum simulation.
\begin{theoremprl}[Informal, Optimality]\label{thm:weak_lower_bound}
    For any positive integer $\tau\le n$, 
    there exists a sparse Hamiltonian $H$ acting on $2n+\log(n)+1$ qubits with $\norm{H}=\tau$ and a constant evolution time $t=\pi/2$ such that 
    the (bounded-error) quantum communication complexity ($\QC_2$) of any generic $\Gamma$-partite 
    protocol for distributed quantum simulation (\dqdhs) scales at least linearly in scaled evolution time, that is, $\QC_2(\dqdhs)= \Omega(\Gamma\norm{H}t)$.
\end{theoremprl}

The proof of this lower bound is by reduction from the \InnerProduct{} problem 
to the distributed Hamiltonian simulation.
It is known that $\Omega(n)$ quantum communication is required to evaluate the inner product $\IP_n(\vby,\vbz)$ of two $n$-bit strings $(\vby,\vbz)$ owned by two parties
\cite{kremerQuantumCommunication1995, dewolfQuantumComputingCommunication2001}.
We can construct the quantum circuit consisting of Toffoli gates as \cref{ip} to evaluate the inner product of any input bit-strings encoded in the bipartite initial state $\ket{\vby}\ket{\vbz}$.
Then, this circuit can be encoded into a sparse Hamiltonian dynamics $H_{\IP}$ such that the evolution $e^{-\ii H_{\IP} t} \ket{\vby}_{A}\ket{\vbz}_{B} \ket{0}_{\oo}$ would yield the inner product in the output register of the final state for $t=\pi/2$ and $\norm{H}=\Theta(n)$.
Therefore, any distributed quantum simulation of $H_{\IP}$ need $\Omega(\norm{H}t)$ communication.
And this bound can be generalized to the $\Gamma$-partite network.

{
For simulating more physically motivated $k$-local Hamiltonians,
we also establish an $\Omega(t)$ quantum communication lower bound 
using an entanglement capacity argument with a physical Hamiltonian. 
With the same construction and the unbounded-error communication lower bound of \InnerProduct{},
we give an important communication lower bound in terms of $\eps$  as $\Omega(\log(1/\epsilon)/\log\log(1/\epsilon))$,
which matches the upper bound by our post-Trotter protocols. 
We defer the proof sketch of these lower bounds to End Matter and complete proof in \cite{smDQS}.
}

\begin{figure}[t]
    \centering
    \includegraphics[width=.95\linewidth]{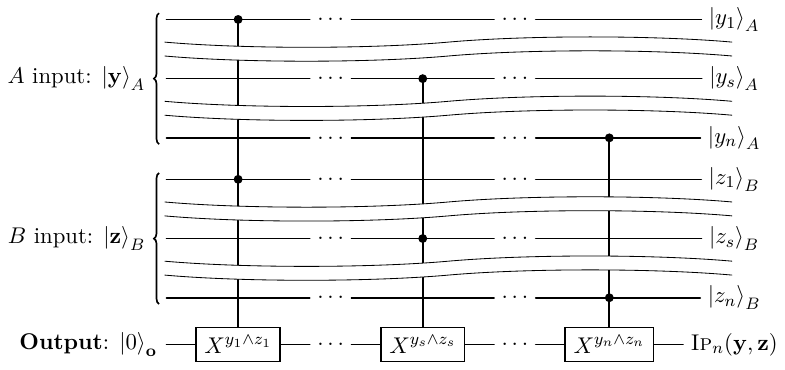}
    \caption{
    The construction from an $\IP_{n}(\vb{y},\vb{z})$ circuit (instance) to a Hamiltonian evolution with $H_{\IP}$ and initial state $\ket{\vb{0}}_{\cc}\otimes\ket{\vb{y}}_A\otimes\ket{\vb{z}}_B\otimes \ket{0}_{\oo}$.
    Alice (Bob) encodes an $n$-bit string $\vby$ ($\vbz$) with $n$-qubit register, and the output register consists of one qubit to store the inner product of the two bit strings.
    }
    \label{ip}
\end{figure}

\emph{Network topology independence.}-- 
For generality, we argue that the quantum communication complexity of our distributed protocols (except the \dpf) is independent of the specific network structure \cite{mengConcurrencePercolationQuantum2021}.
Intuitively, the total quantum communication cost for establishing and disentangling a globally entangled auxiliary state $\sum_{j=1}^J \sqrt{\beta_j} \ket{j}^{(1)}\ket{j}^{(2)}...\ket{j}^{(\Gamma)}$ scales linearly with the number of nodes, i.e., \(\mathcal{O}(\Gamma\log_2J )\), and is independent of the network's specific topology.
Ultimately, networks of arbitrary structure can be virtually equivalent to a star network presented in \cref{fig:clustered} with the same quantum communication overhead for our distributed protocols. 
See more discussion in \cite{smDQS}.

\emph{Applications.}-- Now we compare the concrete performance of our protocols for simulating common physical Hamiltonians.
Consider simulating $k$-local Hamiltonians $H=\sum_{l_1,\dots,l_k} H_{l_1,\dots,l_k}$, 
where $H_{l_1,\dots,l_k}$ acts non-trivially on at most $k$ qubits. 
We show that the quantum communication complexity of the $p$th-order $\dpf$ is $\bigO\qty(\min(k,\Gamma)(\Gamma n)^k\vertiii{H}_1\alpha^{1/p}  t^{1+1/p}/\eps^{1/p})$ 
where $\vertiii{H}_1$ is the induced $1$-norm \cite{childsTheoryTrotterError2021}.
It is proportional to the $k$th power of the number of qubits in the simulated system, i.e., $(\Gamma n)^k$. 
In contrast, the post-Trotter protocols ($\dts$ and $\dqsp$) have complexity 
$\bigO\qty(\alpha  t \Gamma k \log(\Gamma n) \frac{\log(\alpha t/\epsilon)}{\log\log(\alpha t/\epsilon)})$ 
and $\bigO(k\Gamma\log(\Gamma n)[\alpha t+\log(1/\epsilon)])$ respectively which are proportional to $k \log(\Gamma n)$ \cite{smDQS}.

For a $\Gamma n$-qubit one-dimensional lattice Hamiltonian with the nearest-neighbor (NN) interaction $H_{\mathrm{NN}}=\sum_{j=1}^{\Gamma n-1}H_{j,j+1}$, the number of interaction terms across terms is proportional to the number of nodes, i.e. $\abs{\calE}=\bigO(\Gamma)$.
The quantum communication complexity of the $p$th-order $\dpf$, the $\dts$, and the $\dqsp$ are $\bigO(\Gamma ^{1+1/p}t^{1+1/p}/\eps^{1/p})$,  $\bigO\qty( \alpha t \Gamma \log \Gamma  \frac{\log(\alpha t/\epsilon)}{\log\log(\alpha t/\epsilon)})$ and $\bigO\qty(  \Gamma \log \Gamma [\alpha t+\log (1/\epsilon)])$, respectively \footnote{In this case, the quantum communication complexity for $\dpf$ is more friendly than distributed post-Trotter methods. In supplementary materials, we numerically compare the communication cost of the protocols for these two typical Hamiltonians.}.
{Besides, beyond quantum simulation, we also show how to extend our distributed post-Trotter protocols to a distributed setting for quantum phase estimation and Grover's algorithm (see End Matter).}



\emph{Discussion.}--
This work takes a step forward in achieving practical quantum advantage by proposing efficient distributed quantum simulation protocols. 
Previous research has focused on decomposing large-scale quantum computation tasks into smaller partitions using local operations and classical communication only \cite{pengSimulatingLargeQuantum2020,yuanQuantumSimulationHybrid2021,sunPerturbativeQuantumSimulation2022,loweFastQuantumCircuit2023,piveteauCircuitKnittingClassical2024,harrowOptimalQuantumCircuit2025,haradaDoublyOptimalParallel2024}, typically with exponentially growing complexity. In contrast, our proposed distributed quantum simulation protocols with quantum connectivity solve this problem.
We further establish the optimality of the quantum communication complexity of our protocol for the number of partitions and time when $t\le n$. However, achieving optimality for longer times remains challenging.


On the other hand, we analyze in detail the impact of noisy channels on distributed quantum simulation in \cite{smDQS}.
Meanwhile, there is potential to further decrease the quantum communication complexity.
For example, a more fine-grained partition for a specific Hamiltonian in a certain quantum network may have a lower overhead in quantum communication complexity. 
Additionally, with the help of variants of quantum simulation algorithms \cite{ wangFasterQuantumAlgorithms2024, childsFasterQuantumSimulation2019, campbellRandomCompilerFast2019, zhaoExploitingAnticommutationHamiltonian2021} and better performance analyses utilizing prior knowledge 
\cite{sahinogluHamiltonianSimulationLowenergy2021, gongComplexityDigitalQuantum2024,heylQuantumLocalizationBounds2019, yuObservableDrivenSpeedupsQuantum2025,zhaoHamiltonianSimulationRandom2021, chenAverageCaseSpeedupProduct2024,zhaoEntanglementAcceleratesQuantum2025,feng2025trotterizationoperatorscramblingentanglement}, quantum communication complexity for distributed quantum simulation may potentially be further reduced.

\begin{acknowledgments}
    We thank X. Zhang and J. Sun for their helpful comments and suggestions. T. F., J.X., W.Y., and Q.Z. acknowledge funding from National Natural Science Foundation of China (NSFC) via Project No. 12347104 and No. 12305030, Guangdong Natural Science Fund via Project 2023A1515012185, Hong Kong Research Grant Council (RGC) via No. 27300823, N\_HKU718/23, and R6010-23, Guangdong Provincial Quantum Science Strategic Initiative No. GDZX2303007, HKU Seed Fund for Basic Research for New Staff via Project 2201100596.  Z.Y. and P.Y. were supported by the National Natural Science Foundation of China (Grant No. 62332009, 12347104), the Innovation Program for Quantum Science and Technology (Grant No. 2021ZD0302901), NSFC/RGC Joint Research Scheme (Grant No. 12461160276) and Natural Science Foundation of Jiangsu Province (No. BK20243060).
\end{acknowledgments}
\bibliographystyle{truncate.bst}
\bibliography{ref_aps,refer}

\section*{End Matter}

\subsection{Total resource estimation}

{Table \ref{tab:compact_complexity} summarizes the resource trade-offs between monolithic and distributed architectures. The results demonstrate that quantum communication enables the expansion of system size without requiring a large-scale local QPU and increasing the overall gate complexity (see \cite{smDQS} for more details).
}

\begin{table}[t!]
\caption{Resource trade-offs in distributed quantum simulation. By employing quantum communication (Comm.), the distributed architectures reduce the local size of QPU from $\Gamma n$ to $n$. 
Crucially, this modularity is achieved while guaranteeing that the total number of gates is generally bounded by the non-distributed or monolithic one. Here, $\Lambda$ denotes the prefactors for post-Trotter algorithms ($\Lambda_{\text{TS}} = \alpha t K$, $\Lambda_{\text{QSP}} = \alpha t + \log(1/\epsilon)$), and $\abs{\calE}' = \abs{\calE}+\Gamma$ while $r_{\text{global}}$ and $r_{\mathrm{qc}}$ are the segment of monolithic and the distributed methods.}
\label{tab:compact_complexity}
\renewcommand{\arraystretch}{1.4}
\begin{tabular*}{\columnwidth}{@{\extracolsep{\fill}}lccc@{}}
\hline\hline
\textbf{Methods} & \textbf{Local qubits } & \textbf{Total gates} & \textbf{Comm.} \\
\hline
\multicolumn{4}{l}{\textit{Product Formula (PF)}} \\
Monolithic & $\Gamma n$ & $\mathcal{O}(L r_{\text{global}})$ & $0$ \\
Distributed & $n$ & $\lesssim \text{Mono.}$ & $\mathcal{O}(\Gamma \abs{\calE} r_{\text{qc}})$ \\
\hline
\multicolumn{4}{l}{\textit{Post-Trotter Methods (TS \& QSP)}} \\
Monolithic & $\Gamma n + \log L$ & $\mathcal{O}\big(\Lambda [L(\log L + \Gamma n)]\big)$ & $0$ \\
Distributed & $n + \log \abs{\calE}'$ & $\lesssim \text{Mono.}$ & $\mathcal{O}(\Lambda \Gamma \log \abs{\calE}')$ \\
\hline\hline
\end{tabular*}
\end{table}

\subsection{Lower bounds on communication complexity}
We start the proof of \cref{thm:weak_lower_bound} from the bipartite case.
The proof is by reducing from the $\InnerProduct$ problem to the distributed quantum simulation (\dqdhs).
The distributed version of $\parity$, that is, $\InnerProduct$ is a Boolean function, denoted $\IP_n(\vbx, \vby):\qty{0,1}^n\times \qty{0,1}^n \to \qty{0,1}$ which evaluates the inner product of two bit-strings. 
Refs.~\cite{kremerQuantumCommunication1995, dewolfQuantumComputingCommunication2001} proved that quantum communication complexity $\QC_2$ of evaluating $\IP_n$ is $\QC_2(\IP_n)=\Omega(n)$. 
If two parties could simulate quantum dynamics distributively with $o(\norm{H}t)$ quantum communication, 
then the Boolean function $\IP_n$ can be evaluated with $o(n)$ communication, which would contradict the known lower bound.
Specifically, one can write down a quantum circuit consisting of $n$ Toffoli gates to evaluate the inner product of two bit-strings owned by two parties (see \cref{ip}).
Correspondingly, a Hamiltonian evolution can be constructed from this circuit such that the evolution can yield the circuit output and the lower bound $\Omega(\norm{H}t)$ is obtained.
For a $\Gamma$-partite network, Ref.~\cite{legallBoundsObliviousMultiparty2022} proved that the lower bound of $\Gamma$-partite $\InnerProduct$ can be extended to $\Omega(\Gamma n)$. 
So, our proof can be directly extended to $\Gamma$-partite case by replacing the Toffoli gates with $\Gamma$-partite controlled-NOT gates and the lower bound becomes $\Omega(\Gamma \norm{H}t) $ \cite{smDQS}.
Since the quantum communication complexity of our post-Trotter protocol is linear in $t$ and $\Gamma$,
this theorem suggests that our distributed quantum simulation algorithms achieve optimal dependence on evolution time $t$ and the number of partitions $\Gamma$ provided $\norm{H}t\le n$.

{
The above lower bounds rely on the sparse Hamiltonian construction, 
not directly apply to $k$-local Hamiltonians, which are a more natural representation in physics. 
To bridge this gap, we establish an $\Omega(t)$ quantum communication lower bound for
simulating $k$-local Hamiltonians using an entanglement capacity argument. 
The proof has three ingredients. 
First, Ref.~\cite{schuchEntropyGrowthHardness2008} proved that a transverse-field Ising chain with only two crossing terms at the Alice–Bob boundary, 
starting from the product state, 
generates entanglement entropy $\Omega(t)$ across the bipartition within a reasonable parameter regime.
Second, any simulation achieving simulation error $\leq \epsilon$ in trace distance must approximately reproduce this entanglement  
by the Fannes–Audenaert continuity bound \cite{audenaertSharpContinuityEstimate2007}.
So, the simulation of the transverse-field Ising Hamiltonian also needs to generate $\Omega(t)$ entanglement.
Third, we prove that $Q$ qubits of quantum communication can generate at most $Q$ ebits of entanglement from a product state, 
using the entropy triangle inequality \cite{arakiEntropyInequalities1970}.
Combining these steps gives the communication complexity lower bound that simulating $k$-local Hamiltonian requires at least $\Omega(t)$.
See detailed proof of lower bounds in \cite{smDQS}.
}

{
In addition the lower bound of $\dqdhs$ in terms of evolution time $t$ (the no fast-forwarding theorem), 
we prove a fundamental lower bound on the communication complexity in terms of the simulation precision $\epsilon$ by making use of the unbounded-error lower bound of quantum communication complexity for the Inner-Product problem. 
\begin{theoremprl}[Communication lower bound of $\dqdhs$ regarding $\eps$]\label{apd:thm:lower_bound_eps}
    For constant evolution time $t$ and precision $\eps$,
    let $n = \Theta\qty(\frac{\log(1/\eps)}{\log\log(1/\eps)})$.
    There exists a sparse Hamiltonian $H_{\IP}$ acting on $2n+\log(n)+1$ qubits
    such that any generic bi-partite protocol for simulating $e^{-\ii H_{\IP}t}$ to error $\le \epsilon$ requires
    $\Omega\qty(\frac{\log(1/\eps)}{\log\log(1/\eps)})$ communication.
\end{theoremprl}
The key observation is that at any constant evolution time $t > 0$, 
the overlap of $e^{-\ii Ht}|\psi_0\rangle$ with the correct inner product answer is $|\sin(t/n)|^n$, 
while the overlap with the incorrect answer is exactly zero due to the disconnected-path structure of the Hamiltonian's graph \cite{berryExponentialImprovementPrecision2014}. 
An $\epsilon$-approximate simulation perturbs these overlaps by at most $\epsilon$ in trace distance, 
yielding an unbounded-error protocol for $\mathrm{IP}_n$ with success probability $> 1/2$ whenever $|\sin(t/n)|^{2n} > 2\epsilon$. 
Choosing $n = \Theta(\log(1/\epsilon)/\log\log(1/\epsilon))$ satisfies this condition, 
and the unbounded-error communication lower bound $\QC_0(\IP_n) = \Omega(n)$
\cite{forsterLinearLowerBound2001,iwamaUnboundederrorClassicalQuantum2007}
then gives a communication lower bound of $\Omega(\log(1/\epsilon)/\log\log(1/\epsilon))$, 
which matches the upper bound of our post-Trotter protocol. 
}

\subsection{Bounds for bipartite d-BE and d-RO} 

{While our general term-based distributed block-encoding (d-BE) applies universally to arbitrary (multipartite) network topologies, the bipartite setting admits a fundamentally more efficient, asymptotically optimal construction. The core physical distinction lies in the algebraic structure of the network cut: a general multipartite network involves multiple overlapping boundaries where the Hamiltonian decomposition is non-unique and highly complicated, making simultaneous communication optimization highly nontrivial. In contrast, across a single bipartite cut \(A|B\), the interaction Hamiltonian
has a well-defined minimal operator-Schmidt rank (OSR). We fix an arbitrary
minimal OSR decomposition of \(H_{AB}\).}

{By exploiting this property, we can mathematically compress the inter-node interaction $H_{AB}$ into exactly $\abs{\calE_{\rm ir}} := \operatorname{OSR}_{A|B}(H_{AB})$ irreducible non-local branches, i.e., $H_{AB} = \sum_{e=1}^{\abs{\calE_{\rm ir}}} A_e \otimes B_e$. The high-level idea of our optimized bipartite d-BE is to utilize local block-encodings to handle the intra-node complexities of $A_e$ and $B_e$ locally. Consequently, the only required quantum communication is routing a single multiplexed control register across the cut and back to synthesize the $\abs{\calE_{\rm ir}}$ boundary branches alongside the local Hamiltonians. This asymptotically bounded communication overhead is formalized as follows.}

\begin{propositionprl}[Quantum Communication Complexity of OSR-Based d-BE in Bipartite Networks]
\label{thm:bipartite_dbe_upper_bound}
Consider a bipartite distributed quantum network comprising two disjoint
computational nodes, Alice ($A$) and Bob ($B$). Let $H = H_A + H_B + H_{AB}$, where $H_A$ and $H_B$ are local Hamiltonians and $H_{AB}$ contains all
interactions crossing the $A|B$ cut. The exact distributed block-encoding (d-BE) of $H$ can be synthesized with a total quantum communication cost of
$   Q \le 2 \lceil \log_2(\operatorname{OSR}_{A|B}(H_{AB}) + 2) \rceil = \mathcal{O}\left( \log \abs{\calE_{\rm ir}} +2\right),$
where $\abs{\calE_{\rm ir}} = \operatorname{OSR}_{A|B}(H_{AB})$ is the exact operator-Schmidt rank of the interaction Hamiltonian.
\end{propositionprl}
We next show that this logarithmic dependence on the boundary OSR is
information-theoretically optimal up to constant factors.
\begin{theoremprl}[Information-theoretic communication lower bound for distributed block encoding]
Consider a bipartite distributed quantum network comprising two disjoint
computational nodes, Alice ($A$) and Bob ($B$). Let
$  H = H_A + H_B + H_{AB},$
where $H_A$ and $H_B$ are local Hamiltonians and $H_{AB}$ contains all
interactions crossing the $A|B$ cut. Suppose that
$\operatorname{OSR}_{A|B}(H_{AB})   = \abs{\calE_{\rm ir}} .
$
Then any distributed quantum protocol that exactly synthesizes a valid global
block-encoding of $H$, using locally initialized ancillary registers and no
free prior entanglement across the cut, must use
$ Q\ge\Omega\!\left(
    \log \abs{\calE_{\rm ir}}
    \right) = \Omega\!\left(
    \log\operatorname{OSR}_{A|B}(H_{AB})
    \right)$
qubits of quantum communication across the cut.
\label{lowerbound_dbe}
\end{theoremprl}

The proof of communication lower bounds is derived by bounding the growth of the OSR across the bipartite cut. Local operations cannot increase the OSR, while transmitting $Q$ qubits increases it by at most a factor of $4^Q$. Thus, starting from locally initialized ancillas, any synthesized global unitary $U$ satisfies $\operatorname{OSR}(U) \le 4^Q$.
For d-BE, perfectly block-encoding $H$ requires reproducing the boundary interaction $H_{AB}$, which enforces $\operatorname{OSR}(U) \ge \operatorname{OSR}(H)\ge \operatorname{OSR}_{A|B}(H_{AB})-2 = \abs{\calE_{\rm ir}}-2$. Combining these algebraic constraints yields $4^Q \ge \abs{\calE_{\rm ir}}-2$, establishing the fundamental limit $Q \ge \Omega(\log \abs{\calE_{\rm ir}})$. Our OSR-based bipartite d-BE construction matches this lower bound up to
constant factors. A more detailed proof can be found in \cite{smDQS}.

A similar OSR-based argument holds for the communication bound for the distributed reflection operator (d-RO) since $\text{OSR}(R_{AB}) = 2$.

\begin{corollaryprl}[Communication lower bound for distributed reflection operation.]\label{lowerbound_dro}
To exactly synthesize the distributed reflection operator, any distributed quantum protocol requires a minimum inter-node quantum communication cost of  $Q \ge \Omega(1)$ qubits.
\end{corollaryprl}

\subsection{Weak simulation--observable simulation via  \dpf.}

To formalize observable-driven speedup of distributed quantum simulation, we define an \textit{observable-dependent effective commutator norm} that precisely captures the error susceptibility of the target observable during the evolution:
$\tilde{\alpha}_{\mathrm{obs},p} := \sup_{\tau \in [0,t]} \sum_{e=0}^{\abs{\mathcal{E}}} \big\| [C_e^{(p)}, O(\tau)] \big\|$,
where $C_e^{(p)}$ aggregates all non-vanishing $p$-th order nested commutators rooted at the boundary edge $e \in \mathcal{E}$, 
and $O(\tau) = e^{\ii H\tau} O e^{-\ii H\tau}$ is the ideal Heisenberg evolution of the observable.  
The corollary of distributed observable simulation is given by replacing the (worst-case) commutator norm $\tilde{\alpha}_{\mathrm{qc},p}$ with the observable-dependent one $\tilde{\alpha}_{\mathrm{obs},p}$.
\begin{corollaryprl}
[Quantum communication complexity for local observables]
\label{prop:observable_error}
    Consider a $\Gamma$-node quantum network simulating a Hamiltonian evolution $U(t) = e^{-\ii Ht}$ with $H=H_0+\sum_{e=0}^{\abs{\mathcal{E}}} H_e$. 
    Let $O$ be a target local observable. 
    Utilizing a $p$th-order distributed product formula (d-PF), the quantum communication complexity required to estimate the expectation value $\langle O(t) \rangle$ within an additive error $\epsilon$ is 
    $\mathcal{O}\left( \Gamma |\mathcal{E}| \frac{t^{1+\frac{1}{p}} \big( \tilde{\alpha}_{\mathrm{obs},p}\big)^{\frac{1}{p}}}{\epsilon^{\frac{1}{p}}} \right),$
    where $\Gamma$ is the number of nodes and $\abs{\calE}$ is the number of interaction terms across multiple nodes. 
\end{corollaryprl}
The proof can be found in Appendix G in the Supplementary Materials \cite{smDQS}.


\subsection{Applications beyond quantum simulation}
With our distributed architecture, quantum phase estimation and Grover's algorithm can be easily implemented in quantum networks. 
For a $\Gamma$-partite quantum network and a general $\Gamma n$-qubit unitary $U=\sum_{j=1}^{J} \beta_j \bigotimes^{\Gamma}_{\gamma=1} U^{(\gamma)}_j$, the total quantum communication complexity of distributed quantum phase estimation for $U$ is $\bigO\qty(2^{2K}(\sum_j \abs{\beta_j})\Gamma\log J)$, where $K$ represents the number of decimals in the phase estimation. 
Meanwhile, the total quantum communication complexity of distributed Grover's algorithm is $\bigO\qty(\Gamma \sqrt{N})$, with $N$ denoting the number of items in the database \cite{smDQS}.

\onecolumngrid
 \appendix

\ifnum\onlysm=0
    \newpage
    \newcommand{\appendixtoc}{
      \begingroup
      \let\cleardoublepage\clearpage
      \let\clearpage\relax
      \let\appendixsection\section
      \renewcommand{\section}{\addtocontents{toc}{\protect\setcounter{tocdepth}{-1}}\appendixsection}
      \tableofcontents
      \endgroup
      \addtocontents{toc}{\protect\setcounter{tocdepth}{2}}
    }
\else
    \renewcommand{\thesection}{SM~\Roman{section}} 
\fi

\ifnum\onlysm=0
    \begin{center}
    {\bf \large Supplementary Materials: ``Distributed Quantum Simulation"} 
    \end{center}
\else
    \title{Supplementary Materials: ``Distributed Quantum Simulation"}
\fi

\author{Tianfeng Feng}
\email{These authors contributed equally to this work.}
\affiliation{\HKU}

\author{Jue Xu}
\email{These authors contributed equally to this work.}
\affiliation{\HKU}

\author{Wenjun Yu}
\affiliation{\HKU}

\author{Zekun Ye}
\affiliation{State Key Laboratory for Novel Software Technology, Nanjing University, Nanjing 210023, China Hefei National Laboratory, Hefei 230088, China}

\author{Penghui Yao}
\email{phyao1985@gmail.com}
\affiliation{State Key Laboratory for Novel Software Technology, Nanjing University, Nanjing 210023, China Hefei National Laboratory, Hefei 230088, China}

\author{Qi Zhao}
\email{zhaoqi@cs.hku.hk}
\affiliation{\HKU}

\ifnum\onlysm=0
\else
    \maketitle
    \vspace{-25pt}
\fi

\ifnum\onlysm=0
    \renewcommand{\addcontentsline}{\oldacl}
    \renewcommand{\tocname}{Contents}
    \tableofcontents
    \setcounter{theorem}{0}
\else
    \tableofcontents
\fi








\section{Distributed product (Trotter) formula}\label{apd:d_pf}

\subsection{Distributed PF (d-PF) for general Hamiltonians}\label{SI:d-PF}

In this section, we introduce our distributed product formula protocol (denoted $\dpf$) in detail. 
Given a $\Gamma n$-qubit Hamiltonian $H=\sum_{l=1}^L H_l$, 
there is a $\Gamma$-partite partition $\calV$ of $\Gamma n$ qubits, e.g., $\calV =\qty{v_\gamma:\qty[n\gamma+1, n(\gamma+1)]}$.
Based on the specific partition, one may cluster it into two parts, such that
    \begin{equation}\label{eq:SI:cluster_hamiltonian}
        H=\sum_{\gamma=1}^\Gamma H_\gamma+\sum_{e\in\calE} H_e,
    \end{equation}
    where the first summand contains the \emph{local} terms $H_\gamma:=\sum_{l:\supp(H_l)\subseteq v_\gamma} H_l$ and 
    the second summand contains all remaining terms called \emph{interaction} terms, which act between two or more two different parties. 
    This type of $H$ is referred to as the (induced) clustered Hamiltonian. 
Since $\{H_\gamma\}$ represents local Hamiltonian in a quantum network, we can assume that its exponential can be perfectly implemented.
   For simplicity below,  we rewrite \cref{eq:SI:cluster_hamiltonian} as
    \begin{equation}
        H=H_0+ \sum_{e\in\calE} H_e= \sum^{\abs{\calE}}_{e=0} H_e,
    \end{equation}
  where $H_0=\sum_{\gamma=1}^\Gamma H_\gamma$ and $\sum_{e\in\calE} H_e=\sum_{e=1}^{\abs{\calE}} H_e$. In the following, we make use of the $p$th-order product formula \cite{childsTheoryTrotterError2021} to formalize a distributed protocol for quantum simulation.
The $p$th-order product formula of a clustered Hamiltonian can be written as
\begin{equation}
    \pfU_p(t)=\prod_{y=1}^{\Upsilon}\prod_{e=0}^{\abs{\calE}}e^{-\ii ta_{(y,e)}H_{\pi_y(e)}},
\end{equation}
where the coefficients $a_{(y,e)}$ are real numbers. 
The parameter $\Upsilon$ denotes the number of stages of the formula \cite{childsTheoryTrotterError2021}. 
As a special case, the Suzuki formula with $\Upsilon=2 \times 5^{p/2-1}$ needs one to alternately reverse the ordering of summands between neighboring stages. 
Note that $\Upsilon$ and $\pi_y$ are fixed and the coefficients $a_{y,e}$  are uniformly bounded by 1 in absolute value.  

\begin{figure*}[!t]
    \centering
    \includegraphics[width=.8\linewidth]{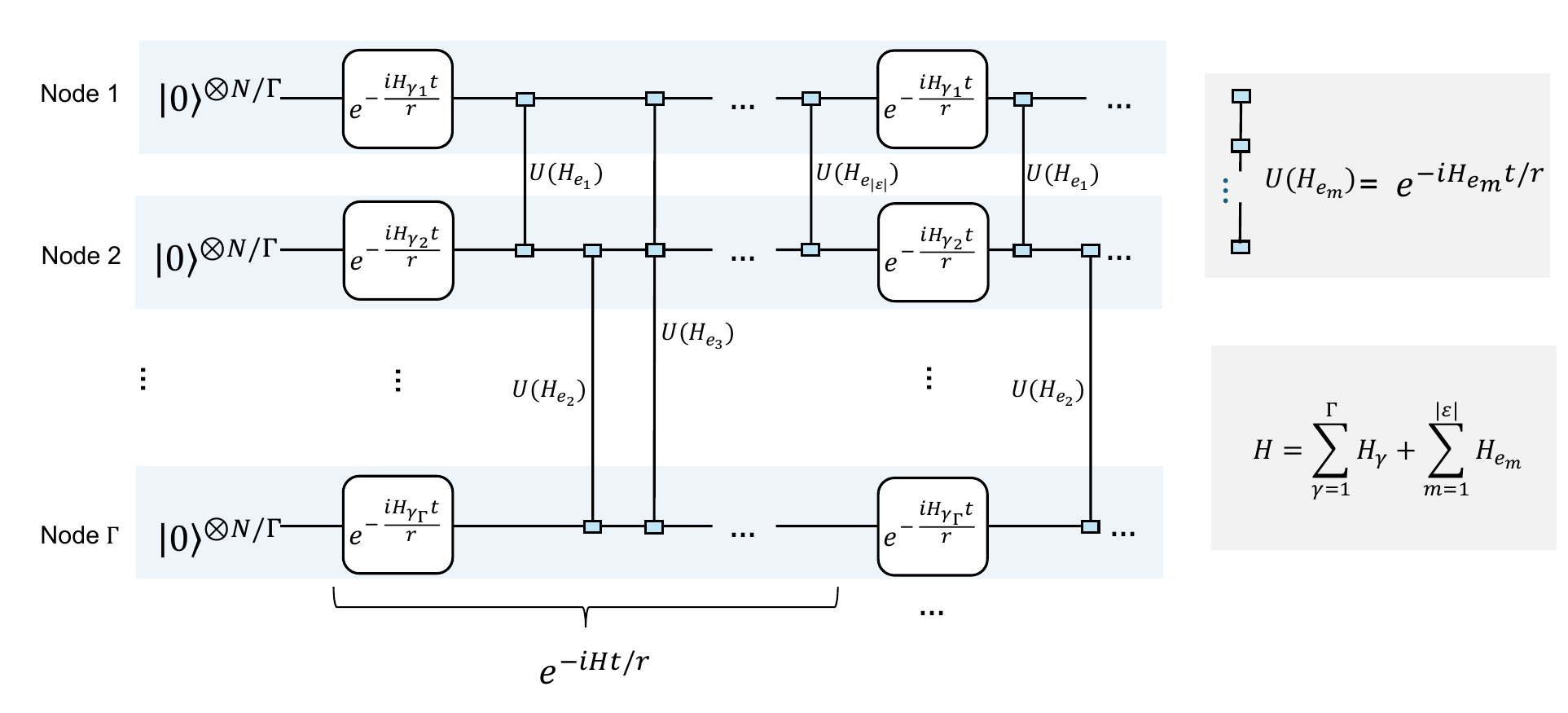}
    \caption{Quantum circuit of $\dpf$ for simulating a $\Gamma n$-qubit clustered Hamiltonian with a $\Gamma$-partite quantum network. 
    Specifically, $U(H_{e_m})$ represents the interaction term, where quantum communication is required for its implementation.}
    \label{fig:d-PF_circuit}
\end{figure*}

We present the simple quantum circuit of $\dpf$ in \cref{fig:d-PF_circuit} and the procedure of $\dpf$ for quantum simulation is as follows:
\begin{enumerate}
    \item  Given the Hamiltonian $H=\sum_{l}^L H_l$, regroup the Hamiltonian terms into the clustered Hamiltonian form, that is $H=\sum^{\abs{\calE}}_{e=0} H_e$, according to the assignment of qubits to the QPUs. 

    \item For a single segment (Trotter step), each node $\gamma$ performs the local terms $e^{-\ii H_\gamma t/r}$ locally. 
    The interaction terms $e^{-\ii H_e t/r}$ are implemented through quantum communication since the support of $H_e$ spans at least two nodes.

    \item Repeat the above steps according to the $p$th-order product formula $\pfU_p(t)=\prod_{y=1}^{\Upsilon}\prod_{e=0}^{\abs{\calE}}e^{-\ii ta_{(y,e)}H_{\pi_y(e)}}$.
\end{enumerate}


It can be seen that once $H$ is written in the clustered form, the whole process of $\dpf$ is quite straightforward, and quantum communication occurs only in the interaction component $U(H_{e_m})$ across nodes.
While a case-specific example of the Suzuki-Trotter decomposition is provided in \cite{buessenSimulatingTimeEvolution2023}, here we simply formalize the $\dpf$ in a general setting. 
It is worth noting that the $\dpf$ protocol can be extended to the analog quantum simulation \cite{ farhiAnalogAnalogueDigital1998, cubittUniversalQuantumHamiltonians2018, liuEfficientlyVerifiableQuantum2025}. 
In this scenario, during a single time segment (Trotter step), each local node evolves based on its local analog Hamiltonian. 
At the same time, the interactions between different nodes require the implementation of quantum control, which results in a certain amount of quantum communication cost.

\subsection{Communication complexity of d-PF}
In the following, we will analyze the quantum communication complexity of Hamiltonian simulations with the $\dpf$ in detail.
To begin with,
we introduce the following Lemma of the $p$th-order product formula (PF).
\begin{lemma}[Trotter error  with commutator scaling \cite{childsTheoryTrotterError2021}]\label{thm:pf}
    Let $H=\sum_{l=1}^L H_l$ be a Hermitian operator consisting with $L$ summands, and let $t \ge 0$. 
    Let  $\pfU_p(t)=\prod_{y=1}^{\Upsilon}\prod_{l=1}^{L}e^{-\ii ta_{(y,l)}H_{\pi_y(l)}}$ be a $p$th-order product formula. 
    Define $\tilde{\alpha}_{\cmm,p}:=\sum\limits_{l_1,\dots,l_{p+1}=1}^{L} \norm{[H_{l_1},[H_{l_2},\dots,[H_{l_p},H_{l_{p+1}}]]]}$. 
    Then, the Trotter error $\norm{e^{-\ii Ht}-\pfU_p(t)}$ can be asymptotically bounded as
    $\bigO(\tilde{\alpha}_{\cmm,p}t^{p+1})$.
    To achieve precision $\epsilon$, we need Trotter steps $r$
    for $\norm{e^{-\ii Ht}-\pfU^r_p(t/r)}\le \epsilon$, i.e.
    \begin{equation}\label{eq:lemmap-pf-trotter-number}
       r=\bigO\qty(\tilde{\alpha}^{1/p}_{\cmm,p}t^{1+1/p}/\epsilon^{1/p}).
    \end{equation}
\end{lemma}
\cref{thm:pf} gives the upper bound of Trotter steps $r$ for $p$th-order PFs in the general case.
Intuitively, for $\dpf$, the quantum communication complexity should be proportional to the Trotter steps $r$ and the number of interaction terms across nodes that require quantum communication.
\begin{theorem}[Quantum communication complexity of $\dpf$, Theorem 1 in maintext]\label{apd:thm:d_pf}
    Given a $\Gamma n$-qubit clustered Hamiltonian $H=H_0 +\sum_{e=1}^{\abs{\calE}} H_e $ with $H_0:=\sum_{\gamma=1}^\Gamma H_\gamma$ and evolution time $t\in\mathbb{R}^+$,
    quantum communication complexity of the $p$th-order distributed product formula (\dpf) to simulate $e^{-\ii H t}$ within accuracy $\eps$ is
    \begin{equation}
        \bigO(\abs{\calE} \Gamma \talpha_{\text{qc},p}^{1/p} \ t^{1+1/p}/\eps^{1/p}), 
    \end{equation}
    where $\tilde{\alpha}_{\text{qc},p}:=\sum\limits_{e_1,\dots,e_{p+1}=0}^{\abs{\calE}} \norm{[H_{e_1},\dots,[H_{e_p}, H_{e_{p+1}}]]}$ is the norm of the nested commutator of terms in $H$ (with $H_{e=0}=H_0$), $\Gamma$ is the number of nodes and $\abs{\calE}$ is the number of the interaction terms across multiple nodes.
\end{theorem}
\begin{proof}[Proof of \cref{apd:thm:d_pf}]
As we mentioned before, given an $n$-qubit Hamiltonian $H=\sum_{l=1}^L H_l$, one may rewrite it in the clustered form, i.e. $H=H_0+ \sum_{e\in\calE} H_e=\sum_{e=0}^{\abs{\calE}} H_e$,
where $H_0=\sum_{\gamma=1}^\Gamma H_\gamma$ and $\abs{\calE}$ is the number of the (induced) interaction terms.
In distributed quantum simulation, different parties correspond to quantum processors at different nodes, with local operations within nodes considered ideal, meaning the implementation of $e^{-\ii\delta t H_0}$ is either perfect or exhibits arbitrarily small error. 
Therefore, the Trotter step of the $p$th-order $\dpf$ should be 
\begin{equation}
   r_{\dpf}=\bigO\qty(\qty(\sum_{e_1,\dots,e_{p+1}=0}^{\abs{\calE}} \norm{[H_{e_1},[H_{e_2},\dots,[H_{e_p},H_{e_{p+1}}]]]})^{1/p} t^{1+1/p} / \epsilon^{1/p}),
\end{equation}
which is the same as the Trotter number for simulating the clustered Hamiltonian with circuit cutting (see Appendix I in \cite{childsTheoryTrotterError2021}).
That is, in this partition,  $\tilde{\alpha}_{\text{qc},p}:=\sum\limits_{e_1,\dots,e_{p+1}=0}^{\abs{\calE}} \norm{[H_{e_1},[H_{e_2},\dots,[H_{e_p},H_{e_{p+1}}]]]}$.%

For 2-local Hamiltonians, the quantum communication complexity of each segment is $\bigO(\abs{\calE})$, 
since each  interaction term $e^{-\ii\delta t H_{e\in\calE}}$ requires $\bigO(1)$ quantum communication to be realized where $\delta t$ can be arbitrary real number.
For example, $H_{e_m}=\norm{H_{e_m}}\sigma_{x_i} \otimes \sigma_{z_j} $, where $i\in v_i$, $j\in v_j$ and $i\ne j$ (where $v_i$ and $v_j$ are the set of qubits of $i$-th node and $j$-th node respectively). 
The implementation of $e^{-\ii\delta t \norm{H_{e_m}}\sigma_{x_i} \otimes \sigma_{z_j} }$ requires $\bigO (1)$ qubits to be teleported or transferred from $v_i$ to $v_j$, and vice versa.

On the other hand, let us consider the interaction term of the more general $k$-fold tensor product of Pauli operators and $k\ge \Gamma$, e.g., 
$H_{e_m}=\norm{H_{e_m}}\bigotimes _{j=1}^{\Gamma}(\otimes_{i\in v_j}\sigma_{q_i})$, 
where $1 \le m \le \abs{\calE}$ and $q_i\in\{0,x,y,z\}$ represents the four Pauli operators. 
This type of Hamiltonian has non-trivial effects on all or nearly all parts of a large system. Unlike the two-local Hamiltonian, which can be realized directly by moderating interactions, the evolution of such a Hamiltonian typically may require at least one auxiliary qubit. 
A standard approach to perform $e^{-\ii\delta t \norm{H_{e_m}}\bigotimes _{j=1}^{\Gamma}(\otimes_{i\in v_j}\sigma_{q_i})}$ is presented in \cite{nielsenQuantumComputationQuantum2010}, 
in which one ancillary qubit needs to be entangled with qubits of interaction.
Without loss of generality, we assume $H_{e_m}=\norm{H_{e_m}} \bigotimes _{j=1}^{\Gamma}(\otimes_{i\in v_j}\sigma_{z_i})$. Specifically, 
\begin{equation}
    e^{-\ii\delta t \norm{H_{e_m}} \bigotimes _{j=1}^{\Gamma}\qty(\otimes_{i\in v_j}\sigma_{z_i}) }
    =\qty(\bra{0}_a \otimes I_S \qty(\prod_{j=1}^{\Gamma}\prod_{i\in v_j} \text{CNOT}_{i,a})^\dagger) 
    \qty(e^{-\ii\delta t \norm{H_{e_m}}\sigma_z} \otimes I ) 
    \qty(\qty(\prod_{j=1}^{\Gamma}\prod_{i\in v_j} \text{CNOT}_{i,a}) \ket{0}_a \otimes I_S ),
\end{equation}
where $\text{CNOT}_{i,a}$ is the CNOT gate of $i$th qubit in $v_j$ (control qubit) and ancillary qubit (target qubit). $I_S$ represents the identity operator on the whole distributed quantum system. 

The cost of quantum communication required to implement $\prod_{i\in v_j} \text{CNOT}_{i,a}$ is $\bigO(1)$ since the auxiliary qubit can be transmitted to the $v_j$th node by performing quantum teleportation only once. Hence, the quantum communication complexity involved in executing $\prod_{j=1}^{\Gamma}\prod_{i\in v_j} \text{CNOT}_{i,a}$ is $O(\Gamma)$.  
In other words, the quantum communication cost of $e^{-\ii\delta t \norm{H_{e_m}} \bigotimes _{j=1}^{\Gamma} \qty(\otimes_{i\in v_j}\sigma_{z_i}) }$ is $O(\Gamma)$.
Similarly, by performing an appropriate single-qubit rotation, the same procedure can be employed to perform single $e^{-\ii\delta t \norm{H_{e_m}}\bigotimes _{j=1}^{\Gamma}\qty(\otimes_{i\in v_j}\sigma_{q_i})}$ with quantum communication complexity $\bigO(\Gamma)$. 
Considering that there are $\abs{\calE}$ induce interaction terms ($H_e$) in each segment, the quantum communication complexity of each segment is $\bigO(\abs{\calE}\Gamma )$ in general. 
Synthesizing the Trotter number of Trotter step of the clustered Hamiltonian, we obtain an upper bound on the quantum communication complexity of $\dpf$ as 
$\bigO(\abs{\calE} \Gamma \talpha_{\text{qc},p}^{1/p} \ t^{1+1/p}/\eps^{1/p})$.
\end{proof}



\section{Distributed post-Trotter protocols}\label{SI:d-AF}

Besides the product formula (Trotter method), the other promising class of quantum simulation algorithms is called post-Trotter methods.
The post-Trotter methods, including the Truncated Taylor series method (TS) and quantum signal processing (QSP) mainly expand the exponential function as a linear combination of functions of $H$. 
That is, generally, Hamiltonian simulation using the post-Trotter methods needs block encoding (BE) of matrix $H$ in the quantum circuits or quantum states. 
Typically, the linear combination of unitaries (LCU) technique \cite{longGeneralQuantumInterference2006,longDualityQuantumComputing2011,childsHamiltonianSimulationUsing2012} is the standard way to implement BE of $H$ \cite{lowHamiltonianSimulationQubitization2019}. 
However, LCU presents a probabilistic issue that requires the incorporation of oblivious amplitude amplification \cite{berryHamiltonianSimulationNearly2015}, necessitating another core operation, the reflection operation (RO). 
Below we will introduce distributed versions of LCU, BE, RO, TS, and QSP step by step.

\subsection{Distributed linear combination of unitaries (d-LCU)}\label{d-LCU}

In the main text, while there is no specific introduction to distributed LCU ($\dlcu$), it serves as the foundation for the implementation of distributed quantum algorithms, such as distributed Block encoding ($\dbe$) and distributed reflection operation ($\dro$). 
This is because both $\dbe$ and $\dro$ fundamentally rely on the nested LCU technique, which is crucial for their operation and effectiveness in distributed quantum computing tasks. 
Therefore, it is essential to recognize the significance of $\dlcu$ as a building block in the development of distributed quantum algorithms, e.g., distributed quantum phase estimation and distributed Grover's search.

Here we briefly review the LCU method \cite{longGeneralQuantumInterference2006,longDualityQuantumComputing2011,childsHamiltonianSimulationUsing2012}. 
Suppose our goal is to implement an $n$-qubit unitary operation $V$ on a quantum computer. 
As prior knowledge,  $V$ can be decomposed into 
  $  V=\sum_{j=1}^{J}\beta_j U_j$,
where $U_j$ is unitary and $\beta_j$ is corresponding coefficient. 
There is a unitary $G$ transforming the ancillary state $\ket{\mathbf{0}}$ to
\begin{equation}\label{apd:eq:prepare_single}
    G\ket{\mathbf{0}} = \frac{1}{\sqrt{s}} \sum_{j=1}^{J} \sqrt{\beta_j} \ket{\mathbf{j}}, 
\end{equation}
where $s$ is the normalized parameter, 
$s =\sum_{j=1}^{J}\beta_j$.  The select $U$ is defined as
$\select (U):= \sum_j \op{\mathbf{j}}  \otimes U_j$,
where the control register $\op{\mathbf{j}}$ acts on the ancillary system $\mathcal{H}_{\anc}$. Constructing an operator as
   $ W:=(G^\dagger \otimes I)\; \select (U) \; (G \otimes I)$.
When the input state is $|\mathbf{0}\rangle \otimes |\psi\rangle  $, acting $W$ on the global system, one has
\begin{equation}
    W|\mathbf{0}\rangle \otimes |\psi\rangle =\frac{1}{s}(|\mathbf{0}\rangle \otimes V|\psi\rangle ) + \frac{\sqrt{s^2-1}}{s} |\phi\rangle,
\end{equation}
where $|\phi\rangle$ is orthogonal to the subspace span $\{ |\mathbf{0}\rangle  \otimes \mathcal{H}_{\sys}\}$. 
After projecting to $|\vb{0}\rangle $ of the ancillary system, one may obtain the desired operation on the target system with probability $\frac{1}{s^2}$.
However, the probability of failure in some tasks should not be underestimated, so we need to amplify the probability of success. When $s = 2$, we can boost the probability through the oblivious amplitude amplification protocol,
\begin{equation}
    -WRW^\dagger RW|\vb{0}\rangle \otimes |\psi\rangle =|\vb{0}\rangle \otimes V |\psi  \rangle ,
\end{equation} 
where $R:= (I-2\op{\vb{0}})\otimes I$ is the reflection operation acting on $\calH_{\anc} \otimes \calH_{\sys} $.
For the case $s>2$, one may decompose $V$ into $V=\prod_{j=1}V^{(j)}$,
in which each $V^{(j)}=\sum_{i=1}\alpha^{(j)}_i U^{(j)}_i$ and $s^{(j)}=\sum_{i=1} \alpha^{(j)}_i=2$.

Here we extend the LCU technique to a quantum network.
Consider a system with a control node consisting of $m\Gamma$ qubits and $\Gamma$ local nodes, where each local node is made up of $n$ working qubits and $m$ auxiliary qubits.
The distributed LCU protocol (denote $\dlcu$) is to implement the $n\Gamma$-qubits operation $V$ in the quantum network. 
As prior knowledge, $V$ may be decomposed as
$V=\sum_{j=1}^{J}\beta_j \bigotimes_{\gamma=1}^{\Gamma} U^{(\gamma)}_j$
where $U^{(\gamma)}_j$ is a unitary of $\gamma$-th node and $m:=\lceil \log(J) \rceil$. 
Next, the control node prepares the following entangled state
\begin{equation}\label{eq:prepare_control}
    G\ket{\mathbf{0}}
    :=\frac{1}{\sqrt{s}} \sum_{j=1}^{J}\sqrt{\beta_j} |\mathbf{j}\rangle^{(1)}\otimes |\mathbf{j}\rangle^{(2)}\otimes ...\otimes|\mathbf{j}\rangle^{(\Gamma)}
    \equiv \frac{1}{\sqrt{s}} \sum_{j=1}^{J}\sqrt{\beta_j} \bigotimes_{\gamma=1}^{\Gamma} |\mathbf{j}\rangle^{(\gamma)}
    , 
\end{equation}
where the normalization factor $s =\sum_{j=1}^{J}\beta_j$, each $|\mathbf{j}\rangle$ is encoded by $m$ qubits and $\ket{\mathbf{j}}^{(\gamma)}$ indicates that these $m$ qubits will be transmitted to the $\gamma$-th node (using quantum teleportation).

Upon receiving the ancillary state \cref{eq:prepare_control} from the control nodes, 
each node executes a locally controlled unitary operation $\sum_{j_\gamma} \op{j_\gamma}\otimes U^{(\gamma)}_{j_\gamma}$,
and the global operation is $\bigotimes_{\gamma=1}^{\Gamma} \qty(\sum_{j_\gamma} \op{j_\gamma}\otimes U^{(\gamma)}_{j_\gamma})$. 
Note that the encoding basis for the ancillary states prepared in the control node is 
$\{\ket{0}^{\otimes \Gamma},\ket{1}^{\otimes \Gamma},...,\ket{J}^{\otimes \Gamma}\}$.
Thus the effective $\select(U)$ in the quantum network is given as $\Pi \bigotimes_{\gamma=1}^{\Gamma} (\sum_{j_\gamma} \op{j_\gamma}\otimes U^{(\gamma)}_{j_\gamma})\Pi$ where $\Pi=\bigotimes_{\gamma=1}^\Gamma \ket{j}^{(\gamma)}\bra{j}^{(\gamma)}=\sum_j\ket{j}^{\otimes\Gamma} \bra{j}^{\otimes\Gamma}$, i.e.
\begin{equation}\label{eq:selectU}
    \select(U):= \sum_j |\mathbf{j}\rangle ^{\otimes \Gamma} \langle \mathbf{j}|^{\otimes \Gamma} 
    \bigotimes_{\gamma=1}^{\Gamma} U^{(\gamma)}_j
\end{equation}
Similar to the non-distributed case, we construct an operator as
$W:=(G^\dagger \otimes I)\; \select(U)\; (G \otimes I)$.
When the initial state is $|\mathbf{0}\rangle \otimes ( \bigotimes_{\gamma=1}^\Gamma\ket{\psi_\gamma} )$, performing $W$ on the whole system, we have  
\begin{equation}
    W\ket{\mathbf{0}} \otimes \qty( \bigotimes_{\gamma=1}^\Gamma|\psi_\gamma\rangle)
    =\frac{1}{s} \qty(\ket{\mathbf{0}} \otimes V \bigotimes_{\gamma=1}^\Gamma\ket{\psi_\gamma})+ \frac{\sqrt{s^2-1}}{s} \ket{\phi},
    \label{eq:w}
\end{equation}
where $|\phi\rangle$ is orthogonal to the subspace span $\{ \ket{\mathbf{0}} \otimes \mathcal{H}_{\sys}^{\otimes \Gamma}\}$. 
After projecting to $\ket{\mathbf{0}} $ of the ancillary systems, we may obtain the desired operation on the target system with probability $\frac{1}{s^2}$.
However, the probability of failure in some tasks can be significant, hence, we need to amplify the probability of success.
When $s = 2$, we can boost the probability through the oblivious amplitude amplification protocol,
\begin{equation}
    -WRW^\dagger RW\ket{\mathbf{0}}  \otimes ( \bigotimes_{\gamma=1}^\Gamma|\psi_j\rangle ) 
    =\ket{\mathbf{0}}  \otimes V ( \bigotimes_{\gamma=1}^\Gamma|\psi_j\rangle ) ,
    \label{eq:ob_aa}
\end{equation} 
where $R:= (I-2\op{\mathbf{0}})\otimes I$ is the reflection acting on $\calH_{\anc}^{\otimes \Gamma} \otimes \calH_{\sys}^{\otimes \Gamma}$. 
The procedure of $\dlcu$ is summarized as follows:
\begin{enumerate}
    \item  The control node uses an $m\Gamma$-qubit processor to prepare entangled states 
    $\frac{1}{\sqrt{s}} \sum_{j=1}^{J}\sqrt{\beta_j} \bigotimes_{\gamma=1}^{\Gamma} |\mathbf{j}\rangle^{(\gamma)}$.
    \item The control node dispatches the prepared entangled states 
    to each of the $\Gamma$ nodes via $\Gamma$ quantum channels, with each node accepting $m$ qubits.

    \item Upon receiving the quantum states from the control node, the $\Gamma$ nodes store them in their local auxiliary quantum systems, each consisting of $m$ qubits. At this stage, each local node contains $m$ qubits of information regarding the entangled state distributed by the control node, as well as an $n$-qubit local quantum state.
    
    \item Each local auxiliary $m$-qubit system starts to interact with the local $n$-qubit system to carry out the corresponding unitary operation $U^{(\gamma)}_j$$(\gamma\in[\Gamma])$.
    
    \item The information of the respective $m$ qubits from the $\Gamma$ nodes after the local operations, is transmitted back to the control system, where the $G^\dagger$ operation is performed and the $m\Gamma$ qubits are measured to achieve the desired operation $V$ on the $n\Gamma$ qubits (if amplitude amplification is performed, instead of measuring qubits, $R$ should be performed). 
    Since this process is probabilistic, the use of amplitude amplification, achieved by repeating the previous steps, can help to execute the target operation on the $\Gamma n$ qubits in a nearly deterministic manner.
\end{enumerate}

Note that any operation can in principle be written in the form of LCU, thus universal quantum gates can also be implemented in the $\dlcu$ protocol, leading to universal quantum computation in quantum networks.\\

 
 \emph{Communication complexity of d-LCU.}
In the previously mentioned $\dlcu$ protocol, both quantum and classical communications occur during the process for distribution of entanglement from the control node to the $\Gamma$ nodes, and the process for the transmission of the auxiliary state from the $\Gamma$ nodes back to the control node. 
Two rounds of quantum teleportation involving $m\Gamma$ qubits are necessary to execute a single $W$ operation. Consequently, a $W$ operation demands the transmission of $2m\Gamma$ qubits of information and $4m\Gamma$ classical bits. It is important to note that the number of ancillary qubits, $m$, is dependent on the size of $L$, i.e., $m= \lceil \log(J) \rceil$.
 
\begin{itemize}
    \item 
    In the case of $s=2$, three $W$ operations are required to achieve the desired $n\Gamma$-qubit operation, which corresponds to a consumption of $6m\Gamma $ qubits and $12m\Gamma$ classical bits for quantum and classical communication, respectively. 
    
     \item 
    For the case $s>2$,  one may decompose $V$ into  
     $V=\sum_{j=1}\beta^{(k)}_j \bigotimes_1^{\Gamma} U^{(k,1)}_j$, 
     in which each $V^{(k)}=\sum_{i=1}\beta^{(k)}_i \bigotimes_{\gamma=1}^{\Gamma} U^{(k,\gamma)}_i$ and $s^{(k)}=\sum_{i=1} \beta^{(k)}_i=2$.
     In this case, the communication complexity is $6r\Gamma\lfloor \text{log}J \rfloor$ qubits and $12r\Gamma\lfloor \text{log}J \rfloor$ bits for quantum and classical communication, respectively. In general,
     by $\bigO(\sqrt{s})$ queries, one may perform $V$ with quantum and classical communication complexity  $\bigO(\sqrt{s}\Gamma \text{log}J) $. 
     
\end{itemize}

In the following section,  we only focus on quantum communication complexity and we will show that $\dlcu$ is a gadget that assists in implementing $\dbe$ and $\dro$.


\subsection{Distributed block encoding (d-BE)}\label{apd:dbe}

Before we introduce distributed TS and QSP protocols, it is necessary to design a construction for the block encoding of $H$ within a quantum network \cite{lowHamiltonianSimulationQubitization2019, berrySimulatingHamiltonianDynamics2015}.
As previously mentioned in the $\dpf$ section, given a $\Gamma n$-qubit Hamiltonian $H=\sum_l H_l$,
there is an interaction graph of clustered Hamiltonian $H=\sum_{\gamma=1}^\Gamma H_\gamma+\sum_{e\in\calE} H_e$ with $\Gamma$ nodes and $\abs{\calE}$ interaction terms $\calE \ne \emptyset$),
In the following sections, we will demonstrate that compared to the decomposition $H=\sum_l H_l$ (where $H_l$ are $k$-fold tensor product of Pauli operators),
the $\sum_{\gamma=1}^\Gamma H_\gamma+\sum_{e\in\calE} H_e $ may give a significantly low quantum communication complexity for d-BE in general cases.




Our task is to implement the $\Gamma n$-qubit block encoding of $H$ within a $\Gamma$-node quantum network, each node containing $n+\log(n)$ qubits where $\log(n)$ qubits are ancillary qubits. 
Recall the standard block encoding of $H=\sum_{l=1}^L H_l$ \cite{lowHamiltonianSimulationQubitization2019}
\begin{equation}\label{eq:standard_block_encoding}
    \langle \mathbf{0}|G^{\dagger} \otimes I \left( \sum_{l=1}^{L}\op{\mathbf{l}} \otimes \frac{H_l }{\norm{H_l}}\right) G|\mathbf{0}  \rangle \otimes I =\frac{H}{\alpha},
\end{equation}
where $G\ket{\mathbf{0}}=\sum \frac{1}{\sqrt{\alpha}} \sqrt{\norm{H_l} }\ket{l}$ and $\alpha=\norm{H}_1$ 
($\norm{\cdot}$ and $\norm{\cdot}_1$ are the spectral norm and 1-norm, respectively).
Let's now delve into how to implement a distributed block encoding within a quantum network.

To reduce quantum communication, we adopt the clustered Hamiltonian representation for block encoding. 
This approach eliminates the need to encode $\Gamma$ local Hamiltonians using quantum communication.
As shown in \cref{SI:fig:dbe_dro} (a), we present the nested block encoding for clustered Hamiltonian $H=\sum_{\gamma=1}^\Gamma H_\gamma+\sum_{e\in\calE} H_e$ in a quantum network. 
The $\dbe$ can be interpreted as either a linear combination of BEs or a nested linear combination of unitaries. 

The procedure for $\dbe$ is as follows.
To begin, the control node prepares
\begin{equation}\label{apd:eq:prepare}
   G_C\ket{\mathbf{0}}_C=
   \frac{1}{\sqrt{\alpha}}\qty(\sum_{\gamma=1}^{\Gamma} \sqrt{\norm{H_\gamma}_1} \ket{\gamma}^{\otimes \Gamma}_C+\sum_{e=1}^{\abs{\calE}}\sqrt{\norm{H_e}}\ket{e}^{\otimes \Gamma}_C), 
\end{equation}
where $\alpha=\norm{H}_1=\sum_{\gamma=1}^{\Gamma} \norm{H_\gamma}_1 +\sum_{e=1}^{\abs{\calE}}\norm{H_e}$. 
Note that $\ket{i}^{\otimes \Gamma}$ is used instead of $\ket{i}$ because, in our distributed protocol, we subsequently have to distribute the prepared auxiliary quantum states to each node. The advantage of using the basis $\ket{i}^{\otimes \Gamma}$ is that it allows the simultaneous transmission of individual qubits to each node, all of which share the same quantum amplitude $\{\sqrt{\norm{H_\gamma}_1}, \sqrt{\norm{H_e}}\}$.

The remaining $\Gamma$ ancillary registers are distributed across $\Gamma$ nodes, with each register used to execute the block-encoding of the local $H_{\gamma}$.
Specifically, $\gamma$-th ancillary register prepares
$G_\gamma\ket{\mathbf{0}}_\gamma=\frac{1}{\sqrt{\alpha_{\gamma}}} \sum_{l\in v_\gamma}\sqrt{\norm{H_l} }\ket{l}_\gamma$
where $\alpha_\gamma=\norm{H_\gamma}_1$ and $v_\gamma$ represents the set of qubits on the $\gamma$th processor. 
From the standpoint of the standard BE of LCU, the distributed $G$ operation for the entire ancillary system which includes control nodes and auxiliary systems for each local node is expressed as
$G=G_C\otimes \bigotimes_{\gamma=1}^{\Gamma} G_\gamma$.
In a manner similar to $\dlcu$, we define the effective $\select(U):=\sum_{\gamma=1}^{\Gamma} |\gamma\rangle^{\otimes \Gamma} \langle \gamma|^{\otimes \Gamma}_C \otimes V_\gamma +V_{\intt}$ 
which is explicitly expanded as
\begin{align}\label{d-BE-sU}
    \select(U)
    &=\sum_{\gamma=1}^{\Gamma} |\gamma\rangle^{\otimes \Gamma} \langle \gamma|^{\otimes \Gamma}_C \otimes \qty(\sum_{l\in v_\gamma}\op{l}_\gamma\otimes \frac{H_l}{\norm{H_l}})
    +\sum_{e=1}^{\abs{\calE}}|e\rangle^{\otimes \Gamma} \langle e|^{\otimes \Gamma}_C \otimes \frac{H_e}{\norm{H_e}}.
\end{align}


It should be noted that we have projected the ancillary system onto its encoding basis $\{\ket{0}^{\otimes \Gamma},\ket{1}^{\otimes \Gamma},...,\ket{J}^{\otimes \Gamma}\}$.  
Now the $\dbe$ of $H$ gives
\begin{equation}
    \qty( \langle \mathbf{0}|_C G_{C}^{\dagger}\bigotimes_{\gamma =1}^{\Gamma}{\langle \mathbf{0}|_{\gamma}}G_{\gamma}^{\dagger} ) 
    \select(U) \qty( G_C|\mathbf{0}\rangle_C \bigotimes_{\gamma =1}^{\Gamma}{G_{\gamma}|\mathbf{0}\rangle}_{\gamma}) 
    =\frac{H}{\alpha }.
\end{equation}
Note that the implementations of $\dbe$ of $H$ have the same denominator $\alpha=\sum_{\gamma=1}^\Gamma  \norm{H_\gamma}_1+\sum_{e\in\calE} \norm{H_e}$ as standard BE \cite{lowHamiltonianSimulationQubitization2019}.  
Expanding $H$ may help us better understand this result. 
It can be seen that
\begin{equation}
    H = \sum_{\gamma=1}^\Gamma H_\gamma+\sum_{e\in\calE} H_e
    \equiv\sum_{\gamma=1}^\Gamma  \norm{H_\gamma}_1  \frac{H_\gamma }{\norm{H_\gamma}_1 }+\sum_{e\in\calE} H_e
    \equiv\qty(\sum_{\gamma=1}^\Gamma  \norm{H_\gamma}_1+\sum_{e\in\calE} \norm{H_e})\frac{H}{\sum_{\gamma=1}^\Gamma  \norm{H_\gamma}_1+\sum_{e\in\calE} \norm{H_e}}.
\end{equation}
The $\dbe$ must be normalized with respect to the partition of $H$ and the local BE for $H_\gamma$. 
Specifically, the auxiliary states of each local BE for $H_\gamma$ are normalized, with the normalization coefficients equal to $\norm{H_\gamma}_1$. 
The above equation shows
that it is decomposed in a way that ensures the normalization of the coefficients for all local nodes, while also normalizing the coefficients of the control nodes (with the normalization coefficient being $\norm{H}_1$).

\emph{Quantum communication complexity of d-BE.}\label{qcc_dBE} 
In the $\dbe$ protocol described above, the quantum communication cost arises from the process of distributing the quantum state 
\cref{apd:eq:prepare}
, which requires the transmission of $\Gamma \log(\abs{\calE}+\Gamma)$ qubits. 
In addition, after executing $\select(U)$, the control node needs to perform the $G^\dagger_C$ operation, thus requiring each node to resend the corresponding $ \log(\abs{\calE}+\Gamma)$ qubits back to the control node. 
Thus the overall quantum communication cost of $\dbe$ is $2\Gamma \log(\abs{\calE}+\Gamma)$, i.e., the quantum communication complexity is $\bigO(\Gamma \log(\abs{\calE}+\Gamma))$. Note that if $\abs{\calE}=0$, this suggests that no interaction terms are crossing over a quantum network, and the communication cost should be $0$ since one can implement the BE of Hamiltonian locally. In this case, $\dbe$ is no longer necessary.

\begin{figure}[!t]\label{SI:fig:dbe_dro}
    \centering
    \sidesubfloat[]{\includegraphics[width=0.45\linewidth]{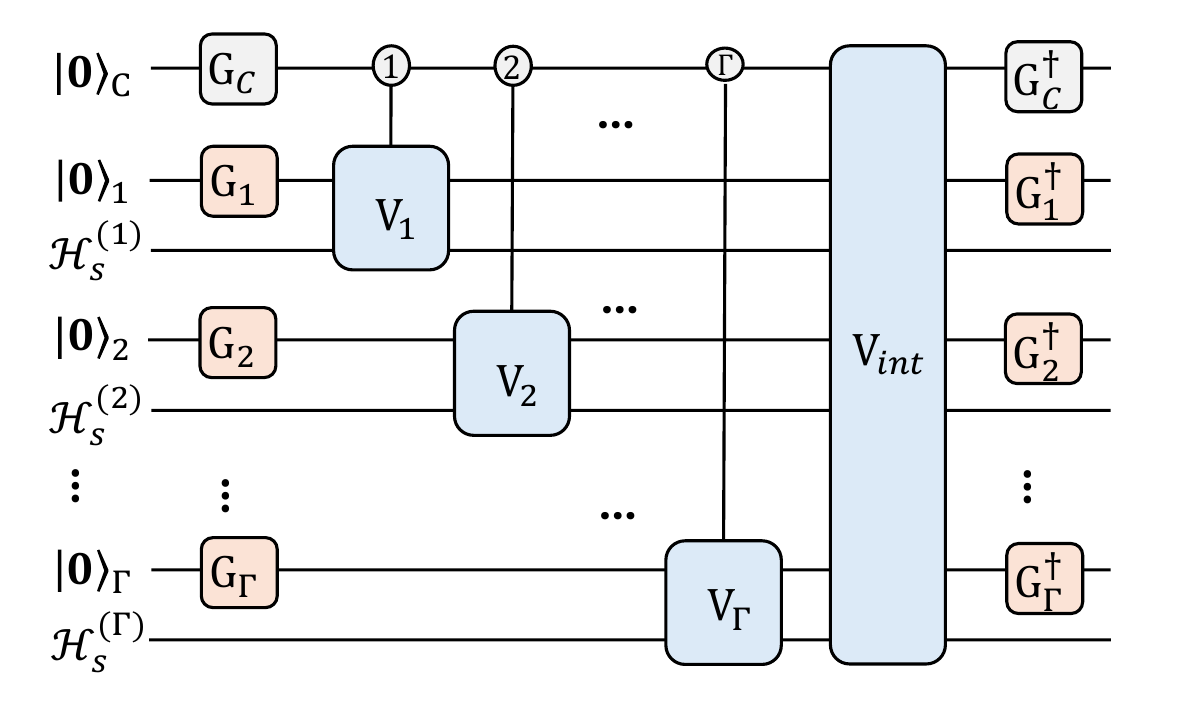} }
    \hfill
    \sidesubfloat[]{\includegraphics[width=0.41\linewidth]{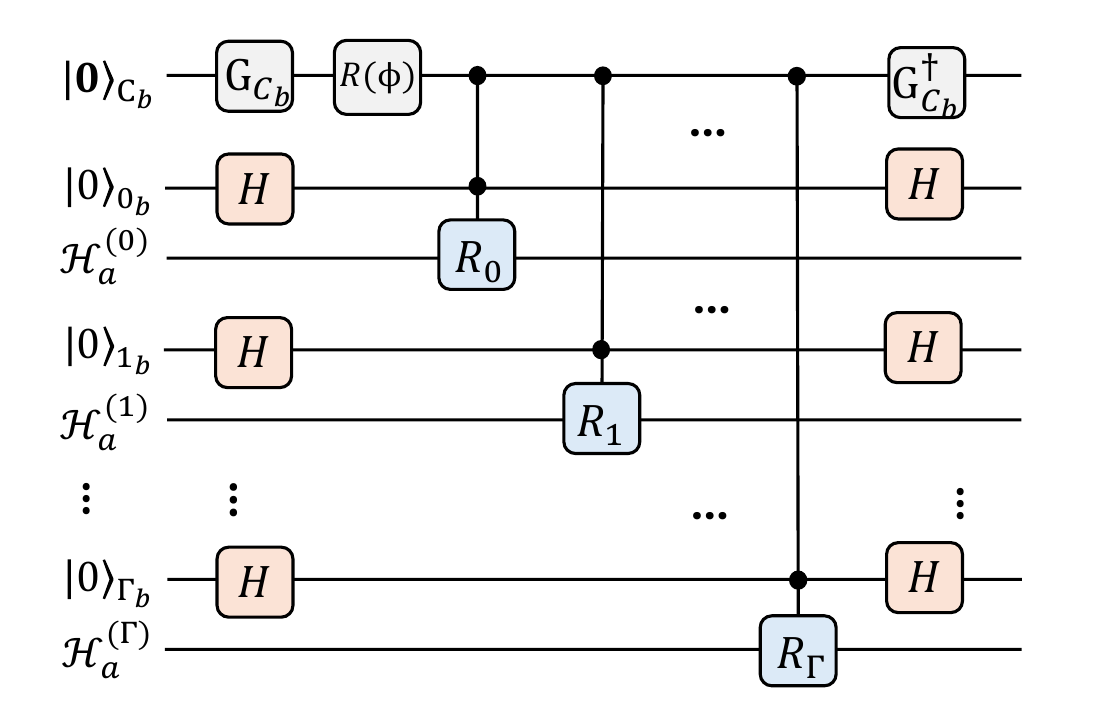}   }
    \caption{(a) Distributed block encoding (\dbe) of $H=\sum_{\gamma=1}^\Gamma H_\gamma+\sum_{e\in\calE} H_e$ using $\dlcu$ in a quantum network with $\Gamma$ nodes and a control node. 
    The control node prepares $G_C\ket{0}$ as \cref{apd:eq:prepare}.
    The implementation of $V_{\intt}$ (for $\{H_e\}$) requires that the qubits of the control node are needed to be transmitted to the entire network.
    (b) Distributed reflection operation ($\dro$).  
    In a quantum (star) network, the distributed reflection operation ($\dro$) can act on the qubits of the $\Gamma$ nodes and the central control node. 
    For simplicity, we designate the control node as node 0, i.e., $\gamma=0$. The newly introduced ancillary system $C_b$ of the control node is then used to facilitate the implementation of $\dro$.  
    Without loss generality, $\ket{i}_{C_b}$ can be defined as $\ket{i}_{C_b}=\ket{i}\otimes \ket{i}^{\otimes \Gamma}$, where each qubit in $\ket{i}^{\otimes \Gamma}$ should be individually sent to each of the $\Gamma$ nodes.
    }
\end{figure}

{

\subsection{Advanced d-BE and Communication Lower Bounds for Bipartite Networks}
\label{apd:dbe_bipartite_osr}

While our general term-based d-BE applies universally to arbitrary (multipartite) network topologies, the bipartite setting admits a fundamentally more efficient, strictly optimal construction. The core physical distinction lies in the algebraic structure of the network cut: a general multipartite network involves multiple overlapping boundaries where the Hamiltonian decomposition is non-unique, making simultaneous communication optimization highly nontrivial. In contrast to multipartite networks, where multiple overlapping
cuts make simultaneous communication optimization more complicated, a single
bipartite cut \(A|B\) provides a well-defined minimal OSR for the boundary
interaction. 

To establish a rigorous mathematical foundation, we first formalize the distributed block-encoding task. Here, for simplicity, we consider the bipartite case.

\begin{definition}[Distributed Block-Encoding Problem]\label{def:distributed_be}
Consider an arbitrary quantum network partitioned by a bipartite cut into two disjoint regions, $A$ and $B$, with any intermediate node formally absorbed into either partition. A distributed block-encoding protocol utilizes local operations and classical communication (LOCC) together with inter-partition quantum communication to synthesize a global unitary $U_{\text{total}}$. This unitary acts on the joint data register $\calH_A \otimes \calH_B$ and localized ancillary registers $\calH_{a_A} \otimes \calH_{a_B}$. The protocol is defined as valid if, given the initial ancillary state $\ket{\mathbf{0}}_{anc} = \ket{0}_{a_A} \otimes \ket{0}_{a_B}$, it satisfies the Hamiltonian projection condition:
\begin{equation}
    \big(\bra{\mathbf{0}}_{anc} \otimes I_{AB}\big) U_{\text{total}} \big(\ket{\mathbf{0}}_{anc} \otimes I_{AB}\big) = \frac{H}{c},
\end{equation}
where $H$ is the target global Hamiltonian and $c \ge \norm{H}$. The quantum communication complexity of this task is defined as the minimum inter-partition quantum communication cost $Q$, measured in transmitted qubits. Equivalently, if teleportation is used, the consumed ebits are charged as the corresponding qubit-communication resources.
\end{definition}

We now consider an optimized d-BE protocol for bipartite networks. In this setting, the communication overhead can be further minimized by exploiting the OSR of the boundary Hamiltonian. We call this OSR-based d-BE protocol.

Consider a two-node quantum network partitioned into nodes $A$ and $B$. We write the global Hamiltonian as
\begin{equation}
    H = H_A + H_B + H_{AB},
\end{equation}
where $H_A$ and $H_B$ are local Hamiltonians on nodes $A$ and $B$, and $H_{AB}$ contains all interactions crossing the cut. We define the irreducible interaction size $\abs{\calE_{\rm ir}}$ via the operator-Schmidt rank: $\abs{\calE_{\rm ir}} := \operatorname{OSR}_{A|B}(H_{AB})$. Note that $\operatorname{OSR}_{A|B}(H_{AB})\le |\mathcal {E}|$ where $\mathcal {E}$ is the raw decomposition with Pauli operators, i.e., $\abs{\calE_{\rm ir}}\le |\mathcal {E}|$.
Equivalently, the boundary Hamiltonian admits an irreducible OSR decomposition into exactly $\abs{\calE_{\rm ir}}$ linearly independent terms:
\begin{equation}
    H_{AB} = \sum_{e=1}^{\abs{\calE_{\rm ir}}} H_e, \qquad H_e = A_e \otimes B_e.
    \label{eq:osr_ir_decomposition}
\end{equation}

Unlike the general term-based construction in \cref{apd:dbe}, here $H_e$ is not restricted to a single Pauli string. Instead, $A_e$ and $B_e$ are linear combinations of local Pauli operators. 
Let $\alpha_{A,e}$ and $\alpha_{B,e}$ denote the respective Pauli 1-norm normalization factors of the local block encodings for $A_e$ and $B_e$. Consequently, each boundary term $H_e = A_e \otimes B_e$ is synthesized with the joint normalization $\alpha_e = \alpha_{A,e} \alpha_{B,e}$. Let $\alpha_A = \norm{H_A}_1$ and $\alpha_B = \norm{H_B}_1$ be the standard normalization factors for the local Hamiltonians.

The procedure for the bipartite OSR-based $\dbe$ is as follows. 
To begin, the control node prepares
\begin{equation}
    G_C\ket{\mathbf{0}}_C = \frac{1}{\sqrt{\alpha_{\rm ir}}} \qty( \sqrt{\alpha_A}\ket{A}^{\otimes 2}_C + \sqrt{\alpha_B}\ket{B}^{\otimes 2}_C + \sum_{e=1}^{\abs{\calE_{\rm ir}}} \sqrt{\alpha_e}\ket{e}^{\otimes 2}_C ),
    \label{eq:prepare_osr_dbe}
\end{equation}
where $\alpha_{\rm ir} = \alpha_A + \alpha_B + \sum_{e=1}^{\abs{\calE_{\rm ir}}} \alpha_e$. This control register $C$ requires $\lceil \log(\abs{\calE_{\rm ir}}+2)\rceil$ qubits.

The remaining ancillary registers are distributed across nodes $A$ and $B$, with each register used to execute the local block encodings of $H_A, H_B$ and $A_e, B_e$. From the standpoint of the standard BE of LCU, the distributed $G$ operation for the entire ancillary system, which includes the control node and auxiliary systems for each local node, is expressed as $G = G_C \otimes G_A \otimes G_B$.

In a manner similar to $\dlcu$,
we define the effective $\select(U)$ for nodes $A$ and $B$. Conditioned on $C$, we have 
\begin{align}
    \select(U)  
   = \op{A}^{\otimes 2}_C \otimes U_A + \op{B}^{\otimes 2}_C \otimes U_B + \sum_{e=1}^{\abs{\calE_{\rm ir}}} \op{e}^{\otimes 2}_C \otimes \qty( U_{A,e} \otimes U_{B,e} ),
    \label{eq:select_BA_combined}
\end{align}
where $U_A$ and $U_B$ are the local block encodings of $H_A/\alpha_A$ and $H_B/\alpha_B$, and $U_{A,e}$ and $U_{B,e}$ are the local block encodings of the boundary components $A_e/\alpha_{A,e}$ and $B_e/\alpha_{B,e}$, respectively.

It should be noted that we project the ancillary system onto its encoding basis $\ket{\mathbf{0}}_{\rm anc} = \ket{\mathbf{0}}_C \otimes \ket{\mathbf{0}}_A \otimes \ket{\mathbf{0}}_B$. Now the $\dbe$ of $H$ gives
\begin{equation}
    \bra{\mathbf{0}}_{\rm anc} U_{\text{total}} \ket{\mathbf{0}}_{\rm anc} = \frac{1}{\alpha_{\rm ir}} \qty( \alpha_A \frac{H_A}{\alpha_A} + \alpha_B \frac{H_B}{\alpha_B} + \sum_{e=1}^{\abs{\calE_{\rm ir}}} \alpha_e \frac{A_e}{\alpha_{A,e}} \frac{B_e}{\alpha_{B,e}} ) = \frac{H}{\alpha_{\rm ir}}.
\end{equation}

The bipartite $\dbe$ must be normalized with respect to the partition of $H$ and the local OSR components. Specifically, the auxiliary states of each local BE for $A_e$ and $B_e$ are normalized, with the combined normalization coefficients equal to $\alpha_e = \alpha_{A,e}\alpha_{B,e}$. The above equation shows that it is decomposed in a way that ensures the normalization of the coefficients for all local nodes, while also normalizing the coefficients of the control node (with the total normalization coefficient being $\alpha_{\rm ir}$). Note that $\alpha_{\rm ir}=\text{ploy}(n)$ since  we assume $\norm{H}_1=\text{ploy}(n)$.

The total quantum communication cost originates solely from routing the control register $C$ across the bipartite cut, which yields exactly $Q = 2\lceil \log(\abs{\calE_{\rm ir}}+2) \rceil = \bigO(\log\operatorname{OSR}_{A|B}(H_{AB}))$. The formal theorem of communication of the OSR-based d-BE protocol is as follows.

\begin{proposition}[Quantum Communication Complexity of OSR-Based d-BE in Bipartite Networks]
\label{thm:bipartite_dbe_upper_bound}
Consider a bipartite distributed quantum network comprising two disjoint
computational nodes, Alice ($A$) and Bob ($B$). Let $H = H_A + H_B + H_{AB}$, where $H_A$ and $H_B$ are local Hamiltonians and $H_{AB}$ contains all
interactions crossing the $A|B$ cut. The exact distributed block-encoding (d-BE) of $H$ can be synthesized with a total quantum communication cost of
$$   Q \le 2 \lceil \log_2(\operatorname{OSR}_{A|B}(H_{AB}) + 2) \rceil = \mathcal{O}\left( \log \abs{\calE_{\rm ir}} +2\right),$$
where $\abs{\calE_{\rm ir}} = \operatorname{OSR}_{A|B}(H_{AB})$ is the exact operator-Schmidt rank of the interaction Hamiltonian.
\end{proposition}

We now show that this OSR-based d-BE protocol is optimal for any (nontrivial) bipartite network.

\begin{theorem}[Information-theoretic communication lower bound for distributed block encoding]
Consider a bipartite distributed quantum network comprising two disjoint
computational nodes, Alice ($A$) and Bob ($B$). Let
$
    H = H_A + H_B + H_{AB},
$
where $H_A$ and $H_B$ are local Hamiltonians and $H_{AB}$ contains all
interactions crossing the $A|B$ cut. Suppose that
$
    \operatorname{OSR}_{A|B}(H_{AB})
    =
    \abs{\calE_{\rm ir}} .
$
Then any distributed quantum protocol that exactly synthesizes a valid global
block-encoding of $H$, using locally initialized ancillary registers and no
free prior entanglement across the cut, must use
$
    Q
    \ge
    \Omega\!\left(
    \log \abs{\calE_{\rm ir}}
    \right)
    =
    \Omega\!\left(
    \log \operatorname{OSR}_{A|B}(H_{AB})
    \right)
$
qubits of quantum communication across the cut.
\label{lowerbound_dbe}
\end{theorem}
\begin{proof}
Let \(U\) be the global unitary synthesized by the distributed protocol,
including all local operations and communicated quantum registers. Since the
protocol exactly block-encodes the full Hamiltonian
$
H=H_A\otimes I_B+I_A\otimes H_B+H_{AB},$
there exists a normalization factor \(\alpha\) such that
$\bra{\mathbf{0}}_{\rm anc} U \ket{\mathbf{0}}_{\rm anc}
=\frac{H}{\alpha}.$
Taking matrix elements over auxiliary registers cannot increase the
operator-Schmidt rank across the bipartite cut, and multiplying by the scalar
\(\alpha\) does not change it. Hence
$\operatorname{OSR}_{A|B}(H)
\le
\operatorname{OSR}_{A|B}(U).$

Let
$H_{\rm loc}:=H_A\otimes I_B+I_A\otimes H_B.$
Since \(H_{\rm loc}\) has operator-Schmidt rank at most \(2\), and
\(H_{AB}=H-H_{\rm loc}\), subadditivity of the operator-Schmidt rank implies
$\operatorname{OSR}_{A|B}(H_{AB})
\le\operatorname{OSR}_{A|B}(H) +\operatorname{OSR}_{A|B}(H_{\rm loc})
\le\operatorname{OSR}_{A|B}(H)+2.$
Therefore
$\operatorname{OSR}_{A|B}(H_{AB})-2
\le\operatorname{OSR}_{A|B}(H)
\le\operatorname{OSR}_{A|B}(U).$
Local operations within nodes \(A\) and \(B\) cannot increase the OSR across the
cut. The only operations that can increase the OSR are quantum transmissions
across the cut. Sending a \(q\)-qubit register across the cut can increase the
operator-Schmidt rank by at most a factor \(4^q\). Therefore, if the total
communication cost is \(Q\) qubits, then
$\operatorname{OSR}_{A|B}(U)
\le4^Q .$
Combining the inequalities gives
$\operatorname{OSR}_{A|B}(H_{AB})-2
\le
4^Q .$
Thus, whenever \(\operatorname{OSR}_{A|B}(H_{AB})>2\),
$
Q\ge\frac{1}{2}\log_2\!\left(
\operatorname{OSR}_{A|B}(H_{AB})-2
\right).
$
In particular,
$Q=\Omega\!\left(\log \operatorname{OSR}_{A|B}(H_{AB})
\right),$
which proves the claimed asymptotic lower bound.
\end{proof}

\begin{remark}
This OSR-based construction is specific to bipartite topologies. In a
general multipartite network, different bipartitions induce different
operator-Schmidt ranks and different irreducible decompositions. A
decomposition optimized for one cut need not be compatible with, or optimal
for, another cut. Therefore, the general multi-node d-BE construction in
\cref{apd:dbe} employs the Pauli-term-based formulation to accommodate
arbitrary network geometries. In contrast, for a two-node network, the
boundary Hamiltonian has a well-defined OSR decomposition with minimal rand, which
leads to the optimal communication complexity above.
\end{remark}

}


\subsection{Distributed reflection operation (d-RO)}\label{apd:d-RO}
In the post-Trotter quantum simulation protocols, a crucial element besides BE is the reflection operation (RO) to amplify of the success probability of the algorithm. 
The reflection operation acting on the ancillary system is defined as $R=I-(1-e^{-\ii\phi})\op{\mathbf{0}}$. 
In the LCU and Grover's algorithm, one typically sets $\phi = \pi$, whereas in QSP the setting of $\{\phi\}$ is dependent on the objective function.
Our goal is to implement distributed reflection operation (denoted $\dro$) in a quantum network ($\Gamma$ nodes and a control node), that is
\begin{equation}
    R = I-(1-e^{-\ii\phi})\ket{0}^{\otimes \Gamma} \bra{0}^{\otimes\Gamma}_C \otimes \bigotimes_{\gamma=1}^{\Gamma} \op{\mathbf{0}}_\gamma=I-(1-e^{-\ii\phi})  \op{\mathbf{0}}_{C} \otimes \bigotimes_{\gamma=1}^{\Gamma}  \op{\mathbf{0}}_\gamma,
\end{equation}
where $\ket{\mathbf{0}}_C=\ket{i}^{\otimes \Gamma}_C$.
The reflection operation essentially performs a phase gate on the target component $\ket{\mathbf{0}}$ (acting on $\Gamma\log{\abs{\calE}}+\sum^{\Gamma}_{\gamma=1}\log {L_\gamma}$ qubits of overall ancillary systems, where $L_\gamma$ is the number of terms of local Hamiltonian $H_\gamma$), which may require $\bigO(1)$ multi-qubit Toffoli gates to implement. 
For performing this reflection operation $R$, a naive approach is for each node to transfer its local ancillary state to the control node with the quantum communication complexity $\bigO(\sum^{\Gamma}_{\gamma=1}\log {L_\gamma})$. 
Obviously, the communication cost is large if each $L_\gamma \gg  \abs{\calE}+\Gamma$.
A smarter approach is to introduce auxiliary qubits and realize this multi-qubit reflection operation by cascading a series of Toffoli gates, such as the method in \cite{nielsenQuantumComputationQuantum2010}. This approach requires the layout of the circuit to $\Gamma$ nodes, and the communication complexity may be $\bigO(\Gamma)$. 
However, as mentioned above, the operation of cascading in fact increases the overall circuit depth.
Considering the parallel nature of distributed computing, here we propose a $\dro$ based on nested d-LCUs in a network, which has the same communication complexity while having a shallower circuit depth compared to the former. 


Here we give a concrete construction of $\dro$ in a quantum network with $\Gamma+1$ nodes. 
First, we introduce the local reflection operation of $\gamma$-th node, which is defined as
$R_\gamma= 2 \op{\mathbf{0}}_\gamma-I_{\gamma}$.
Similarly, the local reflection operation of the control node is 
$R_\gamma= 2 \op{\mathbf{0}}_\gamma-I_{\gamma}$.
Moving the identity operator $I_\gamma$ and $I_{C}$ from the right to the left, we get
\begin{equation}
   \frac{R_\gamma +I_\gamma}{2} =  \op{\mathbf{0}}_\gamma,
     \frac{R_{C} +I_{C}}{2} =  \op{\mathbf{0}}_{C}.
\end{equation}
Now we rewrite the $R$ operation as
\begin{equation}
    R 
    =I-(1-e^{-\ii\phi})  \qty(\frac{R_{C} +I_{C}}{2})\otimes \bigotimes_{\gamma=1}^{\Gamma} \qty(\frac{R_\gamma +I_\gamma}{2})
    =I-(1-e^{-\ii\phi}) \bigotimes_{\gamma=0}^{\Gamma} \qty(\frac{R_\gamma +I_\gamma}{2}).
\end{equation}
For ease of expression, we have set the control node as node 0 i.e. $\gamma=0$.
Since $R$ is unitary, and the above equation shows that $R$ can be rewritten as a linear sum of unitary operations, one may introduce a new auxiliary system to implement $\dro$ via LCU. 
We can use the nested LCU technique to achieve this operation, where the auxiliary system of the
control node is only 2-dimensional. Specifically, we introduce a new ancillary qubit at the control node and each local node, respectively, i.e. $\ket{0}_{C_{b}}\otimes \bigotimes_{\gamma=0}^{\Gamma}\ket{0}_{\gamma_b}$. 
These ancillary registers are prepared to
\begin{equation}
    \qty(G_{C_b}\otimes \bigotimes_{\gamma=1}^{\Gamma} H_{\gamma_b}  )\ket{0}_{C_{b}}\otimes \bigotimes_{\gamma=0}^{\Gamma}\ket{0}_{\gamma_b}
    =\frac{1}{\sqrt{1+2\sin \frac{\phi}{2}}} \qty(\ket{0}+\sqrt{2\sin \frac{\phi}{2}} \ket{1})_{C_b} \otimes  \bigotimes_{\gamma=0}^{\Gamma} \ket{+}_{\gamma_b,} 
\end{equation}
where all $H_{\gamma_b}$ are Hadamard gate and $\ket{+}=\frac{1}{\sqrt{2}}(\ket{0+\ket{1}})$. 
Note that  $-(1-e^{-\ii\phi})=-\ii e^{-\frac{\ii\phi}{2}}2 \sin \frac{\phi}{2}$ and we have assumed $\phi\in[0,\pi/2]$.
Then, a single qubit phase gate $U_{\phi}=\op{0}+\op{1}e^{\ii(-\frac{\phi}{2}+\frac{3\pi}{2})}$ is implemented on the ancillary qubit $C_b$, and the ancillary register state (for $\{C_b,\gamma_b\}$) becomes
\begin{equation}
   \frac{1}{\sqrt{1+2\sin \frac{\phi}{2}}}
   \qty(\ket{0}-\ii e^{-\ii\frac{\phi}{2}}\sqrt{2\sin \frac{\phi}{2} }\ket{1})_{C_b} \otimes  \bigotimes_{\gamma=0}^{\Gamma} \ket{+}_{\gamma_b} . 
\end{equation}
Now we define the effective select($U_R$) for $\dro$ as 
\begin{equation}
    \select(U_R)=\op{0}_{C_b}\otimes I +\op{1}_{C_b}\otimes \bigotimes_{\gamma=1}^{\Gamma} (\op{0}_{\gamma_b}\otimes I_{\gamma} +\op{1}_{\gamma_b} \otimes R_\gamma ).
\end{equation}
We have the $\dro$ as follows,
\begin{equation}
   \qty(G_{C_b}\otimes \bigotimes_{\gamma=1}^{\Gamma} H_{\gamma_b})^\dagger   \select(U_R)  U_\phi 
   \qty(G_{C_b}\otimes \bigotimes_{\gamma=0}^{\Gamma} H_{\gamma_b}) 
   =\frac{\ket{0}_{C_b,\gamma_b} \bra{0}_{C_\gamma,\gamma_b}\otimes R}{1+2\sin\frac{\phi}{2}} + \mathcal{R} ,
\end{equation}
where $\ket{0}_{C_b,\gamma_b}=\ket{0}_{C_{b}}\otimes \bigotimes_{\gamma=0}^{\Gamma}\ket{0}_{\gamma_b}$ and $\mathcal{R} $ is orthogonal to the subspace span $\{ \op{0}_{C_b,\gamma_b}  \otimes \mathcal{H}_{tar}\}$. The success probability of $R$ can be seen to depend on the target phase $\phi$. 
According to oblivious amplitude amplification, one needs to query the whole unitary $\bigO\qty(\sqrt{1+2\sin \frac{\phi}{2}})$ times to guarantee that the success probability of $R$ is 1. 
Remarkably, our scheme can be generalized to the distributed phase operation of any component of interest, i.e. $I-(1-e^{-\ii\phi})\op{j}$.


\emph{Quantum communication complexity of d-RO.}\label{qcc_dRO}
In general, the query complexity  of implementing $\dro$ ($R$) is $\bigO(\sqrt{1+2\sin \frac{\phi}{2}})= \bigO(1)$ since $1\le 1+2\sin \frac{\phi}{2} \le 3$ ($\phi\in[0,\pi/2]$). 
In the above process of performing $\dro$, only the control register $C_{b}$ is needed to be distributed into the whole network ($\Gamma$ nodes). 
That is, for an $\bigO(1)$ query of the oracle of $\dro$, the quantum communication complexity is $\bigO(\Gamma)$.  
Without loss of generality, $\ket{i}_{C_b}$ can be defined as $\ket{i}_{C_b}=\ket{i}\otimes \ket{i}^{\otimes \Gamma}$, where each qubit in $\ket{i}^{\otimes \Gamma}$ should be individually sent to $\Gamma$ nodes.
In this setting, the quantum communication complexity for $\dro$ is still $\bigO(\Gamma)$ but with lower circuit depth.

{\subsection{  Communication lower bounds for bipartite d-RO}
\label{sec:dro_lower_bound}

Having established the $\Omega(\log \abs{\calE_{ir}})$ quantum communication lower bound for the d-BE, it is physically instructive to similarly bound the communication complexity of d-RO. 
Similarly, we formally define the problem of d-RO.

\begin{definition}[Distributed Reflection Problem in Bipartite Networks.]\label{def:distributed_ro}
Consider an arbitrary quantum network partitioned by a bipartite cut into two disjoint regions, $A$ and $B$, with any intermediate topology formally absorbed into either partition. A distributed reflection protocol utilizes local operations and classical communication (LOCC) alongside inter-partition quantum communication to synthesize a global bipartite unitary $R_{\text{total}}$. This operator acts on the localized ancillary registers $\calH_{a_A} \otimes \calH_{a_B}$ and must act trivially as the identity on the joint data register $\calH_A \otimes \calH_B$. The protocol is defined as  valid if it identically implements the reflection over the global initial ancillary state $\ket{\mathbf{0}}_{anc} = \ket{0}_{a_A} \otimes \ket{0}_{a_B}$, satisfying the exact algebraic condition:
\begin{equation}
    R_{\text{total}} = \big(I_{anc} - 2\ket{\mathbf{0}}_{anc}\bra{\mathbf{0}}_{anc}\big) \otimes I_{AB},
\end{equation}
where $I_{anc}$ and $I_{AB}$ are the identity operators on the full ancillary and data spaces, respectively. The quantum communication complexity of this task is defined as the absolute minimum quantum communication cost $Q$—whether physically realized via direct qubit transmission or the equivalent consumption of pre-shared ebits via quantum teleportation—required over all such valid protocols.
\end{definition}
Here, we provide a lower bound on the quantum communication cost for d-RO for a bipartite network.

\vspace{2mm}

\begin{corollary}
    [Communication Lower Bound for Distributed Reflection Operation.]
To exactly synthesize the distributed reflection operator $R$ across a nontrivial bipartite cut separating nodes Alice ($A$) and Bob ($B$), any distributed quantum protocol requires a minimum inter-node quantum communication cost of  $Q \ge \Omega(1)$ qubits.\label{lowerbound_dro}
\end{corollary}

\begin{proof}
Across the bipartite cut, the global reflection operator decomposes algebraically as:
\begin{equation} \label{eq:reflection_decomposition}
    R_{AB} = I_A \otimes I_B - 2\ket{0}\bra{0}_{A} \otimes \ket{0}\bra{0}_{B}.
\end{equation}
To rigorously determine the OSR, we must evaluate this expression within an orthogonal operator basis. We introduce the subspace orthogonal projectors $P_0 = \ket{0}\bra{0}$ and $P_\perp = I - \ket{0}\bra{0}$, which satisfy the Hilbert-Schmidt orthogonality condition $\text{Tr}(P_0 P_\perp) = 0$. Substituting the identity resolution $I = P_0 + P_\perp$ into Eq.~(\ref{eq:reflection_decomposition}) yields:
\begin{align} \label{eq:R_orthogonal_expansion}
    R_{AB} &= (P_0 + P_\perp)_A \otimes (P_0 + P_\perp)_B - 2 P_{0,A} \otimes P_{0,B} \nonumber \\
    &= -P_0 \otimes P_0 + P_0 \otimes P_\perp + P_\perp \otimes P_0 + P_\perp \otimes P_\perp.
\end{align}

Equation~(\ref{eq:R_orthogonal_expansion}) constitutes a precise expansion in the orthogonal tensor basis $\{P_0, P_\perp\}_A \otimes \{P_0, P_\perp\}_B$. The expansion coefficients naturally form a $2 \times 2$ bipartite correlation matrix $M$:
\begin{equation}
    M = \begin{pmatrix} -1 & 1 \\ 1 & 1 \end{pmatrix}.
\end{equation}
The exact Operator Schmidt Rank of the global reflection corresponds directly to the algebraic rank of this correlation matrix. Because the determinant evaluates to $\det(M) = -2 \neq 0$, the matrix is full-rank. Consequently, the OSR of the global reflection is fundamentally constrained to a constant:
\begin{equation}
    \text{OSR}(R_{AB}) = \text{Rank}(M) = 2.
\end{equation}

Applying the operator Schmidt rank lower bound for quantum communication, the minimum required transmitted qubits across the cut is exactly bounded by:
\begin{equation}
    Q\ge  \frac{1}{2}\log_2(\text{OSR}(R_{AB})) = \frac{1}{2}\log_2(2) = \Omega(1).
\end{equation}
Thus, our d-RO protocol is optimal in quantum communication complexity for bipartite networks.
\end{proof}

}

\subsection{Distributed truncated Taylor series (d-TS)}\label{apd:dts}
In this section, we would like to extend the truncated Taylor series (TS) \cite{berrySimulatingHamiltonianDynamics2015} to its distributed version (denoted $\dts$). 
Suppose $\norm{e^{-\ii Ht/r}-\sum_{k=0}^{K} \frac{(-\ii)^k(t/r)^k}{k!} H^k }\le \epsilon/r$. 
Given $H=\sum_l H_l=\sum_{\gamma=1}^\Gamma H_\gamma + \sum_{e\in \calE} H_e$,
we denote the norm of a Hamiltonian term by $\alpha_l:=\norm{H_l}$ and $\alpha:=\sum_l \alpha_l=\norm{H}_1$ where $\norm{H}_1=\sum_j \norm{H_j}$.
For accuracy $\epsilon/r$, the required order of Taylor expansion  must be $K=\bigO\qty(\frac{\log(\alpha t/\epsilon)}{\log\log(\alpha t/\epsilon)})$ \cite{berrySimulatingHamiltonianDynamics2015}. 
Expanding $e^{\ii Ht/r} \approx \sum_{j=1}^{J}\beta_j \bigotimes_{1}^{\Gamma} U^{(\gamma)}_j$ based on truncated TS and one can implement quantum simulation of truncated TS using d-LCU with algorithm error $\epsilon/r$.

Here we generalize the unary-encoding setting of TS method \cite{berrySimulatingHamiltonianDynamics2015} to distributed case. The idea for unary-encoding for TS is intuitive, which essentially prepares an auxiliary system applying control operations allowing the linear combinations different $H^k$ with coefficient $\frac{(\ii t)^k}{k!}$. 
We present the quantum circuit for $\dts$ with unary encoding in \cref{fig:d-TS_circuit}.
Specifically, the overall ancillary systems prepare 
\begin{equation}
    B_{\unary}\ket{\mathbf{0}}
    = G_{\TS}\otimes \bigotimes_{i=1}^{K}G_{(i)} \ket{\mathbf{0}}
    = \frac{1}{\sqrt{\alpha_{\TS}}}\sum^{K}_{k=0} \sqrt{\frac{(\alpha t/r)^k}{k!}}\ket{1^k 0^{K-k}} \otimes \bigotimes_{i=1}^{K}  \sum^{L}_{l=1} \frac{1}{\sqrt{\alpha}}\sqrt{\norm{H_l}}\ket{l}_i,
    \end{equation}
where $\alpha_{\TS}=\sum_{k=0}^{K} \frac{(\alpha t/r)^k}{k!}$ and $\alpha=\sum_l^{L}\norm{H_l}$. 
Note that here $G_{(i)}=G_C\bigotimes_{\gamma=1}^{\Gamma}G_\gamma$ is the gate to encode all coefficients of $H$ 
which is exactly equal to $G$ in \cref{apd:dbe}.
The $\select(U_{\unary})$ is given as 
\begin{equation}
    \select (U_{\unary}) =
    \bigotimes_{i=1}^{K}  (\op{0}_i\otimes I -\ii\op{1}_i\otimes \select(U_i) ),
\end{equation}
where $\select(U_i)=\sum_l^{L} \op{\mathbf{j}}_i \otimes \frac{H_l}{\norm{H_l}}_s$ (which is equivalent to \cref{d-BE-sU}).
Now we define 
$ W_{\unary}:= B^{\dagger}_{\unary}\otimes  I_s (\select(U_{\unary})) B_{\unary}\otimes I_s$
and then 
\begin{equation}
    W_{\unary}\ket{\mathbf{0}}\ket{\psi}_s=\frac{1}{\alpha_{\TS}} \ket{\mathbf{0}} \otimes \qty(\sum_{k=0}^{K} \frac{(-\ii)^k(t/r)^k}{k!} {H}^k \ket{\psi})+\ket{\phi},
\end{equation}
where $H=\sum_l^{L} H_l$. 
For oblivious amplitude amplification \cite{berryHamiltonianSimulationNearly2015},  one may set $\alpha_{\TS}=e^{t/r\alpha}=2$. 
That is $t/r=\ln 2 /\alpha$.
Specifically, one can approximate $\ket{\mathbf{0}} \otimes e^{-\ii Ht/r} ( \bigotimes_{\gamma=1}^\Gamma\ket{\psi_j})$ by
\begin{equation}
    -W_{\unary}RW_{\unary}^\dagger RW_{\unary}\ket{\mathbf{0}}  \otimes ( \bigotimes_{\gamma=1}^\Gamma|\psi_j\rangle ) 
    =\ket{\mathbf{0}}  \otimes \sum_{k=0}^{K} \frac{(-\ii)^k(t/r)^k}{k!} {H}^k ( \bigotimes_{\gamma=1}^\Gamma|\psi_j\rangle )
\end{equation}
where $\ket{\psi (t=0)}_s=\bigotimes_{\gamma=1}^\Gamma\ket{\psi_j}$ is the product state. 
By repeating the above process, the (distributed) quantum dynamics $U \approx (e^{-\ii Ht/r})^r=e^{-\ii Ht}$ can be implemented.


\begin{figure*}[!t]
    \centering
    \sidesubfloat[]{\includegraphics[width=.45\linewidth]{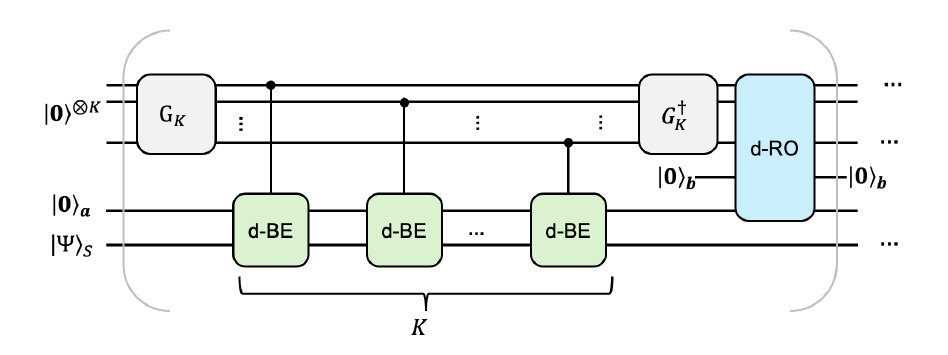}}
    \sidesubfloat[]{\includegraphics[width=.45\linewidth]{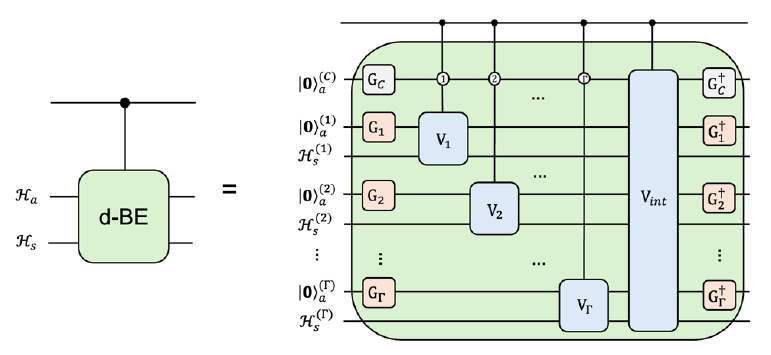}}
    \caption{
    (a) The oracle operator for $\dts$ circuit. 
    $G_{\TS}$ is used to prepare $\frac{1}{\sqrt{\alpha_{\TS}}}\sum^{K}_{k=0} \sqrt{\frac{(\alpha t/r)^k}{k!}}\ket{1^k 0^{K-k}} $, which is needed to be transmitted to the entire network.
    (b) The quantum circuit for controlled $\dbe$. 
    Here $\mathcal{H}_{\anc}$ and $\mathcal{H}_{\sys}$ represent all ancillary systems $\mathcal{H}_{\anc}^{(C)}\otimes \mathcal{H}_{\anc}^{(1)}\otimes\mathcal{H}_{\anc}^{(2)}\dots \otimes\mathcal{H}_{\anc}^{(\Gamma)}$ and the systems of interest $\mathcal{H}_{\sys}^{(1)}\otimes\mathcal{H}_{\sys}^{(2)}\dots \otimes\mathcal{H}_{\sys}^{(\Gamma)}$. }
    \label{fig:d-TS_circuit}
\end{figure*}


\begin{theorem}[Quantum communication complexity of $\dts$, Restatement Theorem 2 in maintext]\label{apd:thm:d_ts}
    Given a $\Gamma n$-qubit Hamiltonian 
    $H=\sum_{\gamma=1}^\Gamma H_\gamma+\sum_{e\in\calE} H_e$
    and evolution time $t$,
    the quantum communication complexity of the distributed TS protocol (\dts) for quantum simulation with accuracy $\eps$
    is 
   \begin{equation}
       \bigO\qty( \alpha t \Gamma \log (\abs{\calE}+\Gamma)  \frac{\log(\alpha t/\epsilon)}{\log\log(\alpha t/\epsilon)} ),
    \end{equation}
    where $\Gamma$ is the number of nodes, $\abs{\calE}$ is the number of the interaction terms, 
    and $\alpha= \norm{H}_1$.
\end{theorem}
\begin{proof}[Proof of \cref{apd:thm:d_ts}]
    Given an $\Gamma n$-qubit Hamiltonian 
    $H=\sum_{\gamma=1}^\Gamma H_\gamma+\sum_{e\in\calE} H_e $
    and evolution time $t$, 
    the quantum communication complexity of $\dts$ for quantum simulation 
    is the query complexity multiply the quantum communication complexity of each oracle (query). 
    As we shown in \cref{fig:d-TS_circuit},
    a single query consists $K$ controlled $\dbe$ of Hamiltonian $H$ and a single $\dro$ for $\pi$ phase. Based on our analysis before, the quantum communication complexity for $\dbe$ and $\dro$ are $\bigO(\Gamma \log(\abs{\calE}+\Gamma))$ and $\bigO(\Gamma)$ respectively. 
    Compared to $\dbe$, here we need to consider the quantum communication complexity of controlled $\dbe$. 
    The quantum circuit of controlled $\dbe$ is presented in \cref{fig:d-TS_circuit} (b). 
    One can see that the introduced control qubit needs to be sent to each node, which in turn performs the corresponding control operation, and hence it introduces an additional quantum communication cost of $\Gamma$ qubits. Thus the quantum communication complexity for controlled $\dbe$ is also $\bigO(\Gamma \log(\abs{\calE}+\Gamma))$. 
    Overall, the quantum communication complexity of $\dts$ is
    $\bigO\qty( \alpha t (1 +\log (\abs{\calE}+\Gamma)) \Gamma K )=\bigO\qty( \alpha t \log (\abs{\calE}+\Gamma) \Gamma K )$
    where $\Gamma$ is the number of nodes, $\abs{\calE}$ is the number of the interaction terms, 
    $\alpha:=\norm{H}_1$ and $K=\bigO\qty(\frac{\log(\alpha t/\epsilon)}{\log\log(\alpha t/\epsilon)})$. Note that if $\abs{\calE}=0$, this suggests that no interaction terms are crossing over a quantum network, and the communication cost should be $0$ since one can implement the quantum simulation locally. In this case, distributed quantum simulation, i.e. $\dts$, is no longer necessary.
\end{proof}

\subsection{Distributed quantum signal processing (d-QSP)}\label{qsp}

Quantum Signal Processing (QSP), sometimes it may be called Qubitization is a powerful technique that allows one to perform a function of high dimensional matrices \cite{lowOptimalHamiltonianSimulation2017,lowHamiltonianSimulationQubitization2019}. 
The key idea of qubitization is to divide the high-dimensional block encoding unitary into two-dimensional subspace, which can be considered an effective qubit. 

QSP can implement a wide range of polynomial functions, what we are interested in here is how it can be applied to quantum (Hamltonian) simulations. 
Specifically, QSP utilizes the Jacobi-Anger expansion to implement
$e^{-\ii Ht}=\sum_{k=-\infty}^{\infty}\ii^kJ_k(t)e^{\ii k \cos^{-1}H}$, by querying multiple $H$ and phase gates, 
where $J_k(t)$ represents the Bessel function of the first kind.
By QSP, 
the query complexity needed to approximate $e^{-\ii Ht}$ within error $\epsilon$ is
$\bigO\qty(\alpha t + \frac{\log(1/\epsilon)}{\log\log(1/\epsilon)})$. Compared to the Taylor truncated method, QSP may have a lower query complexity. 
One may find more details in Ref \cite{lowHamiltonianSimulationQubitization2019}. 
Here we generalize QSP to a quantum network (denoted $\dqsp$).

The QSP algorithm consists of the block-encoding of Hamiltonian and the reflection and phase operations acting on the auxiliary system \cite{lowHamiltonianSimulationQubitization2019}. Similarly, as shown in \cref{fig:d-QSP}, $\dqsp$ shall consist of $\dbe$ and corresponding reflection acting on the auxiliary system acting on the auxiliary system. 
This construction of $\dqsp$ is simple and it follows the same query setup as the non-distributed one. 
The difference between them mainly stems from the partition of Hamiltonian in the network, i.e., the interaction graphs which in turn make the design of $\dbe$ of $H$ and $\dro$ different from the original one. 
By the following Lemma of QSP proposed by Low and Chuang \cite{lowHamiltonianSimulationQubitization2019},
we can obtain the quantum communication complexity of $\dqsp$. 
\begin{lemma}[Hamiltonian Simulation by Qubitization \cite{lowHamiltonianSimulationQubitization2019}] 
    Given access to the oracles of $\langle \mathbf{0}|_a (G^\dagger \otimes I)\; \textup{\select} (U) \; (G \otimes I)|\vb{0}\rangle_a = \frac{H}{\alpha}$ specifying a Hamiltonian $H=\sum_i^d \alpha_i U_i$ where $\alpha=\sum_i \abs{\alpha_i}$, time evolution by $H$ can be simulated for time $t$ and error $\epsilon$ with $\bigO(\alpha t+ \log(1/\epsilon))$ queries.
\end{lemma}
\begin{theorem}[Quantum communication complexity of \dqsp, Theorem 3 in maintext]\label{apd:thm:d_qsp}
    Given a $\Gamma n$-qubit Hamiltonian $H=\sum_{\gamma=1}^\Gamma H_\gamma+\sum_{e\in\calE} H_e$ and evolution time $t$,
    the quantum communication complexity of the distributed QSP protocol (\dqsp) for quantum simulation with accuracy $\eps$ is
    \begin{equation}
        \bigO (\Gamma \log( \abs{\calE} + \Gamma)(\alpha t+ \log(1/\epsilon)),
    \end{equation}
    where  $\Gamma$ is the number of nodes,  $\abs{\calE}$ is the number of the interaction terms, 
    and $\alpha=\norm{H}_1$. 
\end{theorem}
\begin{proof}[Proof of \cref{apd:thm:d_qsp}]
As we analyzed before, for $H=\sum_{\gamma=1}^\Gamma H_\gamma+\sum_{e\in\calE} H_e $,
the quantum communication complexity of the oracles of $\dbe$ and $\dro$ are $\bigO(\Gamma \log(\abs{\calE}+\Gamma))$ (see \cref{qcc_dBE}) and $\bigO(\Gamma)$ (see \cref{qcc_dRO}) respectively. 
Therefore, for a single query of $\dqsp$, the quantum communication complexity is still $\bigO(\Gamma \log(\abs{\calE}+\Gamma))$.
The quantum communication complexity of $\dqsp$ should be the query complexity multiplied by the quantum communication complexity of each $\dbe$ and $\dro$, i.e., $ \bigO (\Gamma \log (\abs{\calE}+\Gamma)(\alpha t+ \log(1/\epsilon)) $. If $\abs{\calE}=0$, one can implement the quantum simulation locally. In this case, distributed quantum simulation, i.e. $\dqsp$, is no longer necessary and the communication cost will be zero. 
\end{proof}
Similar to $\dts$, the quantum communication complexity of $\dqsp$ is linear with $t$ and proportional to the number of nodes $\Gamma$ and the number of interaction terms across nodes $\abs{\calE}$. 

\begin{figure*}[!t]
    \centering
    \includegraphics[width=.65\linewidth]{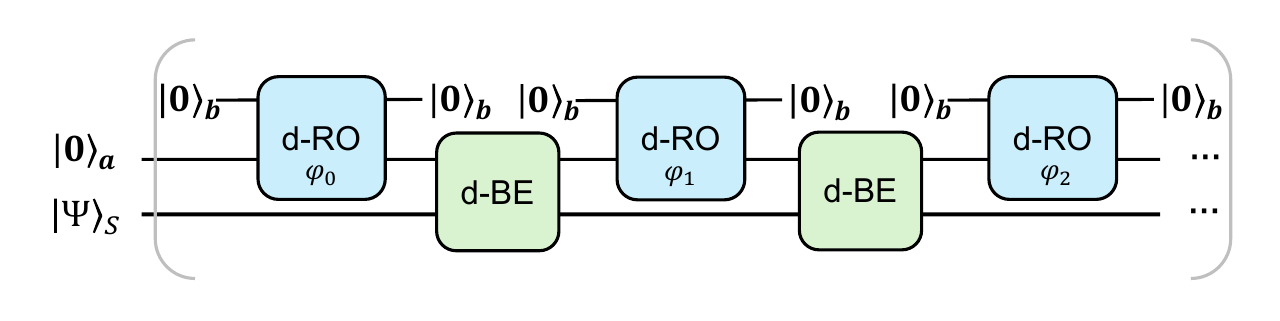}
    \caption{The quantum circuit for $\dqsp$. Here $\dbe$ denotes the oracle of block encoding of the target matrix, i.e. the Hamiltonian $H$. The $\dro$ may act as a specific phase on the component $\ket{\mathbf{0}}_a$ of (all) ancillary system $a$.}
    \label{fig:d-QSP}
\end{figure*}



\section{Applications in distributed quantum simulation}\label{apd:applications}
In this section, we analyze the quantum communication complexity of different distributed quantum simulation protocols for some common physical systems.
Beyond quantum simulation, we show that the distributed (quantum network) protocols can be applied to other important quantum algorithms, such as quantum phase estimation and Grover search, with low quantum communication.

\subsection{Lattice nearest-neighbor Hamiltonians}\label{apd:nn}
We begin with the one-dimensional (1D) lattice spin Hamiltonian with the nearest-neighbor (NN) interaction
$H_{\mathrm{NN}}=\sum_{j=1}^{n-1}H_{j,j+1}$
where $H_{j,j+1}$ is a Hermitian operator that acts nontrivially only on qubits $j$ and $j + 1$.

\begin{lemma}[Product formula for NN Hamiltonians \cite{childsNearlyOptimalLattice2019, childsTheoryTrotterError2021}]
    Let $H$ be a Hamiltonian with $n$ qubits.
    Let $\pfU_p$ be a $p$th-order product formula.  
    The number of Trotter steps for simulation accuracy $\eps$ is $r=\bigO(n^{1/p}t^{1+1/p}/\eps^{1/p})$.
    \label{1dNN}
\end{lemma}


The most straightforward method for quantum simulation of an NN Hamiltonian in a quantum network is to implement the PF distributively as in \cref{apd:d_pf}.
Firstly, we regroup the ($\Gamma n$-qubit) NN Hamiltonian according to its interaction (the terms across nodes) and network topology of $\Gamma$ nodes, i.e.,
$ H = H_0+H_{\calE}$
where $H_0=\sum_{\gamma=1}^\Gamma H_\gamma$ and $H_{\calE}=\sum_{e\in\calE} H_e$. 
The terms $\supp(H_{\gamma})\subseteq N_\gamma$ can be implemented (e.g. by the high-order PFs) on the $\gamma$th node locally to arbitrary precision without communication.
The interaction terms $H_e$ represent the terms whose support is across two neighboring nodes. 
Note that given a $\Gamma$-partite quantum network to simulate a 1D NN Hamiltonian with $\Gamma n$ qubits, the number of interaction terms is $\abs{\calE}=\bigO(\Gamma)$. 
Therefore, we have the following proposition.
\begin{proposition}[$\dpf$ for NN Hamiltonians]\label{thm:nn_pf}
     Given a $\Gamma$-partite quantum network,
    let $H$ be a one-dimensional lattice Hamiltonian with $\Gamma n$ qubits,
    and $\pfU_p$ be a $p$th-order product formula, 
    the number of Trotter steps (round of quantum communications) for simulation accuracy $\eps$ is $r=\bigO((\Gamma )^{1/p}t^{1+1/p}/\eps^{1/p})$.
    Then, the (total) quantum communication complexity of the d-PF protocol for NN Hamiltonians is $\bigO(\Gamma(\Gamma )^{1/p}t^{1+1/p}/\eps^{1/p})$.
\end{proposition}
The proof of this proposition is intuitive. 
Childs and Su \cite{childsNearlyOptimalLattice2019} made use of even-odd ordering to rewrite the Hamiltonian as $H=H_{\even}+H_{\odd}$ 
in which each term inside $H_{\even}$ or $H_{\odd}$ is pairwise commutative.
For 1D NN Hamiltonian implemented by a star network, the clustered Hamiltonian can also be written as $H_0=H_{\even}$ and $H_{\calE}=H_{\odd}$. 
Since $[H_\gamma, H_\gamma^\prime]=0$ for $\gamma \ne \gamma^\prime$ and one may package the interaction terms between the $\gamma$-th node and the $\gamma+1$-th node as $H_{\gamma,\gamma+1}$ so that $H_{\calE}=\sum_{\gamma=1}^{\Gamma-1} H_{\gamma,\gamma+1}$. 
Thus it suffices to take $r = O(\Gamma^{1/p} t^{1+1/p}/\epsilon^{1/p})$ to bound the Trotter error (we have replaced $n$ in \cref{1dNN} with $\Gamma$). 
Similarly, for general case where $\abs{\calE}=\bigO(\Gamma)$,
the complexity of post-Trotter methods is as follows.
\begin{proposition}[$\dts$ for NN Hamiltonians]\label{thm:nn_lcu}
    Given a $\Gamma$-partite quantum network,
    let $H$ be the one-dimensional lattices Hamiltonian with $\Gamma n$ qubits, the quantum communication complexity of the d-TS protocol for NN Hamiltonians is
    $\bigO\qty( \Gamma \alpha t  \log(\Gamma) \frac{\log \frac{\alpha t}{\epsilon}}{\log\log\frac{\alpha t}{\epsilon}})$
    where $\epsilon$ is the simulation accuracy and $\alpha:=\norm{H}_1$.
\end{proposition}

\begin{proposition}[$\dqsp$ for NN Hamiltonians]\label{thm:nn_qsp}
    Given a $\Gamma$-partite quantum network,
    let $H$ be the one-dimensional lattices Hamiltonian with $\Gamma n$ qubits, the quantum communication complexity of the d-QSP protocol for NN Hamiltonians is
    $ \bigO \qty(\Gamma\log(\Gamma) (\alpha t+ \log(1/\epsilon)))$
    where $\epsilon$ is the simulation accuracy 
and $\alpha:=\norm{H}_1$.
\end{proposition}

\begin{table*}[t!]
    \centering
    \begin{tabular}{c|c|c|c}
    \hline\hline
    Models $\backslash$ methods & $p$th-order $\dpf$ & $\dts$ & $\dqsp$\\
    \hline
    General $H$ & $\bigO(\Gamma \talpha_{\cmm}^{1/p}\abs{\calE} \, \olive{t^{1+1/p}/\eps^{1/p}})$ & $\bigO\qty(\log (\abs{\calE}+\Gamma) \violet{\Gamma \alpha t \frac{\log(\alpha t/\epsilon)}{\log\log(\alpha t/\epsilon)} })$ & $\bigO (\log( \abs{\calE} + \Gamma) \teal{\Gamma (\alpha t+ \log(1/\epsilon)})$ \\
    \hline
    $k$-local  & $\bigO\qty(\min(k,\Gamma)(\Gamma n)^k\vertiii{H}_1\alpha^{1/p}  \olive{t^{1+1/p}/\eps^{1/p}})$ & $\bigO\qty(k\log (\Gamma n)\violet{\Gamma \alpha t \frac{\log(\alpha t/\epsilon)}{\log\log(\alpha t/\epsilon)}})$ &  $\bigO\qty(  k \log(\Gamma n)\teal{\Gamma (\alpha t+ \log(1/\epsilon)}) $  \\
    1D NN Lattice & $\bigO({\Gamma }^{1+1/p} \, \olive{t^{1+1/p}/\eps^{1/p}})$ & $\bigO( \log (\Gamma)\violet{\Gamma \alpha t \frac{\log(\alpha t/\epsilon)}{\log\log(\alpha t/\epsilon)} })$ & $\bigO( \log (\Gamma) \teal{\Gamma (\alpha t+ \log(1/\epsilon)})$   \\
    \hline\hline
    \end{tabular}
    \caption{Quantum communication complexity of three quantum simulation protocols 
    (i.e., distributed product formula (\dpf), truncated Taylor series method (\dts), quantum signal processing (\dqsp))
    for different types of $\Gamma n$-qubit physical systems ($\abs{\calE}\ne 0$).
     $\tilde{\alpha}_{\cmm,p}$ is the norm of the nested commutator of terms in $H$ defined in the main text.
    $\Gamma$ is the number of nodes, $\alpha=\norm{H}_1$ and $\abs{\calE}$ is the number of the induced interaction terms that act on multiple nodes. $\vertiii{H}_1=\max_l\max_{{j}_l}\sum_{\substack{{j}_1,\ldots,{j}_{l-1},{j}_{l+1},\ldots,{j}_{k}}}\norm{H_{{j}_1,\ldots,{j}_k}}$ is the induced $1$-norm \cite{childsTheoryTrotterError2021}. 
    For the $k$-local case, we replace $\Gamma$ with $\min(k,\Gamma)$ inside $\dpf$, since the support of the interaction terms under this model is at most $k$. 
    The 1D lattice nearest-neighbor (NN) Hamiltonian is a special case of $k$-local model, i.e. $k=2$.
    We color the same scaling to make it easier for readers.
    } 
    \label{tab:methods}
\end{table*}

\subsection{k-local Hamiltonians}\label{apd:k_local}
The $k$-local Hamiltonian is a common model in the context of quantum computation and simulation.
Each term of a $k$-local Hamiltonian $H_k = \sum_{j_1,\dots,j_k} H_{j_1,\dots,j_k}$ acts non-trivially on at most $k$ qubits
where each term $H_{j_1,\dots,j_k}$ has support $\qty{j_1,\dots,j_k}$.
We first introduce a lemma of the product formula for $k$-local Hamiltonians.
\begin{lemma}[PF for $k$-local Hamiltonians \cite{childsTheoryTrotterError2021}]\label{thm:k_local}
    Let $H$ be a $k$-local Hamiltonian with $n$ qubits.
    Then, the nested commutator of a $p$th-order product formula is
    \begin{equation}
        \tilde{\alpha}_{\cmm} = \sum_{\gamma_1,\gamma_2,\dots,\gamma_{p+1}}
        \norm{\qty[H_{\gamma_{p+1}},\dots, [H_{\gamma_{2}},H_{\gamma_1}]]} = 
        \bigO\qty(\vertiii{H}_1^p \norm{H}_1), 
    \end{equation}
    where $\vertiii{H}_1^p=\max_l\max_{j_l} \sum_{j_1,\dots,j_{l-1},j_{l+1},\dots,j_{k}}\norm{H_{j_1,\dots,j_k}}$ is the induced 1-norm.
    The number of Trotter steps for achieving simulation accuracy $\eps$ is $r=\bigO(\vertiii{H}_1^p\norm{H}_1^{1/p}t^{1+1/p}/\eps^{1/p})$.
\end{lemma}
With the Trotter steps for $k$-local Hamiltonians, we have the following quantum communication complexity of $\dpf$.
\begin{proposition}[$\dpf$ for $k$-local Hamiltonians]\label{thm:d-k_local}
     Given a $\Gamma$-partite quantum network, and let $H=H_0+ \sum_{e\in\calE} H_e=\sum_{e=0}^{\abs{\calE}} H_e$,
    (where $H_0=\sum_{\gamma=1}^\Gamma H_\gamma$ and $\abs{\calE}$ is the number of the induced edge terms) be a $k$-local Hamiltonian with $\Gamma n$ qubits. 
    Let $\pfU_p$ be a $p$th-order product formula. 
    Then, the nested commutator
    \begin{equation}
        \tilde{\alpha}_{\cmm} = \sum_{e_1,e_2,\dots,e_{p+1}=0}^{\abs{\calE}}
        \norm{\qty[H_{e_{p+1}},\dots, [H_{e_{2}},H_{e_1}]]} = 
        \bigO(\vertiii{H}_1^p \norm{H}_1), 
    \end{equation}
    where $\vertiii{H}_1^p=\max_l\max_{j_l} \sum_{j_1,\dots,j_{l-1},j_{l+1},\dots,j_{k}}\norm{H_{j_1,\dots,j_k}}$ is the induced 1-norm.
    The number of Trotter steps (round of quantum communications) for simulation accuracy $\eps$ is $r=\bigO(\vertiii{H}_1^p\norm{H}_1^{1/p}t^{1+1/p}/\eps^{1/p})$. 
    Then, the (total) quantum communication complexity of the d-PF protocol for $k$-local Hamiltonian is
    \begin{equation}
        \bigO(\min(k,\Gamma) (\Gamma n)^k \vertiii{H}_1^p \norm{H}_1^{1/p} t^{1+1/p} \epsilon^{1/p}). 
    \end{equation}
\end{proposition}
\begin{proof}
    In each segment, $\abs{\calE}$ gates for interaction crossing $k$ nodes must be implemented with quantum communication. 
    Suppose each node has $n$ qubits, then the number of interaction terms for global $\Gamma n$-qubit $k$-local Hamiltonian is
    $\abs{\calE}_{k} =\bigO((\Gamma n)^k)$.
    In general, for non-distributed cases, the number of terms for $\Gamma n$-qubit with $k$-local Hamiltonian is also $\bigO((\Gamma n)^k)$ \cite{childsTheoryTrotterError2021}. 
    Considering the distributed case, subtracting the term of each node local we still have $\bigO((\Gamma n)^k)-\bigO(\Gamma n^k)=\bigO((\Gamma n)^k)$.
    More specifically, for 2-local case ($\Gamma \ge k=2$), $\abs{\calE}_{k}
    =c\binom{\Gamma}{2}\binom{n}{1}^2=c\frac{\Gamma(\Gamma-1)}{2}n^2$, where $c$ is a constant. 
    For $\Gamma \ge k >2$, the $\abs{\calE}$ is $\bigO((\Gamma n)^k)$.
Since for $\dpf$, the segment number is $r=\bigO(\vertiii{H}_1^p\norm{H}_1^{1/p}t^{1+1/p}/\eps^{1/p})$, and for each interaction term (even if $k$ is greater than $\Gamma$), it takes at most $\bigO(\min(k,\Gamma))$ communication to implement the corresponding exponential $ e^{-\ii\delta t \norm{H_{e_m}} \otimes _{j=1}^{\Gamma}\qty(\otimes_{i\in v_j}\sigma_{z_i}) }$. 
Therefore, the total quantum communication complexity of the $\dpf$ protocol for $k$-local Hamiltonian is
$\bigO(\min(k,\Gamma) (\Gamma n)^k \vertiii{H}_1^p \norm{H}_1^{1/p} t^{1+1/p})$.

\end{proof}

For general case of $k$-local Hamiltonian which has $ \abs{\calE}_k =\bigO((\Gamma n)^k)$,
the quantum communication complexity of post-Trotter methods ($\dts$ and $\dqsp$) is directly given below.
\begin{proposition}[$\dts$ for $k$-local Hamiltonians]\label{thm:k_local_lcu}
    Consider a $k$-local Hamiltonian $H$ with $\Gamma n$ qubits in a $\Gamma$-partite quantum network, time $t$, and simulation accuracy $\eps$.
    The total quantum communication complexity of d-TS is 
    $\bigO\qty( \alpha t k \Gamma \log (\Gamma n) \frac{\log(\alpha t/\epsilon)}{\log\log(\alpha t/\epsilon)} )$
    where $\alpha=\norm{H}_1$. 
\end{proposition}


\begin{proposition}[$\dqsp$ for $k$-local Hamiltonians]\label{thm:k_local_qsp}   
    Consider a $k$-local Hamiltonian $H$ with $\Gamma n$ qubits in a $\Gamma$-partite quantum network, evolution time $t$, and simulation accuracy $\eps$.
    The total quantum communication complexity of d-QSP is 
    $ \bigO\qty(k  \Gamma \ln(\Gamma n)  (\alpha t+ \log(1/\epsilon)) )$
    where $\alpha=\norm{H}_1$. 
\end{proposition}

\section{Extension to other distributed quantum algorithms}
\subsection{Distributed quantum phase estimation (d-QPE)}\label{apd:qpe}




 Quantum phase estimation (QPE) is a core procedure of many algorithms (e.g., Shor's algorithm \cite{shorPolynomialTimeAlgorithmsPrime1997}), which makes it significant in distributed quantum computing \cite{neumannImperfectDistributedQuantum2020}. Our distributed protocols can be simply extended for distributed quantum phase estimation. Recall that quantum phase estimation has two registers, the first register is the control system and the phase information may be read out by measuring on a Fourier basis. While the second register is the system being used to execute the target unitary $U$.  


In our distributed QPE (denoted $\dqpe$), different powers of $U$ need to be implemented on a composite system of $\Gamma$ nodes. 
In the quantum phase estimation algorithm, $K$ is a constant number and
assuming $U=\sum_{j=1}^{J} \beta_j \bigotimes^{\Gamma}_{\gamma=1} U^{(\gamma)}_j$, $U^k$ can also be implemented directly with $\dlcu$ (see \cref{d-LCU}). 
As shown in \cref{fig:qpe_grover} (a), one may need to introduce additional $K$ qubits at the control node to implement QPE. 
Recall that for $\dlcu$,  one may perform $U=\sum_{j=1}^{J} \beta_j \bigotimes^{\Gamma}_{\gamma=1} U^{(\gamma)}_j$ with quantum communication complexity  $\bigO(\sqrt{\sum_j \abs{\beta_j}}\Gamma\text{log}J )$. Clearly, for $U^{2^K}$, the communication complexity is $\bigO(2^K\sqrt{\sum_j \abs{\beta_j}}\Gamma\log J)$. 
However, for QPE, controlled $U$ is executed in the network, so an analysis of the communication complexity it may introduce is required.

\begin{figure*}[!t]
    \centering
    \sidesubfloat[]{\includegraphics[width=.45\linewidth]{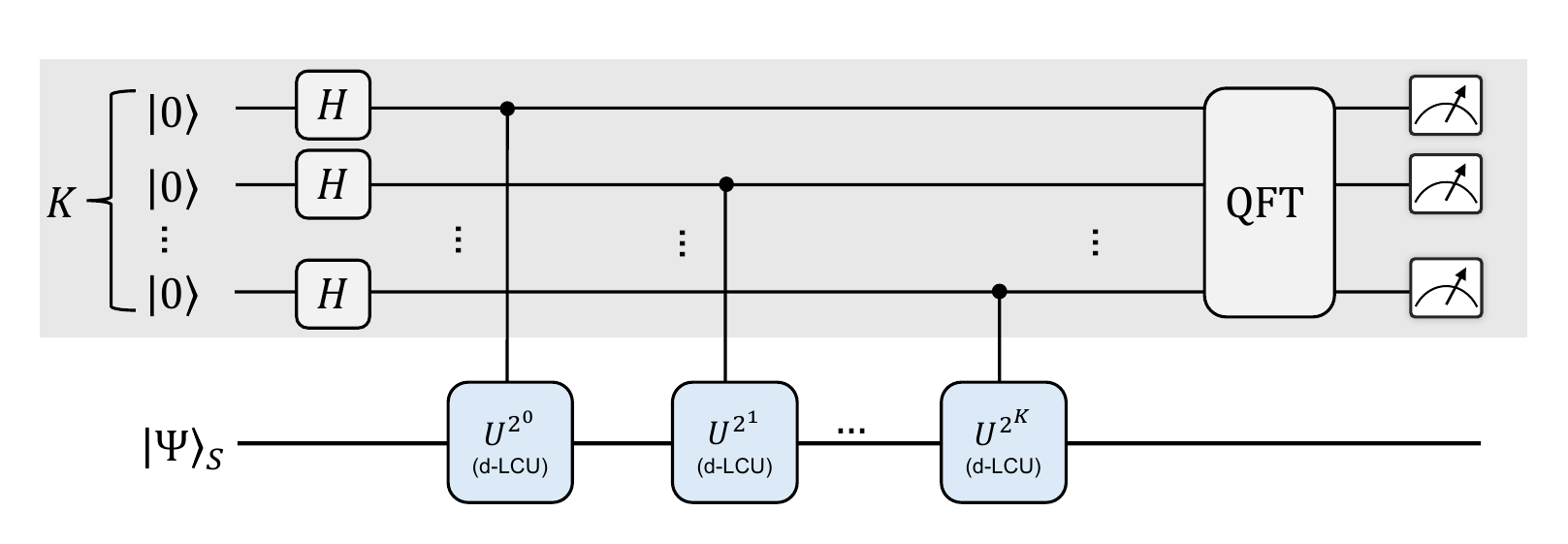}}
    \sidesubfloat[]{\includegraphics[width=.45\linewidth]{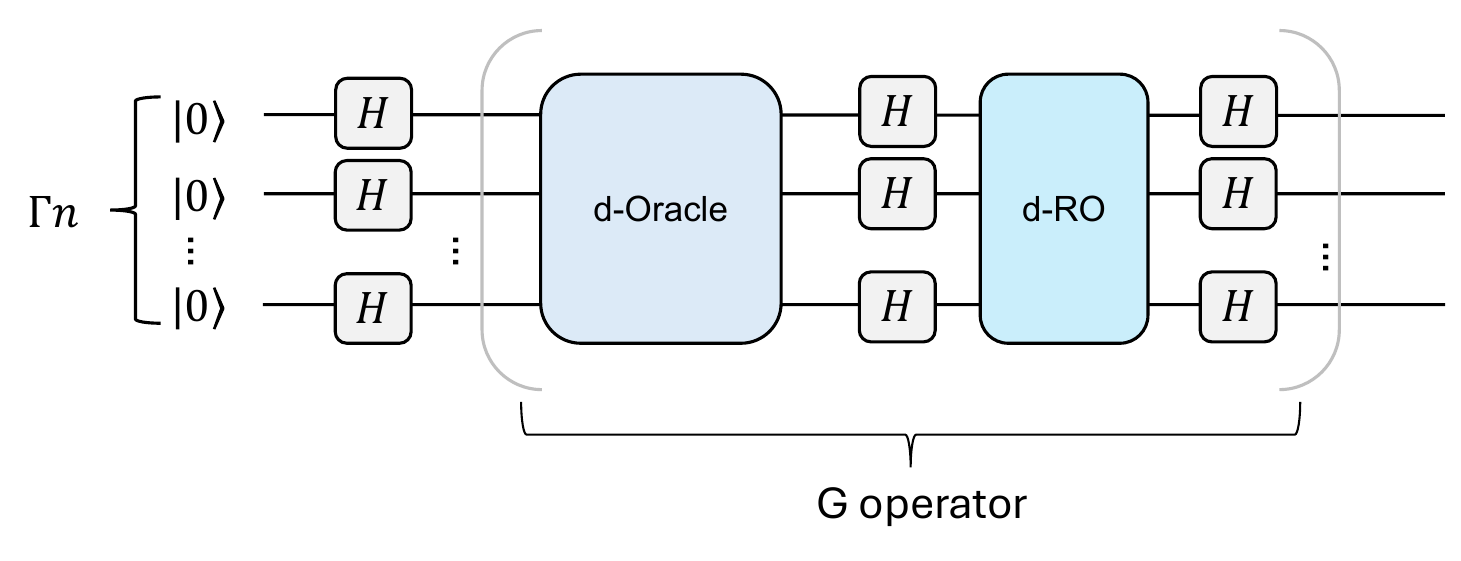}}
    \caption{(a) Distributed quantum phase estimation. 
    There are $K$ ancillary qubits at the control node and $\ket{\psi_S}$ is one of the eigenstates of $U$, which is pre-distributed to a $\Gamma$-partite quantum network. 
    (b) Distributed Grover's algorithm based on the original Grover's algorithm. 
    There are $\Gamma n$ qubits of data distributed in $\Gamma$ nodes. 
    Here the Oracle and reflection operations in Grover's algorithm are substituted with their distributed counterparts, d-Oracle and d-RO, respectively. }
    \label{fig:qpe_grover}
\end{figure*}

Similar to controlled $\dbe$, here the communication complexity of a single controlled $\dlcu$ of $U$ introduced an additional (little) quantum communication overhead of $\Gamma$ qubits since the control qubit needs to be sent to each node, which in turn performs the corresponding control operations. 
Thus, the quantum communication complexity of controlled $\dlcu$ of $U$ is still $\bigO(\sqrt{\sum_j \abs{\beta_j}}\Gamma\log J )$. 
That is for a single controlled $U^{2^K}$, the quantum communication complexity is $\bigO(2^K\sqrt{\sum_j \abs{\beta_j}}\Gamma\log J) $. 
Overall, one has the following proposition for general unitaries.
\begin{proposition}[$\dqpe$ for general $U$]\label{thm:qpe}   
    Given a $\Gamma$-partite quantum network and $U=\sum_{j=1}^{J} \beta_j \bigotimes^{\Gamma}_{\gamma=1} U^{(\gamma)}_j$, 
    the total quantum communication complexity of $\dqpe$ is
    $ \bigO\qty(2^{2K}\sqrt{\sum_j \abs{\beta_j}}\Gamma\log J)$,
    where $\Gamma$ is the number of nodes and  $K$ is the number of decimals in the phase estimate.
\end{proposition}
Similarly, we can consider the case where $U=e^{-\ii Ht}$ is implemented by $\dqsp$ for a Hamiltonian $H$ of $\Gamma n$ qubits.
\begin{proposition}[$\dqpe$ for dynamics of $H$]\label{thm:qpe_dynamics}   
    Given a $\Gamma$-partite quantum network, a Hamiltonian in the form $H=\sum_{\gamma=1}^\Gamma H_\gamma+\sum_{e\in\calE} H_e$ of $\Gamma n$ qubits ($\abs{\calE}>0$), and a unitary $\tilde{U}$ satisfying $\norm{\tilde{U}-e^{-\ii Ht}}\le \eps$, 
    the total quantum communication complexity of $\dqpe$ for $\tilde{U}$ is
    $\bigO\qty(2^{2K}\Gamma \log (\abs{\calE}+\Gamma)(\alpha t+ \log(1/\epsilon)))$,
    where $\Gamma$ is the number of nodes, $\abs{\calE}$ is the number of interaction terms, $\alpha=\norm{H}_1$ and $K$ is the number of decimals in the phase estimate. 
    The total error of $\dqpe$ for the ideal $U=e^{-\ii Ht}$ is $2^{2^k}\epsilon$.
\end{proposition}




\subsection{Distributed Grover's algorithm (d-Grover)}\label{apd:gg}
Grover's algorithm is a quantum algorithm for searching an unsorted database with a speed-up over classical algorithms. 
While classical computation requires linear time to search for a specific item in a database, Grover's algorithm leverages the power of quantum computation to perform the task in square root time complexity, i.e., $\bigO(\sqrt{N})$, where $N$ represents the number of items in the database.
We note that there are different versions of distributed Grover's algorithm. 
Here our protocol (denoted $\dgrover$) refers specifically to the direct extension of the original Grover's algorithm where the Oracle and reflection operations are substituted with their distributed counterparts. 
We present the quantum circuit for $\dgrover$ in \cref{fig:qpe_grover} (b).

It can be seen that the core of Grover's algorithm is the execution of Oracle and reflection operations. 
The latter reflection operation we have already mentioned above can be generalized to $\dro$. 
While Oracle can be considered as generalized reflection operations acting on a specific target component, i.e.  $I-(1-e^{-\ii\phi})\op{j}$. 
Thus in a quantum network with $\Gamma$ nodes, it is straightforward to construct distributed Oracle (d-Oracle) based on $\dro$. 
Recall that, in \cref{apd:d-RO}, we have shown the quantum communication complexity for a general $\dro$ ($I-(1-e^{-\ii\phi})\op{j}$) is $\bigO(\Gamma)$. 
Thus, for a single query of $G$ operator (contains a d-Oracle and a $\dro$), its quantum communication complexity is also $\bigO(\Gamma)$. 
Therefore, the total quantum communication complexity of distributed Grover's search should be $\bigO(\Gamma \sqrt{N})$.


\begin{proposition}[$\dgrover$ in a $\Gamma$-partite network]\label{thm:grover}   
    Given a $\Gamma$-partite quantum network, 
    the total quantum communication complexity for $\dgrover$ is $\bigO\qty(\Gamma \sqrt{N})$,
    where $\Gamma$ is the number of nodes and  $N=2^{\Gamma n}$ represents the number of items in the database.
\end{proposition}

We note that this performance is compatible with Ref. \cite{llovo2025network}.

\section{Generality of distributed protocols in quantum network with any topology}\label{apd:networks}

Here we show that for distributed protocols that require the assistance of control nodes, such as $\dlcu$, $\dts$, $\dqsp$, $\dqpe$, and $\dgrover$ and other such algorithms, their communication complexity is independent of the topology of the realistic quantum network.
This conclusion is predicated on the fact that each node possesses $\log(n)$ auxiliary qubits in addition to $n$ working qubits, and all auxiliary qubits of the whole network can be virtually and equivalently constructed as the centralized control node mentioned in the main text, through quantum communication.

As shown in \cref{fig:networks}, there are two $\Gamma$-node quantum networks ($\Gamma=6$) presented in (a) and (b) with different topologies. 
Each node has a primary quantum processor ($n$ qubits) and a small ancillary system ($\bigO (\log(n))$ qubits), and any local operation can be implemented on the composite system of these two systems. 
The light blue box represents the primary quantum processor, while the small gray box represents the auxiliary quantum system. 

\begin{figure*}[!t]
    \centering
    \includegraphics[width=.9\linewidth]{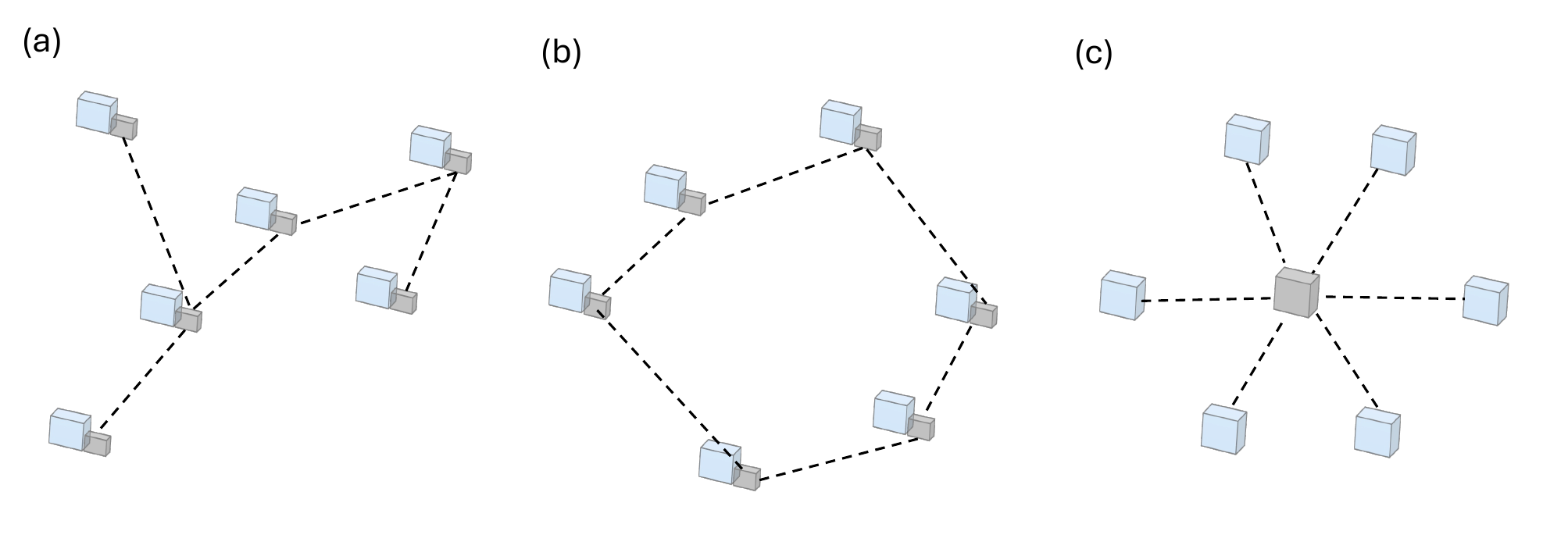}
    \caption{Quantum network with various topologies. Each node has a primary quantum processor and a small ancillary system, and any local operation can be implemented on the composite system of these two systems. The light blue box represents the primary quantum processor, while the small gray box represents the auxiliary quantum system. The dashed line indicates a quantum connection, i.e., a quantum channel. (a) and (b) can be virtually and equivalently constructed as (c), the Star network, through the quantum connection between ancillary systems. }
    \label{fig:networks}
\end{figure*}

As the auxiliary systems (small gray boxes) are connected via quantum channels, entanglement can be established between them via quantum communication. 
For example, if one of the nodes named node 1 has an auxiliary system preparing $G\ket{\mathbf{0}}=\sum^J_{j=1} \sqrt{\beta_j} \ket{j}^{(1)}\ket{j}^{(1)}$, then it may transfer half of the qubits of state to the neighboring node 2 via quantum teleportation of $\log_2 J$ qubits, at this point, node 1 and the neighboring node 2 shares the entangled state $\sum_{j=1}^J \sqrt{\beta_j} \ket{j}^{(1)}\ket{j}^{(2)}$. 
For the information of the state of interest to be shared across the network, node 2 performs multiple local CNOT gates on the qubits on hand and other auxiliary qubits, thus preparing $\sum_{j=1}^J \sqrt{\beta_j} \ket{j}^{(1)}\ket{j}^{(2)}\ket{j}^{(2)}$.
Similarly, node 2 teleports half of these quantum states to node 3. And so on, eventually the entire network will share auxiliary states such as $\sum_{j=1}^J \sqrt{\beta_j} \ket{j}^{(1)}\ket{j}^{(2)}\dots \ket{j}^{(\Gamma)}$, consistent with those prepared by central node in \cref{fig:networks} (c) which is exactly we introduced in $\dbe$, $\dro$ and $\dlcu$.

Consequently, the above process can be considered equivalent to transforming a network of arbitrary topology into a virtual star network. 
Moreover, the quantum communication cost arises only from the two neighboring nodes, so the quantum communication complexity of the network constructing such a globally entangled auxiliary state is equal to the number of quantum connections (dashed lines), which is proportional to the number of nodes. Similarly, its reverse process can also be implemented step by step.
That is, the quantum communication complexity of an auxiliary system state sharing a full network is $\bigO(\Gamma\log_2 J )$, identical to the communication complexity required for the distribution of entanglement by the control nodes in a star network. 

\begin{lemma}{
Any quantum network with arbitrary topology can be transformed into a virtual star network through local operations and quantum communication between neighboring nodes. The total quantum communication cost for establishing and disentangling a globally entangled auxiliary state $\sum_{j=1}^J \sqrt{\beta_j} \ket{j}^{(1)}\ket{j}^{(2)}...\ket{j}^{(\Gamma)}$ scales linearly with the number of nodes, i.e., \(\mathcal{O}(\Gamma\log_2J )\), and is independent of the network's specific topology. Consequently, the quantum communication complexity of our distributed post-Trotter protocols depends only on the number of nodes $\Gamma$, not on the network structure.}
\end{lemma}

On the contrary, $\dpf$ generally requires the number of quantum teleportations or quantum state transfers to connect a certain number of two or more nodes each time it performs the evolution of certain interaction terms, such that its quantum communication complexity is somewhat dependent on the topology of the network.

\section{Gate complexity: monolithic vs. distributed protocols}\label{sec:holistic_complexity}

\subsection{Total gate complexity of d-PF  and standard monolithic PF}

In standard PF, the entire Hamiltonian $H = \sum_{l=1}^L H_l$ is simulated on a single quantum processor. For a $p$th order product formula, the operational error is governed by the global nested commutator norm $\alpha_{\text{global}, p}$, which involves all $L$ interaction terms across the entire system, i.e., $\alpha_{\text{global}, p} =\sum_{l_1,\dots,l_{p+1}=1}^{L} \norm{[H_{l_1},[H_{l_2},\dots,[H_{l_p},H_{l_{p+1}}]]]}$ \cite{childsTheoryTrotterError2021}. To bound the error by $\epsilon$, the required Trotter step count is $r_{\text{global}} = \mathcal{O}\big( t^{1+1/p} (\alpha_{\text{global}, p} / \epsilon)^{1/p} \big)$. Since each step executes all $L$ terms globally, the total gate complexity is  $\mathcal{O}(L \cdot r_{\text{global}})$.

Our distributed quantum simulation (\texttt{d-PF}) re-architects this scaling by partitioning the system into a clustered Hamiltonian:
\begin{equation}
    H = \sum_{\gamma=1}^\Gamma H_\gamma + \sum_{e\in\mathcal{E}} H_e.
\end{equation}
This decomposition naturally decouples the error generator into two distinct sources: the inter-node quantum communication error (driven by the interaction edges $H_e$) and the intra-node local synthesis error (driven by the local terms $H_\gamma$).

To achieve the target precision $\epsilon$, we allocate the error budget equally between communication and local operations: $\epsilon_{\text{qc}} = \frac{1}{2}\epsilon$, and $\sum_{\gamma=1}^\Gamma \epsilon_{\text{loc}}^{(\gamma)} \le \frac{1}{2}\epsilon$. Assuming a uniform error distribution among the $\Gamma$ computational nodes, the local precision requirement per node becomes highly stringent: $\epsilon_{\text{loc}}^{(\gamma)} = \frac{\epsilon}{2\Gamma}$.

The quantum communication step count $r_{\text{qc}}$ is dictated solely by the commutator norm across the network boundaries, defined as $\tilde{\alpha}_{\text{qc},p} := \sum_{e_1,\dots,e_{p+1}=0}^{|\mathcal{E}|} \|[H_{e_1},\dots,[H_{e_p}, H_{e_{p+1}}]]\|$ (where $H_0=\sum_{\gamma=1}^\Gamma H_\gamma$). This boundary norm generally scales with the interaction cut size $\mathcal{O}(\abs{\mathcal{E}})$, which is drastically smaller than the global norm ($\abs{\mathcal{E}} \ll L$). Accordingly, the required communication step count is $r_{\text{qc}} = \mathcal{O}\big( t^{1+1/p} (\tilde{\alpha}_{\text{qc}, p} / \epsilon)^{1/p} \big)$.

Simultaneously, within each segment $\delta t = t/r_{\text{qc}}$, each node $\gamma$ executes a local $p$th order PF to synthesize $e^{-iH_\gamma \delta t}$ up to the stringent local precision $\delta_\gamma = \epsilon_{\text{loc}}^{(\gamma)} / r_{\text{qc}}$. Crucially, the local Trotter error is governed by the local commutator norm $\alpha_{\text{loc}, p}^{(\gamma)}$, which depends exclusively on the local terms confined within node $\gamma$. Summing these local executions over all $r_{\text{qc}}$ outer segments elegantly cancels the outer scaling parameter, yielding a total local step count for node $\gamma$ of $r^{(\gamma)}_{\text{loc}} = \mathcal{O}\big( t^{1+1/p} (\alpha_{\text{loc}, p}^{(\gamma)} / \epsilon_{\text{loc}}^{(\gamma)})^{1/p} \big)$. 

For simplicity, let
\(a_\gamma=\alpha_{\rm loc,p}^{(\gamma)}\). The local gate cost is
$
    O\left(
    t^{1+1/p}
    \sum_{\gamma=1}^{\Gamma}
    L_\gamma
    \left(
    \frac{a_\gamma}{\epsilon_\gamma}
    \right)^{1/p}
    \right),
$
subject to \(\sum_\gamma \epsilon_\gamma\le \epsilon/2\). Optimizing over
\(\epsilon_\gamma\) by Lagrangian gives
$
    \epsilon_\gamma
    =
    \frac{\epsilon}{2}
    \frac{
    L_\gamma^{p/(p+1)}
    a_\gamma^{1/(p+1)}
    }{
    \sum_{\gamma=1}^{\Gamma}
    L_\gamma^{p/(p+1)}
    a_\gamma^{1/(p+1)}
    }.
$
The resulting optimized local cost is
\begin{equation}
    O\left(
    t^{1+1/p}
    \epsilon^{-1/p}
    \left[
    \sum_{\gamma=1}^{\Gamma}
    L_\gamma^{p/(p+1)}
    a_\gamma^{1/(p+1)}
    \right]^{(p+1)/p}
    \right).
\end{equation}
By Hölder's inequality,
\begin{equation}
    \left[
    \sum_{\gamma=1}^{\Gamma}
    L_\gamma^{p/(p+1)}
    a_\gamma^{1/(p+1)}
    \right]^{(p+1)/p}
    \le
    \left(\sum_{\gamma=1}^{\Gamma}L_\gamma\right)
    \left(\sum_{\gamma=1}^{\Gamma}a_\gamma\right)^{1/p}.
\end{equation}
Since \(\sum_\gamma a_\gamma\le \alpha_{\rm global,p}\), the local contribution is bounded by
$
    O\left(
    t^{1+1/p}
    (L-|\mathcal E|)
    \left(
    \frac{\alpha_{\rm global,p}}{\epsilon}
    \right)^{1/p}
    \right).
$
Combining the optimized local synthesis cost with the inter-node implementation cost gives
\begin{equation}
    \mathcal N_{\rm gate}(\texttt{d-PF})
    =
    O\left(
    t^{1+1/p}
    \left[
    \text{min}(k,\Gamma) |\mathcal E|
    \left(
    \frac{\tilde{\alpha}_{\rm qc,p}}{\epsilon}
    \right)^{1/p}
    +
    (L-|\mathcal E|)
    \left(
    \frac{\alpha_{\rm global,p}}{\epsilon}
    \right)^{1/p}
    \right]
    \right).
\end{equation}

Since
\(\tilde{\alpha}_{\rm qc,p}\lesssim \alpha_{\rm global,p}\), 
under standard locality assumptions ($k$-local model, where $k$ is a constant; in the standard monolithic model, PF considers $k$-local model generally and one cannot directly assume a single gate implementation for Pauli rotation of $n$-qubit), \texttt{d-PF}
matches the asymptotic total gate complexity of monolithic PF while reducing
the required local processor size from \(N\) qubits to \(n\) qubits per node.
if
\(\text{min}(k,\Gamma)|\mathcal E| =O(L)\), then
\begin{equation}
    \mathcal N_{\rm gate}(\texttt{d-PF})
    =
    O(L r_{\rm global}).\end{equation}
Thus, under standard constant-locality or sparse-cut assumptions,
\texttt{d-PF} matches the asymptotic total gate complexity of monolithic PF,
while reducing the local processor size from \(N=\Gamma n\) qubits to \(n\)
qubits per node.

We emphasize that this comparison uses the usual constant-locality
term-exponential accounting. In an elementary-gate model, where only one- and
two-qubit gates are primitive, a monolithic PF implementation of a \(k\)-local
Hamiltonian requires
$
O\!\left(r_{\rm global}\sum_{\ell=1}^{L} w_\ell\right)
=
O(kLr_{\rm global})
$
elementary gates, where \(w_\ell\le k\) is the weight of the \(\ell\)-th Pauli
term. In this accounting, the distributed overhead is also bounded by
\begin{equation}
\min(k,\Gamma)|\mathcal E|\, r_{\rm qc}
\le
kL\,r_{\rm qc}
=
O(kLr_{\rm global}),
\end{equation}
and hence \texttt{d-PF} again matches the monolithic elementary-gate scaling up
to constant factors.

\subsection{Total gate complexity for distributed post-Trotter protocols}
\label{sec:post_trotter_gate_complexity}

To provide a clear quantitative comparison, we first review the gate complexity of standard, centralized post-Trotter queries before bounding our distributed architecture. In our distributed post-Trotter protocols, the query complexity is the same as the centralized (non-distributed) one. Therefore, in the following, we focus on the gate complexity for per query for the distributed case and the non-distributed case, i.e., monolithic case.

\vspace{2mm}
In a monolithic quantum processor, a Hamiltonian with $L$ terms is simulated on $N$ qubits. Implementing block-encoding (BE) via the linear combination of unitaries (LCU) uses an auxiliary register of $\mathcal{O}(\log L)$ qubits. The \texttt{select} operation combines the $L$ terms using multi-controlled Toffoli gates across the auxiliary register. 
For a straightforward Pauli-LCU implementation with naive PREPARE/SELECT circuits, the per-query gate complexity is $\mathcal{O}(L(\log L + N))$ for a single BE execution \cite{berrySimulatingHamiltonianDynamics2015,lowHamiltonianSimulationQubitization2019}. The reflection operator (RO) performs a phase flip on the auxiliary register using $\mathcal{O}(\log L)$ gates. A single query of a centralized post-Trotter algorithm has a total gate complexity of:
\begin{equation} \label{eq:complexity_standard}
\mathcal{N}_{\text{gate}}(\text{Single Query}) = \mathcal{O}\Big(L(\log L + N)\Big),
\end{equation}
where $N = \Gamma n$ for comparison with the distributed case.

\vspace{2mm}
\noindent{(a) Gate Complexity of the Distributed Block-Encoding (\texttt{d-BE}).} 
The global Hamiltonian is partitioned into local terms $H_\gamma$ (each containing $L_\gamma$ terms) and inter-node interaction terms $H_e$,  satisfying $L = \sum_{\gamma=1}^{\Gamma} L_\gamma + |\mathcal{E}|$. The total gate complexity of \texttt{d-BE} is derived from three core physical components:
\textit{(i)} The block-encoding of the local terms is executed natively and independently across the $\Gamma$ nodes. For the $\gamma$-th node, multiplexing $L_\gamma$ terms on an $n$-qubit register requires $\mathcal{O}(L_\gamma(\log L_\gamma + n))$ gates. Defining $L_{\max} = \max_\gamma L_\gamma$, the gate complexity of all local BE is
\begin{equation}
\mathcal{N}_{\text{local\_BE}} = \mathcal{O}\Big( \sum_{\gamma=1}^{\Gamma} L_\gamma(\log L_{\max} + n) \Big).
\end{equation}
\textit{(ii)} To coherently synthesize the interaction terms $V_{int}$, the algorithm executes a global \texttt{select} operation over the $|\mathcal{E}|$ interaction edges. This is controlled by the distributed auxiliary register of size $\mathcal{O}(\log(|\mathcal{E}| + \Gamma))$. Multiplexing $|\mathcal{E}|$ distinct interaction branches utilizing this control register introduces a multi-controlled gate overhead of:
\begin{equation}
\mathcal{N}_{\text{int}} = \mathcal{O}\Big( |\mathcal{E}| \log(|\mathcal{E}| + \Gamma) \Big).
\end{equation}
\textit{(iii)} Distributing the auxiliary qubits across the network incurs a quantum communication gate complexity of $\mathcal{O}(\Gamma \log(|\mathcal{E}|+\Gamma))$.
Summing these three physical components yields the exact complete gate complexity for a single \texttt{d-BE} execution:
\begin{equation} \label{eq:complexity_dbe}
    \mathcal{N}_{\text{gate}}(\texttt{d-BE}) = \mathcal{O}\Big( \sum_{\gamma=1}^{\Gamma} L_{\gamma}(\log L_{\max} + n) + |\mathcal{E}|\log(|\mathcal{E}|+\Gamma) + \Gamma\log(|\mathcal{E}|+\Gamma) \Big).
\end{equation}

\vspace{2mm}
\noindent{(b) Complexity of the Distributed Reflection (\texttt{d-RO}).} 
The \texttt{d-RO} aims to perform a multi-controlled phase flip over the distributed auxiliary register. Because the control register size has been compressed to $\mathcal{O}(\log(|\mathcal{E}| + \Gamma))$, the reflection requires $\mathcal{O}(\log(|\mathcal{E}| + \Gamma))$ multi-controlled Z-gates on the control node. Besides, each node needs to implement a local RO gate; the gate cost is  $\mathcal{O}(\sum_{\gamma }(\log L_\gamma ))$ gates.
Inter-node synchronization requires a communication gate cost of $\mathcal{O}(\Gamma)$. Thus, the complete gate complexity for a single \texttt{d-RO} is:
\begin{equation} \label{eq:complexity_dro}
    \mathcal{N}_{\text{gate}}(\texttt{d-RO}) = \mathcal{O}\Big(\log(|\mathcal{E}| + \Gamma) + \sum_{\gamma }(\log L_\gamma )+\Gamma\Big).
\end{equation}

\vspace{2mm}
\noindent{(c)  Gate complexity of Distributed Post-Trotter Protocols with Single Querie.}
For the distributed Taylor Series (\texttt{d-TS}), synthesizing the (short) time evolution up to the truncation order $K$ necessitates a more intricate structural overhead. Using a unary encoding for the auxiliary system, the algorithm first prepares a $K$-qubit superposition state to encode the Taylor coefficients, requiring an initial state preparation cost of $\mathcal{O}(K)$ gates. Subsequently, it requires a sequential cascade of $K$ \textit{controlled} applications of the \texttt{d-BE} primitive, where each unary qubit controls one \texttt{d-BE} execution. 

Crucially, augmenting the multi-controlled \texttt{select} operations of the \texttt{d-BE} primitive with a single global control qubit only increases the Toffoli gate count by an $\mathcal{O}(1)$ additive overhead per multiplexed branch,  preserving the original asymptotic gate complexity bounds. Therefore, the $\mathcal{O}(K)$ state preparation cost is asymptotically absorbed, and the per-query gate complexity scales proportionally:
\begin{equation} \label{eq:complexity_dts_query}
    \mathcal{N}_{\text{gate}}(\text{\texttt{d-TS} Query}) = \mathcal{O}\Bigg( K \cdot \bigg[ \sum_{\gamma=1}^{\Gamma} L_{\gamma}(\log L_{\max} + n) + \big(\abs{\mathcal{E}}+\Gamma\big)\log\big(\abs{\mathcal{E}}+\Gamma\big) \bigg] \Bigg).
\end{equation}


For distributed Quantum Signal Processing (\texttt{d-QSP}), a single algorithmic step explicitly interleaves one \texttt{d-BE} and one \texttt{d-RO} oracle. The per-query complete gate complexity is:
\begin{equation}
    \mathcal{N}_{\text{gate}}(\text{\texttt{d-QSP} Query}) = \mathcal{O}\Big( \sum_{\gamma=1}^{\Gamma} L_{\gamma}(\log L_{\max} + n) + (|\mathcal{E}|+\Gamma)\log(|\mathcal{E}|+\Gamma) \Big).
\end{equation}

\vspace{2mm}

\noindent{(d) Proof of Overhead Reduction.} The distributed gate complexity can be compared to the standard monolithic case to evaluate resource savings. Using the decomposition $L = \sum_{\gamma=1}^\Gamma L_\gamma + |\mathcal{E}|$ and $N = \Gamma n$, the monolithic gate complexity from Eq.~(\ref{eq:complexity_standard}) expands to:\begin{equation}\mathcal{N}{\text{gate}}(\text{Monolithic}) = \mathcal{O}\Big( \sum_{\gamma=1}^{\Gamma} L_\gamma (\log L + \Gamma n) + |\mathcal{E}| (\log L + \Gamma n) \Big).\end{equation}Comparing this to the distributed complexity in Eq.~(\ref{eq:complexity_dbe}), we note that $L_{\max} < L$ and $n < \Gamma n$. The local computational volume is therefore bounded by:\begin{equation} \label{eq:bound_local}\sum_{\gamma=1}^{\Gamma} L_{\gamma}(\log L_{\max} + n) < \sum_{\gamma=1}^{\Gamma} L_\gamma (\log L + \Gamma n).\end{equation}For the interaction terms, assuming $\sum L_\gamma \ge \Gamma$, it follows that $L \ge |\mathcal{E}| + \Gamma$ and $\log L \ge \log(|\mathcal{E}|+\Gamma)$. The interaction control overhead is thus bounded as:\begin{equation} \label{eq:bound_int}|\mathcal{E}|\log(|\mathcal{E}|+\Gamma) \le |\mathcal{E}|\log L < |\mathcal{E}|(\log L + \Gamma n).\end{equation}Since $\sum L_\gamma$ scales with the system size, the reduction in gate count exceeds the communication overhead. This indicates that the \texttt{d-BE} architecture reduces the total gate cost across the network compared to the monolithic counterpart.

\section{Practical considerations for \dpf: observable simulation and weak coupling cases}
\label{app:practical_considerations}

While the main text establishes worst-case communication bounds based on the global operator norm and generic interaction strengths, practical physical simulations often operate in regimes where the communication complexity may be reduced. Here, we analyze two such physically relevant scenarios:  the observable simulation, the weak-coupling limit, and the entanglement feature of the quantum state.

\subsection{Communication complexity for observable simulation}
\label{app:weak_simulation_advantage}

The worst-case Trotter bounds typically scale dynamically with the extensive properties of the system. In our d-PF protocol, this worst-case state-vector communication bound is governed by the global nested commutator norm, defined as $\tilde{\alpha}_{\mathrm{qs},p} := \sum_{e_1,\dots,e_{p+1}=0}^{|\mathcal{E}|} \|[H_{e_1},\dots,[H_{e_p}, H_{e_{p+1}}]]\|$, which inherently scales with the  interaction cut size $\mathcal{O}(|\mathcal{E}|)$. As demonstrated in recent literature \cite{zhaoEntanglementAcceleratesQuantum2025}, this worst-case extensive scaling can be achieved even by separable product states due to the coherent accumulation of local error expectation values.

However, practical quantum simulation typically focuses on \textit{weak simulation}—estimating the expectation value of a target local observable $O$. Here, we rigorously prove that tracking the Heisenberg evolution of the observable may decouple the communication complexity from the global norm $\tilde{\alpha}_{\mathrm{qs},p}$. 

To formalize this advantage, we need to define an observable-dependent effective commutator norm that precisely captures the error susceptibility of the target observable during the evolution. 
 We characterize the single-step discrepancy by defining a Hermitian effective error generator $\Omega$, such that the simulated evolution is factorized as $\mathscr{U}_p(\delta t) = e^{-i\Omega} U(\delta t)$.
The effective error generator for a $p$-th order product formula can be explicitly expanded as a linear combination of nested commutators \cite{childsTheoryTrotterError2021}:
\begin{equation}
    \Omega = \delta t^{p+1} \sum_{l_1,\dots,l_{p+1}=0}^{|\mathcal{E}|} c_{l_1,\dots,l_{p+1}} \big[H_{l_{p+1}}, \dots, [H_{l_2}, H_{l_1}]\dots\big] + \mathcal{O}(\delta t^{p+2}),
\end{equation}
where $c_{l_1,\dots,l_{p+1}}$ are real coefficients uniquely determined by the specific sequence of the product formula. 

To analyze the spatial propagation of errors, we must geometrically anchor these nested commutators. We achieve this by grouping the terms according to their outermost operator index, $e = l_{p+1}$. This allows us to define the local effective error block rooted at the $e$-th interaction term:
\begin{equation}
    C_e^{(p)} := \sum_{l_1,\dots,l_p=0}^{|\mathcal{E}|} c_{l_1,\dots,l_p, e} \big[H_e, \dots, [H_{l_2}, H_{l_1}]\dots\big].
\end{equation}
Physically, while the inner $p$ summations expand the spatial support of the operator by a finite radius proportional to the Trotter order $p$, the outermost index $e$ fundamentally dictates the spatial epicenter of this localized error block.
This exact algebraic regrouping ensures that $\Omega = \delta t^{p+1} \sum_{e=0}^{|\mathcal{E}|} C_e^{(p)} + \mathcal{O}(\delta t^{p+2})$. Consequently, we define the observable-dependent norm as the supremum of the local commutations over the ideal Heisenberg evolution:
\begin{equation}
    \tilde{\alpha}_{\mathrm{obs},p} := \sup_{\tau \in [0,t]} \sum_{e=0}^{|\mathcal{E}|} \big\| [C_e^{(p)}, O(\tau)] \big\|,
\end{equation}
where $C_e^{(p)}$ aggregates all non-vanishing $p$-th order nested commutators rooted at the boundary edge $e \in \mathcal{E}$, and $O(\tau) = e^{iH\tau} O e^{-iH\tau}$ is the ideal Heisenberg evolution of the observable.
By explicitly exposing the outermost spatial anchor $e$, this definition perfectly bridges the algebraic Trotter error with the geometric Lieb-Robinson bound utilized in our subsequent analysis.


\begin{corollaryprl}[Quantum Communication Complexity for Local Observables]
\label{prop:observable_error}
Consider a $\Gamma$-node quantum network simulating a Hamiltonian evolution $U(t) = e^{-iHt}$ with $H=H_0+\sum_{e=0}^{\abs{\mathcal{E}}} H_e$. Let $O$ be a target local observable. Utilizing a $p$-th order distributed product formula (d-PF), the quantum communication complexity required to estimate the expectation value $\langle O(t) \rangle$ within an additive error $\epsilon$ is rigorously bounded by:
\begin{equation}
    \mathcal{O}\left( \Gamma |\mathcal{E}| \frac{t^{1+\frac{1}{p}} \big( \tilde{\alpha}_{\mathrm{obs},p}\big)^{\frac{1}{p}}}{\epsilon^{\frac{1}{p}}} \right),
\end{equation}
where $\Gamma$ is the number of nodes and $\abs{\calE}$ is the number of the interaction terms across multiple nodes. 
\end{corollaryprl}
This bound substitutes the global extensive parameter $\tilde{\alpha}_{\mathrm{comm},p}$ with the  observable-dependent parameter $\tilde{\alpha}_{\mathrm{obs},p}$.

\begin{proof}
Let $U = e^{-i H \delta t}$ be the ideal evolution for a single segment $\delta t = t/r$, and let $\mathscr{U}_p$ be the corresponding $p$-th order Trotter unitary. We define a Hermitian effective error generator $\Omega$ such that $\mathscr{U}_p = e^{-i\Omega} U$. By the Taylor expansion, $\Omega = \delta t^{p+1} \sum_{e =0}^{\abs{\mathcal{E}}} C_e^{(p)}+\bigO(\delta t^{p+2} )$ .

We evaluate the observable error in the Heisenberg picture via a telescoping sum. The final simulated observable is $\tilde{O}(t) = {\mathscr{U}_p^\dagger}^r  O \mathscr{U}^r_p$, and the ideal evolution is $O(t) = {U^\dagger}^r O U^r$. Utilizing the unitary invariance of the spectral norm, the total observable error is bounded by the accumulation of single-step errors:
\begin{equation}
    \epsilon_{\text{obs}} = \| \tilde{O}(t) - O(t) \| \le \sum_{m=1}^r \Big\| \mathscr{U}_p^\dagger O(\tau_m) \mathscr{U}_p - U^\dagger O(\tau_m) U \Big\|,
\end{equation}
where $\tau_m = (r-m)\delta t \le t$ is the remaining ideal evolution time. Substituting $\mathscr{U}_p = e^{-i\Omega} U$ and applying the first-order expansion of the adjoint action, each single-step error is dictated by the commutator:
\begin{align}
    \Big\| U^\dagger e^{i\Omega} O(\tau_m) e^{-i\Omega} U - U^\dagger O(\tau_m) U \Big\| &= \Big\| e^{i\Omega} O(\tau_m) e^{-i\Omega} - O(\tau_m) \Big\| \nonumber \\
    &\le \| [\Omega, O(\tau_m)] \|.
\end{align}

Expanding $\Omega$, the single-step error is rigorously bounded by our observable-dependent norm:
\begin{equation}
    \| [\Omega, O(\tau_m)] \| \le \delta t^{p+1} \sum_{e =0}^{ \mathcal{E}}  \big\| [C_e^{(p)}, O(\tau_m)] \big\| \le \delta t^{p+1} \tilde{\alpha}_{\mathrm{obs},p}.
\end{equation}

Summing this algebraic bound over all $r$ segments yields the total observable error: $\epsilon_{\text{obs}} \le r \delta t^{p+1} \tilde{\alpha}_{\mathrm{obs},p} = \frac{t^{p+1}}{r^p} \tilde{\alpha}_{\mathrm{obs},p}$. Equating this to the target accuracy $\epsilon$ directly determines the required number of Trotter steps:
\begin{equation}
    r = \mathcal{O}\left( \frac{t^{1+\frac{1}{p}} \big( \tilde{\alpha}_{\mathrm{obs},p} \big)^{\frac{1}{p}}}{\epsilon^{\frac{1}{p}}} \right).
\end{equation}
Multiplying the step count $r$ by the single-step communication overhead $\mathcal{O}(\Gamma |\mathcal{E}|)$ yields the final complexity, completing the proof.
\end{proof}

\textbf{Information of State and Observables: Toward a Tighter Bound via Scrambling.} While the preceding analysis based on the operator norm $\tilde{\alpha}_{\mathrm{obs},p}$ establishes the general advantage of weak simulation, the bound can be further tightened by incorporating the specific information of both the state $|\psi\rangle$ and the observable $O$. Following the scrambling-based framework in \cite{feng2025trotterizationoperatorscramblingentanglement}, we can refine the error metric by considering the expectation value of the commutator with the effective error generator $\Omega \propto \delta t^{p+1} \sum C_e^{(p)}$. We define a \textit{state-dependent scrambling norm}:
\begin{equation}\mathcal{S}p(\psi, O, t) := \max_k \sqrt{\bra{\psi{k}} \big| [O(k \delta t), \Omega] \big|^2 \ket{\psi_k}},\end{equation}
where $|\psi_k\rangle$ is the evolved state. This parameter quantifies the dynamical interference between the observable and the Trotter error generator within the specific Hilbert space manifold traversed by the system.By substituting $\tilde{\alpha}_{\mathrm{obs},p}$ with $\mathcal{S}_p(\psi, O, t)$ in Theorem \ref{prop:observable_error}, the required Trotter step count $r$ becomes:
\begin{equation}r = \mathcal{O}\left( \frac{\left( \mathcal{S}_p(\psi, O, t) \right)^{1/p} t^{1+1/p}}{\epsilon^{1/p}} \right).\end{equation}
Since $\mathcal{S}_p(\psi, O, t) \le \tilde{\alpha}_{\mathrm{obs},p}$ holds, this approach reveals that the communication speedup in our distributed framework is even more pronounced when the state $|\psi\rangle$ remains localized or within a low-scrambling regime.

\subsection{An explicit 2D distributed QIMF example}
\label{app:explicit_worst_case_bypass}

In the distributed product formula (d-PF) framework, the total quantum communication complexity is fundamentally determined by the product of the number of Trotter steps $r$ and the communication overhead per step, which scales with the interaction cut size $\mathcal{O}(|\mathcal{E}|)$. Consequently, any extensive scaling in the operational error directly inflates the required step count $r$, introducing more quantum communication.

To explicitly demonstrate how weak simulation can dynamically to be more efficient, we analyze the Quantum Ising Model with Transverse Field (QIMF) on a 2D lattice graph $G=(V, E_{\text{grid}})$ with:
$$H = h_x \sum_{j \in V} X_j + h_y \sum_{j \in V} Y_j + J \sum_{\langle j, k \rangle \in \mathcal{E}} X_j X_k,$$
partitioned across a distributed quantum network. Unlike non-distributed Trotterization, our d-PF decomposes the Hamiltonian  along the network boundaries: $H = H_0 + H_{\text{cut}}$. For example, consider a natural bipartite ($\Gamma=2$) architecture (where we denote node 1 as $L$ and node 2 as $R$), splitting the 2D grid. Here, $H_0= h_x \sum_{j \in V} X_j + h_y \sum_{j \in V} Y_j$ encapsulates all intra-node interactions within each chip, while $H_{\text{cut}} = \sum_{e \in \mathcal{E}} J X_{e_L} X_{e_R}$ contains exclusively the inter-node interactions across the network boundary of size $|\mathcal{E}| \propto \sqrt{N}$.

The leading-order error generator of PF1 for this distributed splitting is determined by the commutator $\Omega \propto \delta t^2 E$, where $E = [H_0, H_{\text{cut}}]$. Evaluating this commutator explicitly yields a sum of localized error terms supported exclusively on the boundary edges:
\begin{align}
    E = \Big[ h_x \sum X_j + h_y \sum Y_j + J \sum_{e \notin \mathcal{E}} X_j X_{k}, \sum_{e \in \mathcal{E}} J X_{e_L} X_{e_R} \Big] \nonumber = -2i J h_y \sum_{e \in \mathcal{E}} (Z_{e_L} X_{e_R} + X_{e_L} Z_{e_R}).
\end{align}

For state simulation, the error is dictated by $E^\dagger E$. As highlighted in \cite{zhaoEntanglementAcceleratesQuantum2025}, specific product states can force the Trotter errors into a worst-case deviation. 
Consider initializing the network in a tilted product state $|\psi\rangle = |\theta\rangle^{\otimes N}$, where $|\theta\rangle = \cos(\pi/8)|0\rangle + \sin(\pi/8)|1\rangle$. In this state, both $\langle Z \rangle$ and $\langle X \rangle$ are non-vanishing (evaluating to $1/\sqrt{2}$). When evaluating the global error, the cross-terms between different boundary edges perfectly interfere constructively:
\begin{equation}
    \langle \psi | E^\dagger E | \psi \rangle = 4J^2 h_y^2 (|\mathcal{E}|^2 + |\mathcal{E}|) = \Theta(|\mathcal{E}|^2).
\end{equation}
The resulting single-step state-vector error scales  extensively with the cut number: $\| E |\psi\rangle \| = \Theta(|\mathcal{E}|)$. To suppress this error accumulation, the required number of Trotter steps $r$ must inevitably grow with $|\mathcal{E}|$, severely penalizing the total communication complexity $r \times \mathcal{O}(|\mathcal{E}|)$ and rendering state simulation highly inefficient for large boundaries.

To rigorously understand why observable estimation bypasses this worst-case scaling, we analyze the operational error in the Heisenberg picture. Here we consider the target observable directly on the network cut. Let $O = X_{e_1, L}$ be the local observable of the left qubit residing on the very first boundary edge $e_1 \in \mathcal{E}$. The single-step Trotter error for first segment is quantified by the commutator between the error generator and the local observable:
\begin{align}
    \epsilon_{\text{step}} &= \| \tilde{O} - O_{\text{ideal}} \| \nonumber \\
    &= \big\| U^\dagger e^{i\Omega} O e^{-i\Omega} U - U^\dagger O U \big\| \nonumber \\
    &= \big\| e^{i\Omega} O e^{-i\Omega} - O \big\| \nonumber \\
    &\approx \big\| [\Omega, O] \big\| \propto \delta t^2 \big\| [E, O] \big\|=\delta t^2 \Big\| \sum_{e \in \mathcal{E}} [C_e, X_{e_1, L}] \Big\|.
\end{align}

Crucially, because local Pauli operators commute on disjoint supports ($[C_e, X_{e_1, L}] = 0$ for all $e \neq e_1$), the summation over the $|\mathcal{E}|$ boundary edges cancel instantaneously. 
Taking the operator norm, we obtain $\big\| [E, O] \big\| =\| [C_{e_1}, X_{e_1, L}]\|= \mathcal{O}(1)$.  This suggests
 even from the state simulation, the error has worst-case behavior $\Theta(|\mathcal{E}|)=\Theta(\sqrt{N})$, the local observable error remains  $\mathcal{O}(1)$ for the short-time evolution regime. In this case, the required Trotter step count $r_{\text{obs}}$  becomes decoupled from the dimensions of the cut.

\subsection{High-order Trotter analysis in the weak-coupling limit}
\label{app:weak_coupling}

In distributed quantum simulation, we partition the global Hamiltonian as $H = H_0 + \lambda H_I$, where $H_0 = \sum_{\gamma=1}^\Gamma H_\gamma$ represents the local terms executed without communication, and $H_I = \sum_{e \in \mathcal{E}} H_e$ represents the interaction graph across nodes with a coupling strength $\lambda \ll 1$.

For a $p$-th order distributed product formula (d-PF), the spectral norm of the simulation error is bounded by the nested commutators of the Hamiltonian components. Since local terms on different nodes commute ($[H_\gamma, H_{\gamma'}] = 0$ for $\gamma \neq \gamma'$), every non-vanishing nested commutator in the error expansion must contain at least one interaction term $\lambda H_I$. In the weak-coupling limit ($\lambda \to 0$), the dominant error terms are  linear in $\lambda$.

This scaling significantly mitigates the worst-case communication bounds presented in the main text. We formalize this adaptive behavior in the following proposition:

\begin{proposition}[Communication Complexity in the Weak-Coupling Cases]
\label{prop:weak_coupling}
Consider a $\Gamma$-node quantum network simulating the Hamiltonian $H = H_0 + \lambda H_I$ for time $t$ within an additive error $\epsilon$ and $\lambda\ll 1$. Utilizing the $p$-th order distributed product formula (d-PF), the quantum communication complexity is asymptotically bounded by:
\begin{equation}
    \mathcal{O}\left( |\mathcal{E}| \Gamma \left( \frac{\lambda t^{p+1}}{\epsilon} \right)^{\frac{1}{p}} \right).
\end{equation}
In particular, for the second-order symmetric d-PF ($p=2$), if the physical coupling strength satisfies the threshold condition:
\begin{equation}
    \lambda \le \frac{12 \epsilon}{t^3 \| [H_0, [H_0, H_I]] \|},
\end{equation}
the global evolution effectively factorizes into a single Trotter segment ($r=1$). Consequently, the quantum communication complexity becomes entirely independent of the evolution time $t$, reducing  to $\mathcal{O}(|\mathcal{E}|\Gamma)$.
\end{proposition}

\begin{proof}
For a $p$-th order product formula $\tilde{U}_p(t)$, the approximation error for $r$ segments is bounded by $\epsilon_p = \mathcal{O}( \frac{t^{p+1}}{r^p} \| [H_0, \dots [H_0, \lambda H_I]\dots] \| ) \propto \lambda \frac{t^{p+1}}{r^p}$ (we have remove the terms with $\bigO(\lambda ^2)$). Setting $\epsilon_p \le \epsilon$ gives the required number of Trotter steps $r = \mathcal{O}( ( \lambda t^{p+1} / \epsilon )^{1/p} )$. Since our d-PF protocol incurs a communication cost of $\mathcal{O}(|\mathcal{E}|\Gamma)$ per step, the total communication complexity follows directly as $\mathcal{O}( |\mathcal{E}|\Gamma \cdot r )$.

To prove the exact threshold for $p=2$, we employ the symmetric Suzuki-Trotter decomposition $\mathscr{U}_2(\delta t) = e^{-i H_0 \frac{\delta t}{2}} e^{-i \lambda H_I \delta t} e^{-i H_0 \frac{\delta t}{2}}$ with $\delta t = t/r$. By the Baker-Campbell-Hausdorff (BCH) expansion, the global error $\epsilon_2$ accumulated over $r$ steps is upper-bounded by:
\begin{equation}
    \epsilon_2 \le \frac{t^3}{12 r^2} \Big( \lambda \big\| [H_0, [H_0, H_I]] \big\| + \frac{\lambda^2}{2} \big\| [H_I, [H_I, H_0]] \big\| \Big).
\end{equation}
In the weak-coupling regime ($\lambda \ll 1$), the quadratic term vanishes rapidly, yielding $\epsilon_2 \approx \frac{\lambda t^3}{12 r^2} \| [H_0, [H_0, H_I]] \|$. Setting $r=1$, the condition to achieve the target accuracy $\epsilon_2 \le \epsilon$ exactly reshuffles to:
\begin{equation}
    \lambda \le \frac{12 \epsilon}{t^3 \| [H_0, [H_0, H_I]] \|}.
\end{equation}
When this threshold is met, the simulation protocol requires executing the interaction unitary $e^{-i \lambda H_I t}$ exactly once. The communication overhead is  dictated by the topology of the boundary interactions, yielding the absolute minimum cost of $\mathcal{O}(|\mathcal{E}|\Gamma)$, completing the proof.
\end{proof}

\subsection{ Distributed \dpf{} with random inputs and entanglement}


Recent analysis \cite{zhaoEntanglementAcceleratesQuantum2025} shows that the Trotter error of random input \cite{zhaoHamiltonianSimulationRandom2021, yuObservableDrivenSpeedupsQuantum2025,liPrincipalTrotterObservation2024} or highly entangled states is characterized by the Frobenius norm rather than the worst-case spectral norm. One may apply this state-dependent bound to the distributed framework to analyze the quantum communication overhead. 

To isolate the communication complexity, one can consider an idealized limit where the local dynamics $e^{-\ii H_\gamma t}$ are synthesized exactly. Under this zero-local-error assumption, the simulation error is determined entirely by the inter-node communication. Applying the entanglement-based bound to the partitioned Hamiltonian, the total simulation error for a time step $t$ on a global pure state $\ket{\psi}$ is bounded by
\begin{equation} \label{eq:ErrEnt_comm}
    \|(e^{-\ii Ht}-\mathscr{U}_{\text{d-PF}})\ket{\psi}\| \le \mathcal{O}\left( t^{p+1}\sqrt{\sum_{e,e' \in \mathcal{E}} \|E_e^{\dag}E_{e'}\| \sqrt{\log (d_{e,e'})-S(\rho_{e,e'})}} + t^{p+1} \frac{\|E_{\mathcal{E}}\|_F}{\sqrt{d}} \right),
\end{equation}
where $E_{\mathcal{E}} = \sum_{e \in \mathcal{E}} E_e$ is defined as the leading error operator from the interaction cut $\abs{\mathcal{E}}$, $S(\cdot)$ is the von Neumann entropy, and $\rho_{e,e'} := \tr_{[N]\setminus \text{supp}(E_e E_{e'})}(\ket{\psi}\bra{\psi})$ is the reduced density matrix on the boundary subsystem supporting the cross-node commutators $E_e E_{e'}$.

This bound provides a quantitative metric for how entanglement affects the required communication steps. When the global state $\ket{\psi}$ exhibits significant entanglement across the partition cuts, the boundary reduced density matrix $\rho_{e,e'}$ becomes highly mixed. As the local entropy $S(\rho_{e,e'})$ approaches $\log(d_{e,e'})$, the spectral deviation coefficient $\sqrt{\log (d_{e,e'})-S(\rho_{e,e'})}$ decreases. Consequently, the boundary error scaling reduces from the worst-case $\mathcal{O}(\abs{\mathcal{E}} t^{p+1})$ to the average-case Frobenius scaling $\mathcal{O}(\sqrt{\abs{\mathcal{E}}} t^{p+1})$. Since the required communication step count scales polynomially with this boundary error coefficient, the entanglement suppression directly decreases the number of inter-node routing operations needed to maintain a target precision $\epsilon$.

\section{Numerical comparison of the distributed quantum simulation protocols}\label{apd:numerics}

\subsection{Two typical Hamiltonians: the nearest-neighbor and $k$-local Hamiltonian}
To substantiate our theoretical bounds on communication complexity, we implement a numerical estimation and comparison of the distributed quantum simulation protocols.
Specifically, we consider two typical (i.e., the nearest-neighbor and the $k$-local) Hamiltonians that exhibit different complexity scalings of the protocols.
To compare the communication cost concretely, we need the explicit values rather than the big-O scaling in \cref{tab:methods}.
Therefore, we simply take the example $\Gamma=2$ (i.e., two connected QPUs) and the following Hamiltonian setting.

For the nearest-neighbor (NN) interaction Hamiltonian,
we take the Heisenberg Hamiltonian on an $n$-site one-dimensional lattice (with the periodic boundary condition) 
\begin{equation}\label{eq:lattice_hamiltonian}
    H_{\mathrm{NN}} = J\sum_{j=1}^{n} X_jX_{j+1} + Y_jY_{j+1} + Z_jZ_{j+1} + h \sum_{j=1}^n Z_j, 
\end{equation}
where $J$ is the coupling strength and $h$ is the external field strength.
The 1-norm of this Hamiltonian is $\alpha_{\mathrm{NN}}=\norm{H_\mathrm{NN}}_1=3Jn+hn$.

For the $k$-local Hamiltonian,
we take the following $k=2$ local, all-to-all interaction Hamiltonian
\begin{equation}\label{eq:local_hamiltonian}
    H_{2\mathrm{L}} = \sum_{i<j}^{n} J_{ij} X_i X_j + \sum_{j=1}^n h_j Z_j,
\end{equation}
where the subscript $i,j$ indicates the index of the qubit and the all-to-all interaction with strength $J_{ij}$. 
In this case, we ignore the commutator information between different Hamiltonian summands. 
For simplicity, we set the uniform coefficients $J_{ij} = J$ and $h_j = h$.
Then, the 1-norm of this Hamiltonian is $\alpha_{\mathrm{2L}}=\norm{H_\mathrm{2L}}_1=Jn(n-1)+hn$,
while the induced 1-norm is
$\vertiii{H_{\mathrm{2L}}}_1=\max_l\max_{{j}_l}\sum_{\substack{{j}_1,\ldots,{j}_{l-1},{j}_{l+1},\ldots,{j}_{k}}}\norm{H_{{j}_1,\ldots,{j}_k}}=J(n-1)+h$.

\subsection{Details of communication resource estimation}
The communication cost of the distributed product formula ($\dpf$) protocol is mainly proportional to the number of Trotter steps $r$ and the number of terms across QPUs.
Compared with other orders' product formulas, we take the fourth-order ($p=4$) product formula because it typically has the best performance, that is, requiring the least cost for the same simulation precision.
The number of Trotter steps is given by the theoretical bound $r=\bigO(\tilde{\alpha}_{\cmm,p}^{1/p}t^{1+1/p}/\eps^{1/p})$ 
where the nested commutator norm $\tilde{\alpha}_{\cmm,p}$ is as defined in \cref{apd:thm:d_pf} according to the induced cluster Hamiltonian form.
Then, the fourth-order nested-commutator norm $\tilde{\alpha}_{\cmm,4}$ of the NN Hamiltonian is independent of the system size 
(we bound it by the 1-norm of the nested commutators),
while the $k$-local one is bounded by $\vertiii{H}_1^4\norm{H}_1$ \cite{childsTheoryTrotterError2021}.
The overhead in each Trotter step is just $\Theta(\Gamma)$ for the NN Hamiltonian,
while the overhead is $\min(k,\Gamma)\times (\Gamma n)^k$ for the $k$-local Hamiltonian.
To demonstrate the dependence of $\dpf$ on precision, we set one typical simulation precision $\eps=10^{-3}$ and another high precision $\eps=10^{-5}$.

The communication cost of the truncated Taylor series ($\dts$) method is mainly proportional to the number of segments $r=\lceil T/\ln(2) \rceil$ and the truncated order $K$, that is $\Theta(r K)$,
where $T:=\norm{H}_1 t$ is the scaled evolution time.
The truncation order of the Taylor series $K(T,\eps)=\bigO\qty(\frac{\log(T/\eps)}{\log\log(T/\eps)})$ 
such that the truncation error is bounded by
\begin{equation}
    \norm{\tilde{U}-e^{-\ii Ht}}
    \le \sum_{k=K+1}^{\infty} \frac{\ln(2)^k}{k!}
    \le 2 \frac{(\ln 2)^{K+1}}{(K+1)!}.
\end{equation}
With $\delta:= 2 \frac{(\ln 2)^{K+1}}{(K+1)!}$, 
the $\dts$ algorithm achieve error at most $\sqrt{1+\xi}-\sqrt{1-\xi}$ with success probability at least $(1-\xi)^2$, 
where $\xi:=r \frac{\delta^2+3\delta+4}{2} \delta$ \cite[Theorem G.5]{childsFirstQuantumSimulation2018}.
The overhead of each segment (round of communication) is $\Theta(\lceil \log (\Gamma)\Gamma\rceil)$ for the NN Hamiltonian,
while $\Theta(k\lceil \log (\Gamma n)\Gamma\rceil)$ for the $k$-local Hamiltonian.


The communication cost of the distributed QSP ($\dqsp$) is mainly proportional to the  
the number of phased iterates $M=\bigO\qty(\norm{H}_1 t+ \frac{\log(1/\eps)}{\log\log(1/\eps)})$.
To estimate the number of phased iterates (query complexity) $M(t,\eps)$ for a certain evolution time $t$ and precision $\eps$, we rely on the explicit upper bound \cite[Lemma 59]{gilyenQuantumSingularValue2019}
\begin{equation}
    \forall q\in \mathbb{R}_+ , \;
    M(t, \eps) < e^q t + \frac{\ln(1/\eps)}{q}.
\end{equation}
Analogous to $\dts$, the overhead of $\dqsp$ in each segment is $\Theta(\lceil \log(\Gamma)\rceil \Gamma)$ for the NN Hamiltonian and $\Theta(k\lceil \log(\Gamma n)\rceil \Gamma)$ for the $k$-local Hamiltonian.


It is worth noting that,
unlike the $\dpf$, the post-Trotter protocols $\dts$ and $\dqsp$ are probabilistic.  
If a measurement of the ancilla register indicates a failure, the simulation needs to be restarted. 
While the success probability (i.e., the number of expected repetitions) should be taken into account to make a fair comparison,
Ref.~\cite{childsFirstQuantumSimulation2018} shows that the resulting overhead (very close to 1) is almost negligible.

As shown in \cref{method:fig:comparison} (a), for the nearest-neighbor Hamiltonian, the $\dpf$ has an advantage over $\dts$ and $\dqsp$ for large system sizes because the number of cross-node terms that requires communication is independent of system sizes. 
On the contrary, for the $k$-local Hamiltonian, \cref{method:fig:comparison} (b) shows the advantage of the post-Trotter protocols ($\dts$, $\dqsp$) over $\dpf$ because they have the logarithmic dependence on the inverse of simulation precision $1/\eps$.

\begin{figure*}[t]
    \centering
    \includegraphics[width=.85\linewidth]{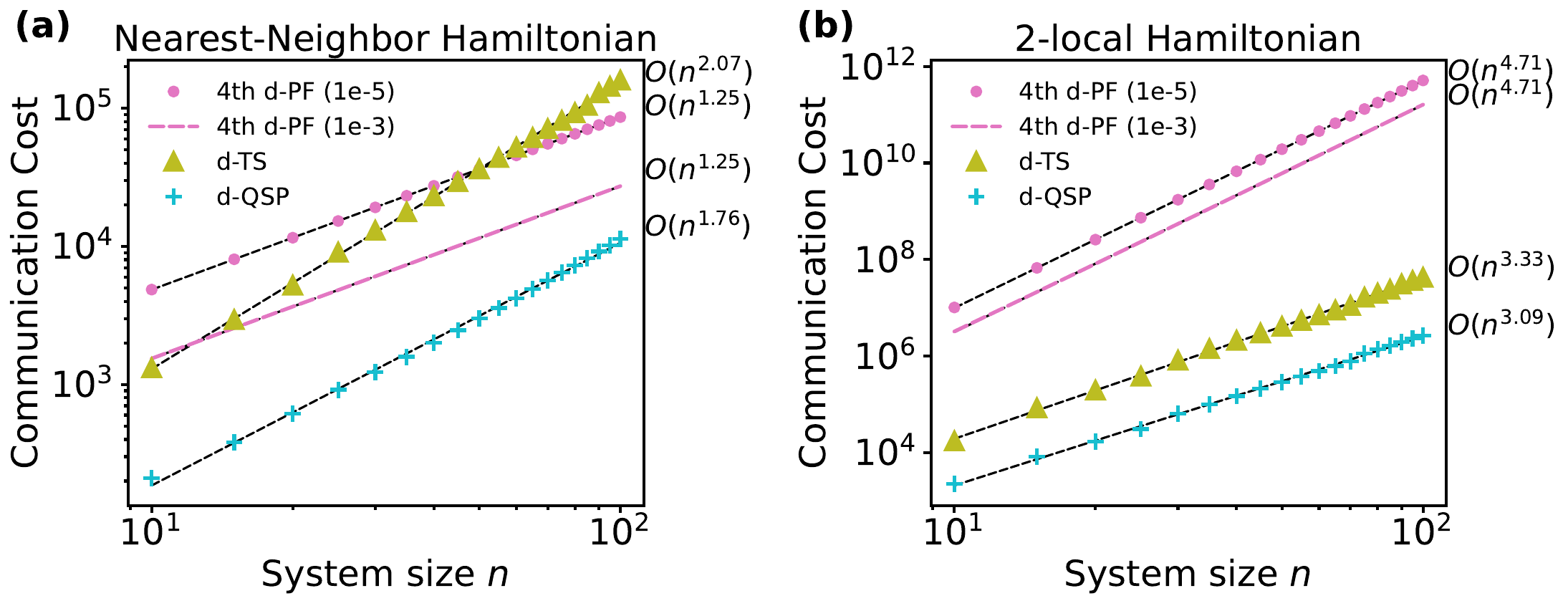}
    \caption{
    Estimation of communication cost for simulating two typical Hamiltonians distributedly by two ($\Gamma=2$) connected QPUs. 
    The $x$-axis is the system size $n$, and we set the simulation time as $t=n$.
    We choose two typical simulation precision $\eps=10^{-3}$ and $10^{-5}$ for the product formula protocol, while we set the high precision $\eps=10^{-5}$ for the post-Trotter protocols.
    (a) The nearest-neighbor Heisenberg Hamiltonian as \cref{eq:lattice_hamiltonian} with 
    ($J=0.1, h=0.2$). 
    The scalings of the communication cost with respect to system sizes are fitted and labeled.
    (b) The 2-local Hamiltonian as \cref{eq:local_hamiltonian} with 
    ($J_{ij}=0.1, h_j=0.2$).
    }
    \label{method:fig:comparison}
\end{figure*}


\section{Discussion of quantum communication complexity in noisy quantum network}\label{noysiN}



Similar to other algorithmic studies, our work assumes the idealized scenario where quantum computers and quantum channels operate without any noise. In other words, our conclusions are framed within the context of fault-tolerant quantum computation, focusing on the communication complexity of distributed quantum simulation protocols and other quantum algorithms.

In this section, we assume that each local quantum computer/processing unit is noiseless, while we assume that quantum channels in distributed quantum computation exhibit a certain degree of noise
(i.e., an imperfect Bell state may with an infidelity of 0.005). 
Furthermore, we suppose that quantum error correction is either partially ineffective or currently incapable of correcting channel errors, necessitating the application of entanglement purification. 
Thus, we investigate the additional quantum communication overhead required to ensure the accuracy of distributed quantum algorithms under these conditions.



It has been demonstrated that the Bell state serves as an ideal quantum channel for transmitting a single qubit of quantum information between two distant parties \cite{bennett1993teleporting}. 
However, in practical scenarios, noise within the quantum channel is unavoidable due to local dephasing, depolarizing, and amplitude damping that may occur during the distribution process \cite{xingFundamentalLimitationsCommunication2023}.
In quantum communication theory, the quantum purification protocol is employed to distill multiple copies of a noisy Bell state into either a perfect Bell state or one with high fidelity \cite{bennett1996purification}. 
Here, we estimate the quantum communication cost for distributed quantum simulation in the presence of noisy channels. 
We analyze the communication overhead necessary to achieve the target fidelity in purification, assuming channel noise. 
By taking these noises into account, we propose a revised upper bound for quantum communication in distributed quantum simulation.

\textbf{Isotropic twirling and noisy singlet.} 
Before delving into purification, it is essential to first examine the process of twirling a quantum channel \cite{HHHPhysRevA.60.1888}. 
Twirling is a technique used to transform a quantum channel into a depolarizing channel, which simplifies error characterization by averaging out certain types of noise. 

\begin{definition}
    Given a bipartite state $\rho$, the isotropic twirling of $\rho$ is defined by
    \begin{equation}
        \mathcal{T}_{iso}(\rho):= \int dU U\otimes U^* \rho U^\dagger \otimes U^T,
    \end{equation}
    where the integral is over the Haar measure.
\end{definition}

\begin{lemma} [Noisy singlet \cite{HHHPhysRevA.60.1888}]\label{noisyS}
    Assume the isotropic twirling $\mathcal{T}_{iso}$ is a map for system $\mathcal{H}_A\otimes \mathcal{H}_B$  with $d:=\text{dim}\mathcal{H}_A=\text{dim}\mathcal{H}_B$, and $\rho$ being an input state acting on system $\mathcal{H}_A\otimes \mathcal{H}_B$, then the
    output state under $\mathcal{T}_{iso}$ is a noisy singlet, i.e., 
    \begin{equation}
    W=p\ket{\phi^+} \bra{\phi^+}+   (1-p)\frac{I\otimes I}{d^2},
    \end{equation}
    where $\ket{\phi^+}=\frac{1}{\sqrt{d}}\sum_i \ket{i}_A\ket{i}_B$, and  the parameter $p$ is characterized by
    \begin{equation}
     p= \frac{\Tr (\ket{\phi^+} \bra{\phi^+} \rho)d^2-1}{d^2-1}.
    \end{equation}
\end{lemma}
Define the fidelity as $F:=\Tr (\ket{\phi^+} \bra{\phi^+} \rho)$. 
It can be easily verified that for a noisy singlet state, $\Tr (\ket{\phi^+} \bra{\phi^+} W)=F$. 
This indicates that isotropic twirling preserves the entanglement fidelity.

Without loss of generality, we assume that the quantum channel used for one-qubit teleportation is a Bell state affected by white noise, commonly referred to as a noisy singlet state. 
This assumption is justified because any two-qubit density matrix can be transformed into a noisy singlet state (also called Werner state) through isotropic twirling (see \cref{noisyS}). 

A two-qubit noisy singlet state is given as 

\begin{equation}
    \rho_{noise}= p\ket{\phi^+}\bra{\phi^+} +(1-p) \frac{I\otimes I}{4}=F\ket{\phi^+}\bra{\phi^+}+\frac{1-F}{3}(\op{\phi^-}+\op{\psi^+}+\op{\psi^-}),
\end{equation}
where $\ket{\phi^{\pm}},\ket{\psi^{\pm}}$ form a Bell basis. 
Here $F=\frac{3p+1}{4}$.\\

\textbf{Entanglement purification.}
An entanglement purification protocol enhances the fidelity of entangled states by utilizing two or more copies of Werner states. 
In \cite{bennett1996purification}, Bennett et al. demonstrate that by applying an entanglement purification protocol to two Werner states with initial entanglement fidelity $F$, it is possible to produce an entangled state with enhanced fidelity, potentially achieving a form given by

\begin{equation}\label{purification_bennett}
    F_B = \frac{F^2+\frac{1}{9}(1-F)^2}{F^2+\frac{2}{3}F(1-F)+\frac{5}{9}(1-F)^2},
\end{equation} 
with a successful probability $ P_{succ.}=(5-4F+8F^2)/9$.

   Recently, it has been proposed that one can further enhance the performance of entanglement purification by using three pairs of Werner states \cite{Wpurification}, resulting in a boosted fidelity

\begin{equation}\label{purification3}
    F^\prime=\frac{30F^3-6F^2+3F}{16F^3+18F^2-12F+5},
\end{equation}
with a successful probability $P_{succ.}=(16F^3+18F^2-12F+5)/{27}.$

\begin{figure*}[t]
    \centering
    \includegraphics[width=.45\linewidth]{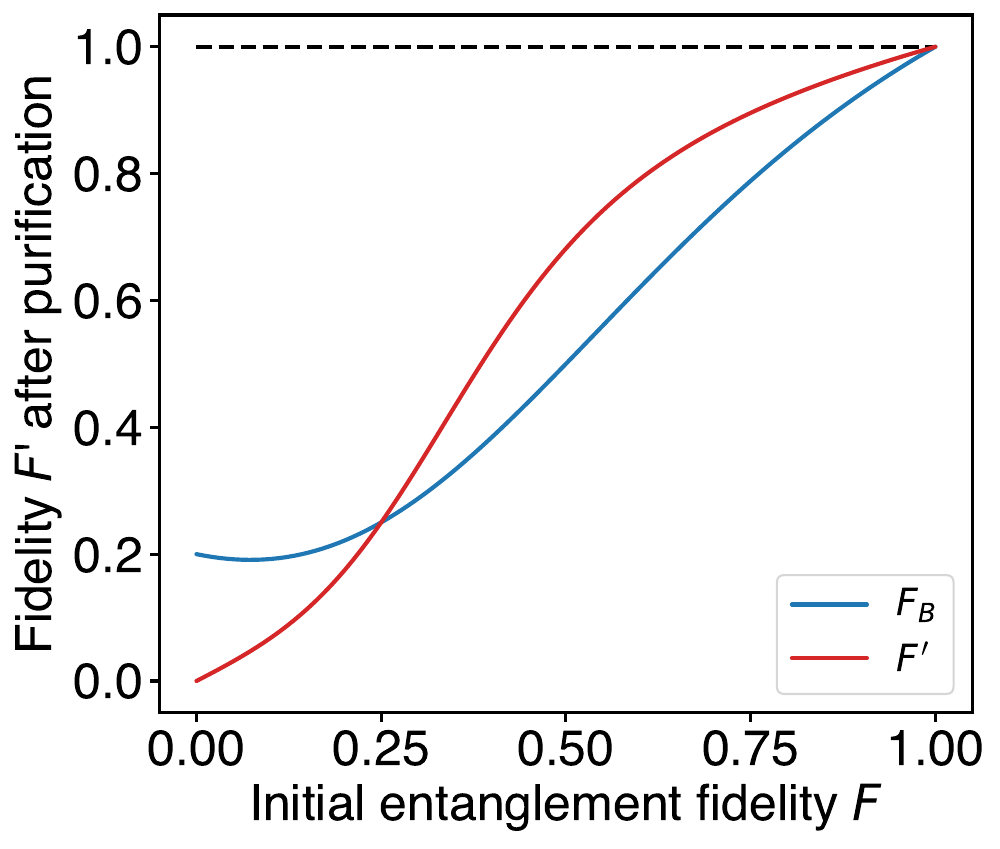}
    \includegraphics[width=.45\linewidth]{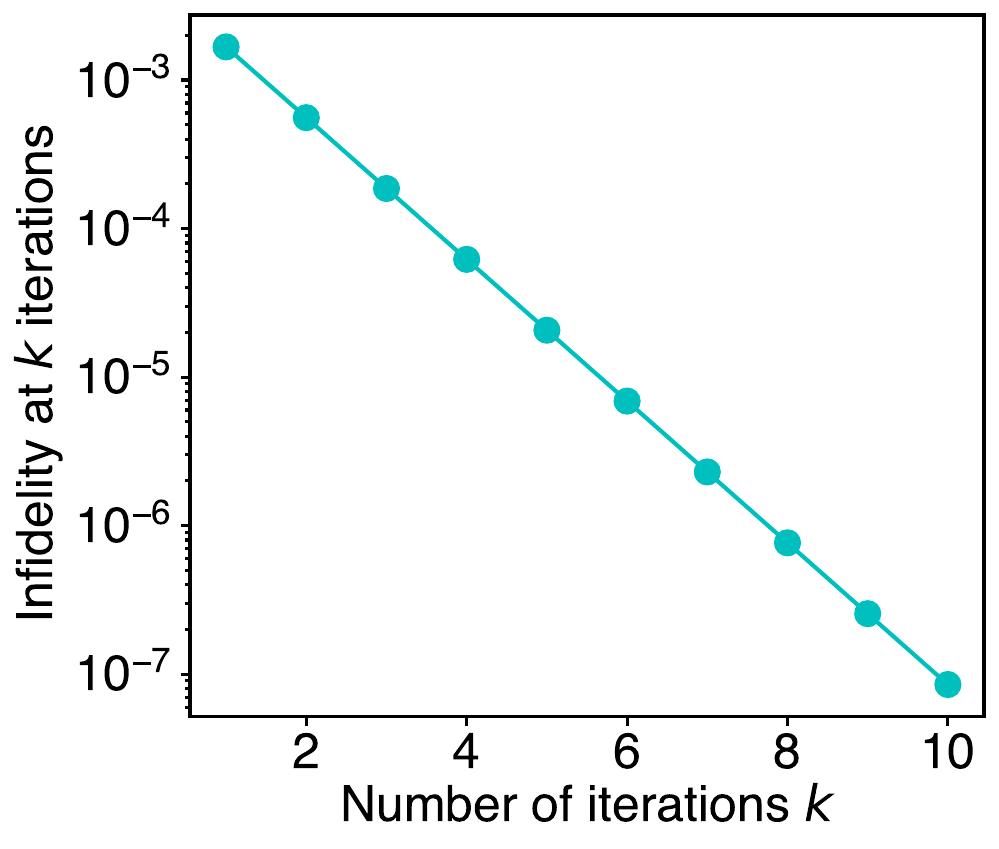}
    \caption{
    \emph{Left panel}. The fidelity $F'$ after one purification \cref{purification3} where $F_B$ is Bennett's protocol \cref{purification_bennett} for reference. 
    \emph{Right panel}. The infidelity $1-F_k$ at the $k$th iteration of \cite{Wpurification} with the initial fidelity is 0.995.
    }
    \label{method:fig:noisy_fidelity_estimation}
\end{figure*}

Assuming an (early) fault-tolerant regime, our distributed protocols are designed with an initial entanglement fidelity of $99.5\%$ (can be worse than this) for the quantum channel, exemplified by noisy EPR pairs. We further assume noiseless local quantum processing units, facilitated by available quantum error correction. 

Suppose $F_0 = 0.995$, using the recursive formula Eq. (\ref{purification3}), one can calculate $F_k$ where $k$ is the iteration number.
We run the numerical test in \cref{method:fig:noisy_fidelity_estimation}.  
Since $P_{succ.}=(16F^3+18F^2-12F+5)/{27}$ is close to $1$ if for large $F$, for simplicity, we suppose the overall successful probability for entanglement purification is $P\approx1$.
The cost of copies of entangled states is 
$$N=3^k.$$
Suppose $F_0=0.995$, $F_{final}=0.999999$ such that $k=8$, and $3^8=6561$. 
That is the total cost for a single nearly perfect quantum channel (a Bell pair) is $ 6,561$ copies of a noisy Bell pair we distribute. 

Note that, although the cost of the entangled copies is high, the purification scheme can, in principle, be implemented and the entangled copies purified and stored in advance (assuming quantum memory) as a preshared resource.\\

\textbf{Modified error bound for distributed quantum simulation with noisy channels.}
Given an entangled state with entanglement fidelity $F$, one can transform it to a noisy singlet state (Werner state) with entanglement fidelity $F$ by isotropic twirling. Here, we utilize the 1-norm to quantify the error in the global quantum state after teleportation of a single qubit of interest. The quantum teleportation with a noisy singlet state can be seen as an effective partial depolarizing channel, e.g. $\mathcal{E}(\rho)=p\rho+(1-p)\frac{I}{2}$, where $p=\frac{4F-1}{3}$.

Suppose the ideal $n$-qubit state is $\ket{\psi}=\alpha \ket{0}_k\ket{\psi}^{(0)}_{n\backslash k}+\beta \ket{1}_k\ket{\psi}^{(1)}_{n\backslash k}$ and its density matrix is given as
$\rho_{0}=\op{\psi}$. Define the error states $\ket{\psi}_{E_1}=\alpha \ket{0}_k\ket{\psi}^{(0)}_{n\backslash k}-\beta \ket{1}_k\ket{\psi}^{(1)}_{n\backslash k}$, $\ket{\psi}_{E_2}=\alpha \ket{1}_k\ket{\psi}^{(0)}_{n\backslash k}+\beta \ket{0}_k\ket{\psi}^{(1)}_{n\backslash k}$, and $\ket{\psi}_{E_3}=\alpha \ket{1}_k\ket{\psi}^{(0)}_{n\backslash k}-\beta \ket{0}_k\ket{\psi}^{(1)}_{n\backslash k}$.
After a single teleportation of the $k$-th qubit, the state becomes 
\begin{equation}
\begin{split}
       \mathcal{E}_k(\rho_{0})
       &=p\rho_{0}+(1-p) \qty(\frac{1}{2} \ket{\psi}^{(0)}_{n\backslash k}\ket{0}_k\bra{0}_k\bra{\psi}^{(0)}_{n\backslash k} +\frac{1}{2} \ket{\psi}^{(1)}_{n\backslash k}\ket{1}_k\bra{1}_k\bra{\psi}^{(1)}_{n\backslash k})\\
       &=F\rho_{0}+(1-F) \qty(\frac{1}{3} \op{\psi}_{E_1}+\frac{1}{3} \op{\psi}_{E_2}+\frac{1}{3} \op{\psi}_{E_3}). 
\end{split}
\end{equation}

Now the 1-norm error between $\mathcal{E}_k(\rho_{0})$ and $\rho_{0}$ is given as 

\begin{equation}
\begin{split}
    \norm{\rho_{0}-\mathcal{E}_k(\rho_{0})}_1
    &= \norm{(1-F)\rho_{0}-(1-F) \qty(\frac{1}{3} \op{\psi}_{E_1}+\frac{1}{3} \op{\psi}_{E_2}+\frac{1}{3} \op{\psi}_{E_3})}\\
    &= (1-F)\norm{\rho_{0}- \qty(\frac{1}{3} \op{\psi}_{E_1}+\frac{1}{3} \op{\psi}_{E_2}+\frac{1}{3} \op{\psi}_{E_3})}\\
    &\le 2(1-F),
\end{split}
\end{equation}
since for any denstiy matrix $\rho$ and $\sigma$, $\norm{\rho-\sigma}_1\le 2$. The same error upper bound holds for the case where $\rho_{0}$ is a mixed state, since the mixed state is the classical mixture of pure states. 

Therefore, for multiple instances of quantum teleportation in distributed quantum simulation or computation, one can use the triangle inequality to bound the total error as 
\begin{equation}
\begin{split}
      \epsilon_{comm.}& =\norm{\rho_{0}-\mathcal{E}_{all}(\rho_{0})}_1\\
      &=    \norm{\rho_{0}-\mathcal{E}_{k_q}(...\mathcal{E}_{k_2}(\mathcal{E}_{k_1}(\rho_{0})))}_1  \\
       &=\norm{\rho_{0}-\mathcal{E}_{k_q}(...\mathcal{E}_{k_2}(\mathcal{E}_{k_1}(\rho_{0}))) + \mathcal{E}_{k_{q-1}}(...\mathcal{E}_{k_2}(\mathcal{E}_{k_1}(\rho_{0}))) - ... + \mathcal{E}_{k_1}(\rho_{0})-\mathcal{E}_{k_1}(\rho_{0})}_1 \\
       &\le  \norm{\rho_{0}-\mathcal{E}_{k_1}(\rho_{0})}+\norm{\mathcal{E}_{k_1}(\rho_{0})-\mathcal{E}_{k_2}(\rho_{0})}+...\\
       &\le 2(1-F) q,
\end{split}
\end{equation}
where $q$ is the number of quantum teleportations (communications). 

Suppose our algorithm error is $\epsilon_{alg.}$, the totol error should be 

\begin{equation}
\begin{split}
        \epsilon_{total}& \le \epsilon_{alg.}+2(1-F) q,
\end{split}
\end{equation}
where $q$ is the quantum communication cost depends on the distributed methods of quantum simulation. Setting $2(1-F) q=\epsilon_{alg.}$, one has $F=1-\epsilon_{alg.}/(2q)$.

\section{Communication lower bound of distributed quantum simulation}\label{sec:proof_lower_bound}


The area of communication complexity deals with the scenario where multiple parties do computation collectively. 
Under this model, local computation is free, but communication is expensive and has to be minimized. 
    Assume there are two separated parties, Alice and Bob who have a bit-string $\vby\in \qty{0,1}^n$ and $\vbz\in \qty{0,1}^n$ respectively, they want to collectively evaluate a Boolean function 
    $\predicate(\vby,\vbz): \qty{0,1}^n \times \qty{0,1}^n \to\qty{0,1}$ by local computation and communication with each other.
    Under the (standard) two-party quantum communication model, Alice and Bob each have a quantum computer and are allowed to communicate qubits (or to make use of the quantum correlations given by shared EPR pairs).
This model was introduced by Yao \cite{yaoComplexityQuestionsRelated1979} and has been studied extensively \cite{dewolfQuantumComputingCommunication2001, buhrmanNonlocalityCommunicationComplexity2010}.
The communication complexity of $\predicate(\vby,\vbz)$ is the minimal amount of communication required for the worst-case input.

\subsection{Lower bound in time $t$}
\begin{definition}[$\QC_2$]\label{method:def:qc2}
    Under the standard two-party quantum communication model,
    the \emph{(two-sided) bounded error quantum communication complexity} of a Boolean function $\predicate(\vby,\vbz)$, denoted $\QC_2(\predicate)$,
    is the minimal number of qubits that Alice and Bob require to exchange to return $\predicate(\vby,\vbz)$ with correct probability at least $2/3$ on the worst-case input.
\end{definition}

    
To prove the lower bound of quantum communication complexity of distributed quantum simulation ($\dqdhs$), 
we utilize the distributed version of $\parity_n(\vbx)$ known as the $\InnerProduct$ problem with known lower bound.
\begin{problem}[$\IP_n(\vby,\vbz)$]\label{method:prm:inner_product}
    Under the standard two-party quantum communication model,
    the $\InnerProduct$ (denoted $\IP_n$) problem refers to evaluating the Boolean function 
    $ \IP_n(\vby,\vbz):=\bigoplus_{j=1}^{n} (y_j \wedge z_j) \equiv \vby \cdot \vbz \textup{ mod }2$
    by communicating as few as possible (qu)bits between Alice and Bob who have the bit-string $\vby\in\qty{0,1}^n$ and $\vbz\in\qty{0,1}^n$ respectively.
\end{problem}
\begin{lemma}[Tight bound of $\IP_n$]\label{method:thm:innerproduct}
    $\QC_2(\IP_n)=\Theta(n)$.
\end{lemma}
The intuition behind \cref{method:thm:innerproduct} is by the following:
one can use the protocol for $\InnerProduct$ to transfer Alice’s $n$-bit input to Bob, 
and then by Holevo’s theorem which states that $n$ qubits must be communicated to achieve this.
We refer readers to \cite{kremerQuantumCommunication1995, dewolfQuantumComputingCommunication2001} for careful proof of this theorem.
For $\Gamma$-partite lower bound, Le Gall and Suruga \cite{legallBoundsObliviousMultiparty2022} proved the lower bound of $\QC_2(\IP_n)$ can be generalized from bipartite to $\Gamma$-partite case as $\Omega(\Gamma n)$.
\begin{problem}[$\Gamma$-partite $\InnerProduct$ $\IP_{n,\Gamma}$] \label{prm:multi_innerproduct}
    The $\Gamma$-partite $\InnerProduct$ Problem $\IP_{n,\Gamma}:\{0,1\}^{\Gamma n}\rightarrow \{0,1\}$ is defined as
    \[
    \IP_{n,\Gamma}(\vbx_1,...,\vbx_\Gamma) = \oplus_{j=1}^n (x_{1j}\land x_{2j} \cdots \land x_{\Gamma j}),
    \]
    where $\vbx_i \in \{0,1\}^n$ for $i = \{1,...,\Gamma\}$.
\end{problem}
\begin{lemma}[Quantum communication complexity of $\IP_{n,\Gamma}$ function \cite{legallBoundsObliviousMultiparty2022}]\label{apd:thm:kparty}
    To compute $\IP_{n,\Gamma}$ function with bounded-error using an oblivious quantum protocol, the number of qubits required to communicate in $\Gamma$ partite is at least  $\Theta(\Gamma n)$.
\end{lemma}
Their lemma requires the communication model to be oblivious, which means the amount of communication exchanged between each pair of players at each round is fixed independently of the input before the execution of the protocol.
Our setting satisfies this requirement.

To reduce $\IP_n$ to $\dqdhs$, we introduce a useful constructive tool called \emph{circuit-to-Hamiltonian map} (or known as \emph{perfect state transfer}) which maps a quantum circuit to a Hamiltonian evolution:
\begin{lemma}[Circuit-to-Hamiltonian map \cite{feynmanQuantumMechanicalComputers1985, christandlPerfectStateTransfer2004, kayReviewPerfectState2010}]\label{method:thm:circuit_hamiltonian}
    Given a quantum circuit $U=U_{N}\dots U_1$ composed of $N$ gates, 
    there exists a Hamiltonian $H_{U} \in \mathcal{H}_{\cc} \otimes \mathcal{H}_{\oo}$ acting on a clock and output register, i.e.
    \begin{equation}\label{method:eq:circuit_hamiltonian}
        H_{U} = \sum_{j=1}^{N} \sqrt{j(N-j+1)} 
        \qty(\op{j}{j-1}_{\cc}\otimes U_j 
        + \HC
        ),
    \end{equation}
    where $\HC$ means Hermitian conjugate terms.
    The quantum dynamic evolution governed by $H_U$ and $t=\pi/2$ yields the original quantum circuit output as
    \begin{equation}
         e^{-\ii H_{U}t}\ket{0}^{\otimes \log(N)}_{\cc}\ket{\psi_0}_{\oo}
         =\ket{N}_{\cc}\otimes U \ket{\psi_0}_{\oo}.
    \end{equation}
\end{lemma}

\begin{figure*}[t]
    \centering
    \includegraphics[width=.95\linewidth]{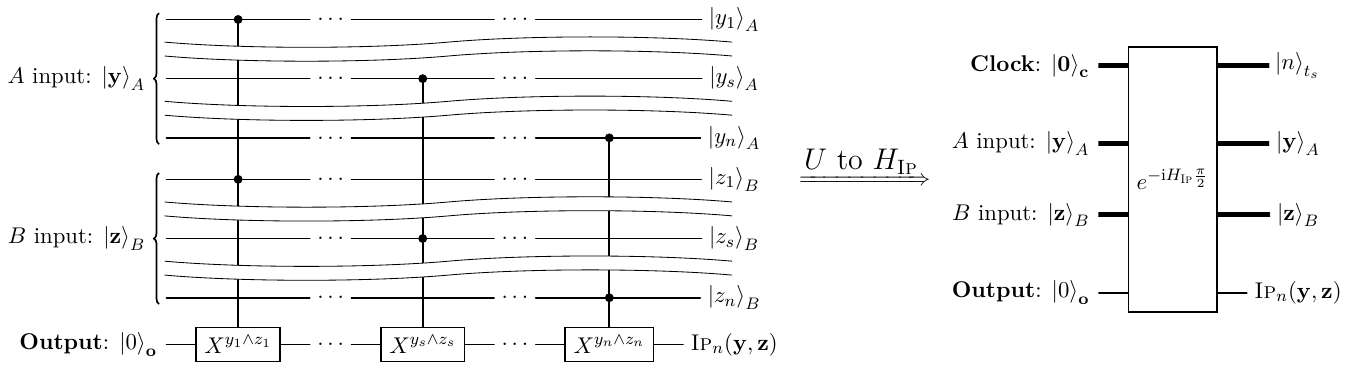}
    \caption{
    The construction from an $\IP_{n}(\vb{y},\vb{z})$ circuit (instance) to a Hamiltonian evolution (instance) with $H_{\IP}$ and initial state $\ket{\vb{0}}_{\cc}\otimes\ket{\vb{y}}_A\otimes\ket{\vb{z}}_B\otimes \ket{0}_{\oo}$.
    Alice (Bob) encodes an $n$-bit string $\vby$ ($\vbz$) with $n$-qubit register and the output register consists of one qubit to store the inner product of two bit-strings.
    The L.H.S circuit consists of $n$ Toffoli gate controlled by every pair of qubits of Alice and Bob,
    that is, apply an $X$ gate to the output register when both $y_s$ and $z_s$ are 1.
    By the circuit-to-Hamiltonian map, a Hamiltonian $H_{\IP}$ can be constructed by introducing $\log(n)$-qubit ancilla qubits as a clock register such that the dynamics of $H_{\IP}$ with $t=\pi/2$ yields the inner product of Alice and Bob's input bit-strings.
    }
    \label{method:fig:innerproduct_weak}
\end{figure*}

With these ingredients (\cref{method:thm:innerproduct} and \cref{method:thm:circuit_hamiltonian}), we are ready to prove a lower bound on $\QC_2$.
\begin{theorem}[Lower bound on $\QC_2$ of $\dqdhs$, Theorem 4 in maintext]\label{apd:thm:weak_lower_bound}
    For any positive integer $\tau\le n$, 
    there exists a sparse Hamiltonian $H$ acting on $2n+\log(n)+1$ qubits with $\norm{H}=\tau$ and a constant evolution time $t=\pi/2$ such that 
    the (bounded-error) quantum communication complexity ($\QC_2$) of any generic oblivious $\Gamma$-partite protocol for distributed quantum simulation (\dqdhs) scales at least linearly in scaled evolution time, that is, $\QC_2(\dqdhs)= \Omega(\Gamma\norm{H}t)$.
\end{theorem}
\begin{proof}[Proof of \cref{apd:thm:weak_lower_bound}]
    Given an instance of $\IP_n(\vby,\vbz)$,
    one can write a quantum circuit composed of $n$ Toffoli gates to evaluate $\IP(\vby,\vbz)$ by two parties (Alice and Bob) as in the L.H.S of \cref{method:fig:innerproduct_weak}.
    By \cref{method:thm:circuit_hamiltonian}, given a quantum circuit, a Hamiltonian $H_{\IP}$ 
    can be constructed as 
    \begin{equation}
        H_{\IP} := 
        \sum_{s\in[n]} 
        \sqrt{s(n-s+1)} \op{s}{s-1}_{\cc} \otimes \CCNOT(y_s,z_s) 
        + \HC
    \end{equation}
    of which the dynamics with the scaled evolution time $\norm{H_{\IP}}t=n\pi/2$ perfectly yields the circuit output.
    According to the known lower bound of quantum communication complexity for $\IP_n$ is $\Omega(n)$ (\cref{method:thm:innerproduct}), 
    any distributed quantum simulation algorithm needs at least $\Omega(n)$ communication.
    For any positive integer $\tau< n$, we can similarly map an instance $\IP_{\tau}(\vby,\vbz)$ to a Hamiltonian with $\norm{H_{\IP}}=\tau$ such that the linear lower bound also applies.
It is worth noting that in the proof of $\nff$ theorem, the input $\vbx$ is encoded into the Hamiltonian $H_{\vbx}$ and the evolution is simulated by querying entries of $H_{\vbx}$.
In contrast, in this proof, the input $\vby$ ($\vbz$) is encoded into Alice (Bob)'s quantum state, and the Hamiltonian $H_{\IP}$ is public to two parties, which matches the realistic scenario.
For $\Gamma$-partite case, by \cref{apd:thm:kparty}, our lower bound can be directly extended to $\Omega(\Gamma \norm{H}t)$ by replacing $\IP_n$ with $\IP_{n,\Gamma}$ and replacing the Toffoli gates in \cref{method:fig:innerproduct_weak} with multi-partite controlled-NOT gates. 
Here, we underscore that for the task of quantum simulation dynamics, the amount of communication is fixed independently of the input.
Therefore, it fits within the oblivious protocol context and our lower bound on party number $\Gamma$ is optimal.
\end{proof}

In practical scenarios such as Variational Quantum Eigensolvers (VQE), one is often interested in estimating the expectation value $\langle \psi(t) | O | \psi(t) \rangle$ of an observable $O$, rather than obtaining the full quantum state $|\psi(t)\rangle$. Here, we show that the communication lower bound derived in Theorem 4 (Main text) and \cref{apd:thm:weak_lower_bound} implies a similar lower bound for the distributed estimation of expectation values.
\begin{corollary}[Communication lower bound for expectation values]\label{apd:thm:lower_bound_ob}
    For a $\Gamma$-partite quantum network, there exists a sparse Hamiltonian $H$ and an observable $O$ such that estimating the expectation value $\langle\psi_0 |  e^{\ii Ht} O  e^{-\ii Ht} |\psi_0 \rangle$ within an additive error $\epsilon $ requires $\Omega(\Gamma \|H\| t)$ quantum communication.
\end{corollary}
\begin{proof}
We employ the same reduction from the $\Gamma$-partite Inner Product ($\IP_{n, \Gamma}$) problem as used in the proof of Theorem F.1. 
Let us consider the bipartite case first, where Alice and Bob each hold an $n$ bit-string $(\mathbf{y},\mathbf{z})$ encoded in their states.
Recall that the circuit-to-Hamiltonian map constructs a Hamiltonian $H_{\IP}$ acting on a clock register $\mathcal{H}_c$ and an output register $\mathcal{H}_{out}$ (along with input registers), such that the evolution at time $t=\pi/2$ yields:
\begin{equation}
    e^{-\ii H_{\IP} t} |\Psi_{in}\rangle = |N\rangle_c \otimes |\mathbf{y}\rangle_A \otimes |\mathbf{z}\rangle_B \otimes 
    |\IP_{n}(\mathbf{y},\mathbf{z})\rangle_{out},
\end{equation}
where $N$ corresponds to the final clock step, and the last register encodes the boolean result of the inner product function.
Consider a specific observable $O$ defined as:
\begin{equation}
    O = |N\rangle\langle N|_c \otimes Z_{out},
\end{equation}
where $Z_{out}$ is the Pauli-$Z$ operator acting on the output qubit, 
$|N\rangle\langle N|_c$ is the projector on the clock register,
and the identity operator acts on all other registers. 
The expectation value of this observable with respect to the evolved state is:
\begin{align}
    \langle O \rangle 
    &= \langle \Psi_{in} | e^{\ii H_{\IP} t} (|N\rangle\langle N|_c \otimes Z_{out}) e^{-\ii H_{\IP} t} | \Psi_{in} \rangle \nonumber \\
    &= \langle \IP_{n}(\mathbf{y},\mathbf{z}) | Z | \IP_{n}(\mathbf{y},\mathbf{z}) \rangle_{out} \nonumber \\
    &= (-1)^{\IP_{n}(\mathbf{y},\mathbf{z})}.
\end{align}
Since $\IP_{n}(\mathbf{y},\mathbf{z}) \in \{0, 1\}$, 
the expectation value $\langle O \rangle$ takes values strictly in $\{+1, -1\}$. 
If a distributed protocol can estimate $\langle O \rangle$ with an additive error $\epsilon < 1$, one can unambiguously distinguish between the outcomes $+1$ (implying $\IP=0$) and $-1$ (implying $\IP=1$) by simply rounding the estimated value to the nearest integer. 
This would effectively solve the Inner Product problem and directly generalize to the $\Gamma$-partite case.
Therefore, the quantum communication complexity of estimating this expectation value is lower-bounded by the complexity of the $\IP_{n,\Gamma}$ problem, which is $\Omega(\Gamma n)$. 
Given that the simulation time is constant $t=\pi/2$ and $\|H_{\IP}\| \approx n$, the lower bound scales as $\Omega(\Gamma \|H\| t)$. 
This demonstrates that for generic Hamiltonians, weak simulation (expectation value estimation) does not offer a communication advantage over strong simulation in the worst-case scenario.
\end{proof}


\subsection{Lower bound in precision $\eps$}

Besides the lower bound with respect to evolution time $t$, we can prove the lower bound regarding simulation precision $\eps$ with an unbounded-error argument.
The difference between the bounded error \nameref{method:def:qc2} and the unbounded-error $\QC_0$ is
that the unbounded case only requires a success probability larger than $1/2$, even though by an exponentially small probability.
\begin{definition}[$\QC_0$]\label{method:def:qc0}
    Under the standard two-party quantum communication model,
    the unbounded error quantum communication complexity of a Boolean function $\predicate(\vby,\vbz)$, denoted $\QC_0(\predicate)$,
    is the minimal number of qubits that Alice and Bob require to exchange to return $\predicate(\vby,\vbz)$ with correct probability larger than $1/2$ on the worst input.
\end{definition}
Ref.~\cite{forsterLinearLowerBound2001} proved that the unbounded-error classical communication complexity $\CC_0$ for \nameref{method:prm:inner_product} is $\Omega(n)$.
On the quantum side, Ref.~\cite{iwamaUnboundederrorClassicalQuantum2007} proved that the unbounded-error quantum communication complexity $\QC_0(\IP_n)$ is exactly half of the classical one $\CC_0(\IP_n)$.
\begin{lemma}[Lower bound on $\QC_0$ of $\IP_n$ \cite{forsterLinearLowerBound2001,iwamaUnboundederrorClassicalQuantum2007}]\label{thm:unbounded_lower_ip}
    The unbounded-error quantum communication complexity $\QC_0$ of \nameref{method:prm:inner_product} is
    $\QC_0(\IP_n)= \Omega(n)$.
\end{lemma}

\begin{theorem}[Lower bound on communication complexity of $\dqdhs$ regarding $\eps$]\label{apd:thm:lower_bound_eps}
    For any constant evolution time $t>0$ and precision $\eps > 0$,
    let $n = \Theta\qty(\frac{\log(1/\eps)}{\log\log(1/\eps)})$.
    There exists a sparse Hamiltonian $H$ with $\norm{H}\le 1$ acting on $2n+\log(n)+1$ qubits
    such that any generic bi-partite protocol for $\eps$-approximate distributed quantum simulation of $e^{-\ii Ht}$ requires
    $\Omega\qty(\frac{\log(1/\eps)}{\log\log(1/\eps)})$ (quantum) communication complexity.
\end{theorem}
\begin{proof}
    Let $n$ be a positive integer to be determined.
    Given an instance $(\vby, \vbz) \in \qty{0,1}^n \times \qty{0,1}^n$ of $\IP_n$,
    Alice and Bob encode their bit-strings as $\ket{\vby}_A$ and $\ket{\vbz}_B$ respectively.
    By the circuit-to-Hamiltonian construction (\cref{method:thm:circuit_hamiltonian}),
    the quantum circuit of $n$ Toffoli gates for evaluating $\IP_n(\vby,\vbz)$
    (as in the proof of \cref{apd:thm:weak_lower_bound})
    yields a Hamiltonian $H_{\IP}$ acting on $2n + \log(n) + 1$ qubits
    with $\norm{H_{\IP}} = n$.
    Define $H := H_{\IP} / n$ so that $\norm{H} = 1$.
    The initial state is $\ket{\vb{0}}_{\cc} \ket{\vby}_A \ket{\vbz}_B \ket{0}_{\oo}$, a product state with no superposition.

    By the spin-$n/2$ structure of the circuit-to-Hamiltonian construction~\cite{berryExponentialImprovementPrecision2014},
    for any $t > 0$, the overlap with the correct inner product is
    \begin{equation}\label{eq:overlap_eps}
        \abs{\bra{n}_{\cc} \bra{\vby} \bra{\vbz} \bra{\IP_n(\vby,\vbz)}_{\oo}
        e^{-\ii Ht} \ket{\vb{0}}_{\cc} \ket{\vby} \ket{\vbz} \ket{0}_{\oo}}
        = \abs{\sin(t/n)}^n,
    \end{equation}
    while the overlap with the incorrect answer $\ket{n}_{\cc} \ket{\vby} \ket{\vbz} \ket{\IP_n(\vby,\vbz) \oplus 1}_{\oo}$
    is exactly zero, since the Hamiltonian's graph consists of two disconnected paths.
    We construct an unbounded-error protocol for $\IP_n$ as follows:
    measure the clock register of $\tilde{\rho}$;
    if the outcome is $n$, output the value of the output register;
    otherwise, output $0$ or $1$ with equal probability.

    Suppose we have an $\eps$-approximate distributed simulation of $e^{-\ii Ht}$ using $C$ qubits of communication,
    producing a state $\tilde{\rho}$ with $\norm{\tilde{\rho} - \op{\psi_t}}_1 \leq \eps$ where $\ket{\psi_t} = e^{-\ii Ht} \ket{\vb{0}}_{\cc} \ket{\vby} \ket{\vbz} \ket{0}_{\oo}$.
    Let $p_{\mathrm{c}}$ and $p_{\mathrm{w}}$ be the probability of measuring the correct answer $\ket{n}_{\cc}\ket{\IP_n}_{\oo}$ and the wrong answer $\ket{n}_{\cc}\ket{\IP_n\oplus 1}_{\oo}$ respectively.
    The success probability of this protocol is
    \begin{equation}
        \Pr[\mathrm{correct}] = p_{\mathrm{c}} + \frac{1}{2}(1 - p_{\mathrm{c}} - p_{\mathrm{w}})
        = \frac{1}{2} + \frac{1}{2}(p_{\mathrm{c}} - p_{\mathrm{w}}).
    \end{equation}
    For the ideal state $\ket{\psi_t}$, we have $p_{\mathrm{c}} = \abs{\sin(t/n)}^{2n}$ and $p_{\mathrm{w}} = 0$.
    For the approximate state $\tilde{\rho}$, by the trace distance bound,
    $p_{\mathrm{c}} \geq \abs{\sin(t/n)}^{2n} - \eps$ and $p_{\mathrm{w}} \leq \eps$.
    Therefore $\Pr[\mathrm{correct}] \geq \frac{1}{2} + \frac{1}{2}(\abs{\sin(t/n)}^{2n} - 2\eps)$.
    For $\Pr[\mathrm{correct}] > 1/2$, it suffices to have $\abs{\sin(t/n)}^{2n} > 2\eps$.

    Since $\sin(x) \geq 2x/\pi$ for $x \in [0, \pi/2]$, a sufficient condition is $(2t/(\pi n))^{2n} > 2\eps$,
    i.e., $n\log(\pi n / (2t)) < \frac{1}{2}\log(1/(2\eps))$.
    Choosing $n = \Theta\qty(\frac{\log(1/\eps)}{\log\log(1/\eps)})$ satisfies this condition for any constant $t > 0$.
    By \cref{thm:unbounded_lower_ip}, the unbounded-error quantum communication complexity $\QC_0(\IP_n) = \Omega(n)$.
    Therefore, the communication complexity is at least $C\ge\Omega(n) = \Omega\qty(\frac{\log(1/\eps)}{\log\log(1/\eps)})$.
\end{proof}

\subsection{Lower bound in $t$ for $k$-local Hamiltonians via entanglement}\label{apd:lower_bound_k-local_entanglement}

In the previous section, we established communication complexity lower bounds for simulating \emph{sparse} Hamiltonians by reducing from the inner product problem in the communication model.
However, the sparse Hamiltonians are not necessarily $k$-local,
which are more physically relevant and commonly studied in the Hamiltonian simulation literature.
To bridge this gap, we establish an $\Omega(t)$ quantum communication lower bound for simulating $k$-local Hamiltonians using an entanglement capacity argument and a simple Ising Hamiltonian.
The core principle is that entanglement across the Alice-Bob cut can only be generated by quantum communication; local operations preserve the entanglement entropy.

The proof has three ingredients.
First, we bound the entanglement that quantum communication can create (\cref{lem:comm-capacity}).
Second, a continuity bound ensures that any good simulation must approximately reproduce the target entanglement (\cref{lem:fannes,cor:sim-entanglement}).
Third, we exhibit a hard instance---a transverse-field Ising chain---whose dynamics generates linear entanglement (\cref{thm:schuch-entropy}).
Combining these, any simulation of this Ising chain must also generate $\Omega(t)$ entanglement, forcing $\Omega(t)$ qubits of quantum communication.

\begin{lemma}[Communication entanglement capacity]\label{lem:comm-capacity}
    In any unitary protocol where Alice and Bob start from a product state $\ket{\Psi_0} = \ket{\alpha}_A \otimes \ket{\beta}_B$ and exchange $Q$ qubits of quantum communication (with arbitrary local unitaries in between), the final state satisfies $S(\rho_A) \leq Q$.
\end{lemma}
\begin{proof}
    Without loss of generality, each communication round sends exactly one qubit (a multi-qubit message decomposes into single-qubit transfers with local unitaries in between).
    We prove by induction on $k$ that, after $k$ communicated qubits, $S(\rho_A) \leq k$.

    \emph{Base case} ($k = 0$).
    Only local unitaries are applied.
    Since $(U_A \otimes U_B)\ket{\alpha}\ket{\beta}$ remains a product state, $S(\rho_A) = 0$.

    \emph{Inductive step.}
    Assume that after $k$ qubits have been communicated, the (generally entangled) state $\ket{\Psi_k}$ satisfies $S(\rho_A) \leq k$.
    Between communication rounds, local unitaries $U_A \otimes U_B$ preserve $S(\rho_A)$ (the reduced state $\rho_A \mapsto U_A \rho_A U_A^\dagger$ has the same spectrum).
    It suffices to show that sending one qubit across the cut increases $S(\rho_A)$ by at most $1$; crucially, the following argument applies to any state, not only product states.

    Suppose Alice sends qubit $q$ to Bob (the reverse direction is symmetric).
    The global pure state is $\ket{\Psi} \in \mathcal{H}_{A'} \otimes \mathcal{H}_q \otimes \mathcal{H}_B$ where $A' = A \setminus \{q\}$.
    The state vector is unchanged by the transfer; only the partition changes: before, Alice holds $A' \cup \{q\}$ with reduced state $\rho_{A'q} = \Tr_B(\ketbra{\Psi}{\Psi})$ and entropy $S(\rho_{A'q})$; after, Alice holds $A'$ with reduced state $\rho_{A'}$ and entropy $S(\rho_{A'})$.
    The triangle inequality for von Neumann entropy~\cite{arakiEntropyInequalities1970} gives
    \begin{equation}
        \abs{S(\rho_{A'}) - S(\rho_q)} \leq S(\rho_{A'q}) \leq S(\rho_{A'}) + S(\rho_q),
    \end{equation}
    so subadditivity and the Araki-Lieb inequality together yield $\abs{S(\rho_{A'q}) - S(\rho_{A'})} \leq S(\rho_q) \leq \log_2 2 = 1$.
    Hence Alice's entropy changes by at most $1$, and $S(\rho_A) \leq k + 1$ after $k+1$ communicated qubits.
\end{proof}

Then, a continuity bound relates simulation error to entanglement.
\begin{lemma}[Fannes-Audenaert inequality~\cite{audenaertSharpContinuityEstimate2007}]\label{lem:fannes}
    For states $\rho, \sigma$ on a $d$-dimensional system with trace distance $\eps < 1/e$,
    \begin{equation}
        \abs{S(\rho) - S(\sigma)} \leq \eps \log_2 d + H_2(\eps),
    \end{equation}
    where $H_2(\eps) = -\eps \log_2 \eps - (1-\eps) \log_2(1-\eps)$ is the binary entropy.
\end{lemma}

\begin{corollary}\label{cor:sim-entanglement}
    If a simulation achieves trace distance $\eps < 1/e$ on the global state of $n_A + n_B$ qubits, then
    $S(\rho_A^{\mathrm{sim}}) \geq S(\rho_A^{\mathrm{target}}) - \eps n_A - H_2(\eps)$.
\end{corollary}
\begin{proof}
    The trace distance condition on the global state implies $\norm{\rho_A^{\mathrm{sim}} - \rho_A^{\mathrm{target}}}_1/2 \leq \eps$ by contractivity of partial trace.
    Applying \cref{lem:fannes} with $d = 2^{n_A}$ gives $\abs{S(\rho_A^{\mathrm{sim}}) - S(\rho_A^{\mathrm{target}})} \leq \eps n_A + H_2(\eps)$, which is the claim.
\end{proof}

To construct a hard instance where the entanglement grows at least linearly in time,
we use the transverse-field Ising model following Schuch et al.~\cite{schuchEntropyGrowthHardness2008}.

\begin{theorem}[Linear growth of entanglement in Hamiltonian dynamics \cite{schuchEntropyGrowthHardness2008}]\label{thm:schuch-entropy}
    Consider a chain with an odd number $n$ of spins with periodic boundary conditions, initially in state $\ket{\psi(0)} = \ket{1\cdots 1}$, evolving under the transverse-field Ising Hamiltonian
    \begin{equation}\label{eq:schuch-ham}
        H_{\mathrm{Ising}} = -\frac{1}{2}\sum_j \qty(\sigma_j^x \sigma_{j+1}^x + \sigma_j^z),
    \end{equation}
    where $S_L(t):= -\Tr[\rho_L(t) \log_2 \rho_L(t)]$ and $\rho_L(t)$ is the reduced density matrix on the first $L$ sites.
    Then for $n \geq 20$, $L \geq 10$, $4 \leq t \leq eL/4$, and $t \leq n/5$, the entanglement entropy satisfies
    \begin{equation}
        S_L(t) \geq \frac{4}{3\pi}t - \frac{1}{2}\ln t - 1.
    \end{equation}
\end{theorem}

Place the Ising chain across the Alice-Bob boundary: a transverse-field Ising chain~\eqref{eq:schuch-ham} on $n$ sites (odd, $n \geq 21$) with periodic boundary, initial state $\ket{1\cdots 1}$, where Alice holds the first $L = \floor{n/2}$ sites and Bob the remaining $n - L$.

\begin{theorem}[Communication lower bound for local Hamiltonians via entanglement]\label{thm:no-fast-forwarding}
    For any sufficiently large odd integer $n$ and any fixed $\eps < 8/(15\pi)$, there exists an initial product state and a transverse-field Ising chain on $n$ sites with $2$ crossing terms such that, with evolution time $t = \floor{n/5}$, any unitary protocol simulating $e^{-\ii H_{\mathrm{Ising}} t}$ to trace distance $\leq \eps$ requires $\Omega(t)$ qubits of quantum communication.
\end{theorem}
\begin{proof}
Set $t = \floor{n/5}$ and $L = \floor{n/2} = (n-1)/2$, so $5t \leq n \leq 5t+4$ and $L \leq (5t+3)/2$.
The hypotheses of \cref{thm:schuch-entropy} are satisfied for $n \geq 21$ odd: $t \geq 4$ since $n \geq 21$, $L \geq 10$, $t \leq n/5$ by construction, and $t \leq eL/4$ since $eL/4 \geq e(n-1)/8 > n/5 \geq t$ (as $e/8 > 1/5$).
By \cref{thm:schuch-entropy}, the target state has entanglement $S_L(t) \geq \frac{4}{3\pi}t - \frac{1}{2}\ln(t) - 1 = \Omega(t)$, with Alice holding $n_A = L = (n-1)/2$ qubits.
By \cref{cor:sim-entanglement}, any simulation with trace distance $\eps < 1/e$ on the initial state $\ket{1\cdots 1}$ satisfies
$S(\rho_A^{\mathrm{sim}}) \geq S_L(t) - \eps L - H_2(\eps)$.
Substituting Schuch's bound and $L \leq (5t+3)/2$,
\begin{equation}
    S(\rho_A^{\mathrm{sim}}) \geq \qty(\frac{4}{3\pi} - \frac{5\eps}{2}) t - \frac{1}{2}\ln t - 1 - \frac{3\eps}{2} - H_2(\eps).
\end{equation}
The sub-leading terms are $o(t)$, so the bound is $\Omega(t)$ provided the leading coefficient is positive: $\tfrac{4}{3\pi} > \tfrac{5\eps}{2}$, i.e., $\eps < \tfrac{8}{15\pi} \approx 0.17$.
By \cref{lem:comm-capacity}, $Q$ qubits of communication generate at most $Q$ units of entanglement entropy across the cut from a product state, so $Q \geq S(\rho_A^{\mathrm{sim}}) = \Omega(t)$.


\end{proof}

\end{document}